\pdfoutput=1
\documentclass[11pt]{article}

\usepackage[margin=1in]{geometry}
\usepackage{amsmath,amssymb,amsthm,mathtools,mathrsfs}
\numberwithin{equation}{section}
\usepackage{stmaryrd}
\usepackage{silence}
\usepackage{microtype}
\usepackage{tikz}
\usepackage{tikz-cd}
\usepackage{longtable}
\usepackage[hidelinks]{hyperref}
\hypersetup{
  pdftitle={Quantum Query Complexity and Span Programs from Pre-Geometry},
  pdfauthor={Justin Roy Cox, Neil Epstein, Zhirui Hu, Michael Jarret, Thomas De Mastri}
}

\newtheorem{definition}{Definition}[section]
\newtheorem{proposition}[definition]{Proposition}
\newtheorem{theorem}[definition]{Theorem}
\newtheorem{lemma}[definition]{Lemma}
\newtheorem{corollary}[definition]{Corollary}
\theoremstyle{remark}
\newtheorem{remark}[definition]{Remark}

\newcommand{\cl}{\operatorname{cl}}

\newcommand{\Part}{\operatorname{Part}}

\title{Quantum Query Complexity and Span Programs from Pre-Geometry}
\author{
Justin Roy Cox$^{2}$, Neil Epstein$^{2}$, Zhirui Hu$^{1}$,\\
Michael Jarret$^{1,2,3,4}$, and Thomas De Mastri$^{1}$\\[0.75em]
\small $^{1}$Department of Computer Science, George Mason University\\
\small $^{2}$Department of Mathematical Sciences, George Mason University\\
\small $^{3}$Quantum Science and Engineering Center, George Mason University\\
\small $^{4}$Center for Social Complexity, George Mason University
}
\date{}

\begin{document}
\maketitle

\begin{abstract}
Span programs provide a general method for designing quantum query algorithms,
and for Boolean coordinate queries their optimum cost agrees with the general
adversary bound. We develop a matroidal organization of this optimization that
separates dependence data determined by the query problem from the dependence
structure chosen by a span program. The resulting \emph{source matroid} records
query dependence, target determination, and separation data, while a
\emph{program matroid} records dependencies among the target and response-labeled
input elements of a candidate program.

When a query problem admits a source matroid, we use its structure to organize
adversary upper and lower bounds. For separating binary query sets in this
setting, the global matroid-span-program optimum equals
$\mathrm{Adv}^{\pm}_{\mathcal Q}(f)$. For a fixed regular program matroid,
projective uniqueness reduces optimization over real representations to
positive element weights. We express positive and negative costs as weighted
basis ratios and derive exact $1$-, $2$-, and $3$-sum reductions, giving a
compositional evaluation of the fixed-weight objective along a Seymour
decomposition.

For pointed $R_{10}$ closure, exact certificates give
$3.9087553<\mathrm{Adv}^{\pm}_{\mathcal Q_R}(f_R)<3.9301$. The natural
$R_{10}$ program has optimized cost $5$, while another regular program has cost
at most $\sqrt{87/5}$. Recursive self-composition yields a family on $N$ input
bits with randomized query complexity $\Omega(N^{0.7324867\ldots})$, an explicit
quantum algorithm using $O(N^{0.6500178\ldots})$ queries, and
$Q(F)=\Theta(N^\alpha)$ for
$0.6204277\ldots<\alpha<0.6229063\ldots$.
\end{abstract}

\tableofcontents

\section{Introduction}
\label{sec:introduction}

Span programs were introduced by Karchmer and Wigderson as a
linear-algebraic model of computation~\cite{KarchmerWigderson1993}. They
later became a central tool in quantum query complexity: span programs give
quantum algorithms for formula evaluation~\cite{ReichardtSpalek2008}, and
for Boolean coordinate queries their optimal quadratic cost coincides with
the general adversary bound
~\cite{Reichardt2009,Reichardt2014Equivalent}. The general adversary bound
in turn characterizes bounded-error quantum query complexity up to constant
factors
~\cite{HoyerLeeSpalek2007,Reichardt2011,
LeeMittalReichardtSpalekSzegedy2011}.

Recent work has organized adversary-based quantum algorithmic frameworks by
their simulation and composition relations. Cornelissen shows that the
$s$--$t$ connectivity framework subsumes adaptive and extended learning
graphs and weighted decision trees, and introduces graph composition, in
which span programs determine the availability of graph edges
~\cite{Cornelissen2025GraphComposition}. Because an arbitrary span program
can be placed on a single $s$--$t$ edge, graph composition is itself
query-optimal. This paper addresses a different structural question. Rather
than compare algorithmic frameworks by simulation, we separate the
dependence structure chosen by a span program from its numerical
representation and from the dependence data already present in the query
problem.

Matroids already occur naturally in span programs. Karchmer and
Wigderson's incidence construction for $s$--$t$ connectivity has the cycle
dependencies of a graph~\cite{KarchmerWigderson1993}. More generally, the
target and input vectors of any represented span program determine a
matroid. We call this the \emph{program matroid}. Its target designation and
response labels remain part of the labeled program, while a real
representation restores the numerical span program.

The query problem carries a different matroidal object. A collection of
query answers can determine the answer to another query, giving the standard
closure relation of functional dependency~\cite{Armstrong1974,Matus1991}.
When this closure satisfies matroid exchange, it defines a query matroid. A
compatible single-element extension records which fixed query sets determine
the target function, with existence controlled by Crapo's modular-cut
criterion~\cite{Crapo1965,Oxley2011}. Retaining the pairwise separation
relation of each query gives the colored targeted matroid
$\mathfrak M_f=(M_f,\Sigma_f)$, which we call the \emph{source matroid}. Its
underlying matroid $M_f$ is the targeted query matroid. The source matroid is
determined by the query problem; the program matroid is a structural choice
in the design of a span program.

For opposite-output inputs $x$ and $y$, the queries on which they agree form
a noncomputing flat of the query matroid. Grouping pairs by these agreement
flats reorganizes the standard general-adversary optimization
~\cite{HoyerLeeSpalek2007,LeeMittalReichardtSpalekSzegedy2011}. This
reindexing yields structured adversary lower certificates,
correlation-transport upper certificates, and, when one output fiber is a
singleton, an exact fractional packing and covering problem. These
constructions organize the adversary filters through the source matroid; they
do not replace the adversary method.

On the program side, we fix an abstract response-labeled program and optimize
its real representations. For a regular program matroid, projective
uniqueness of real representations~\cite{BrylawskiLucas1976,Oxley2011}
reduces this freedom to one positive weight for each non-target element. The
reconstruction and separation energies become opposite target-inclusion
ratios of weighted basis sums on restriction and contraction minors.

Regularity also brings Seymour's decomposition theorem into the
optimization~\cite{Seymour1980}. We derive exact transformations of the
weighted target-basis ratio under $1$-, $2$-, and $3$-sums. A $1$-sum factor
cancels; a target-free component of a $2$-sum is replaced by one effective
element weight; and the weighted basis data of a $3$-sum component can be
replaced by a weighted graphic configuration. On input minors the $2$-sum
effective weight specializes to the inverse reconstruction energy for a
positive child or the separation energy for a negative child, with deletion
and contraction covering the degenerate cases. The fixed-weight objective
can therefore be evaluated recursively along a Seymour decomposition.

The exceptional regular summand $R_{10}$ remains finite. An intact copy has
no triangle through which it can participate in a $3$-sum, so its nontrivial
attachment is through a $2$-sum and contributes a ratio of two explicit
weighted basis polynomials. Moreover,
$R_{10}\setminus e\cong M(K_{3,3})$ and
$R_{10}/e\cong M^*(K_{3,3})$ for every element $e$
~\cite{Oxley2011}. Once deletion or contraction changes an $R_{10}$ summand,
the resulting calculation is graphic or cographic. Thus, apart from the
finite intact-$R_{10}$ calculation, Seymour decomposition transfers the
fixed-weight problem to graphic and cographic matroids.

The source/program distinction is compatible with the established optimality
theory. For Boolean coordinate queries, minimizing over all represented span
programs gives the general adversary bound
~\cite{Reichardt2009,Reichardt2014Equivalent}. Appendix
~\ref{app:universal-normal-form} reorganizes the standard
filtered-$\gamma_2$ construction in response-labeled matroidal form. For a
finite query set, if
\[
    \kappa(\mathcal Q):=
    \max_{q\in\mathcal Q}(|\mathcal B_q|-1),
\]
then
\[
    \operatorname{Adv}^{\pm}_{\mathcal Q}(f)
    \leq
    \operatorname{MSP}(f,\mathfrak M_f)
    \leq
    \sqrt{\kappa(\mathcal Q)}\,
    \operatorname{Adv}^{\pm}_{\mathcal Q}(f).
\]
Binary queries recover the established equality. Related non-binary
span-program models appear in~\cite{BeigiTaghavi2019}.

Grover search gives a basic example. Its targeted query matroid is
$U_{n,n+1}$, while an optimal regular program is realized by the rank-one
program matroid $U_{1,n+1}$; both calculations give $\sqrt n$.

For pointed $R_{10}$ closure, exact adversary certificates give
\[
    3.908755342845<
    \operatorname{Adv}^{\pm}_{\mathcal Q_R}(f_R)<3.9301.
\]
The targeted query matroid is $U_{9,10}$ even though the function is defined
by closure in $R_{10}$. Taking the defining matroid $R_{10}$ itself as the
program matroid gives optimized cost $5$, whereas a second explicit regular
program satisfies
$\operatorname{lincost}\leq\sqrt{87/5}=4.171330722\ldots$.
Thus the source matroid, the defining matroid, and an efficient program
matroid are three distinct objects in this example.

Let $F_h$ be the depth-$h$ self-composition of this nine-query predicate on
disjoint blocks. Exact adversary composition
~\cite{HoyerLeeSpalek2007,Reichardt2009} and adversary tightness
~\cite{Reichardt2011,Reichardt2014Equivalent} give
$Q_{\mathcal Q_h}(F_h)=\Theta(a^h)$ for
$a=\operatorname{Adv}^{\pm}_{\mathcal Q_R}(f_R)$, while sensitivity gives
$R^{1/3}_{\mathcal Q_h}(F_h)\geq\frac13 5^h$. The only semidefinite
optimization is the fixed nine-query calculation at depth one; exact
adversary composition and the regular-matroid sum identities propagate the
bounds to arbitrary depth. This fixed-base composition technique is standard
~\cite{ReichardtSpalek2008,Reichardt2009}; the application-specific content
is the certified pointed-$R_{10}$ calculation and the comparison among its
program constructions.

This separation is modest compared with known total Boolean examples:
Grover search already gives a quadratic separation, and the cheat-sheet
framework combined with stronger Forrelation bounds gives separations
approaching a third power
~\cite{AaronsonBenDavidKothari2016,SherstovStorozhenkoWu2023}. The recursive
family is therefore used here as a structural test case rather than as an
extremal separation result. It shows how a constant-size adversary
calculation and explicit regular-program calculations can be propagated
through a non-graphic, non-cographic family. We are not aware of a prior
quantum-query analysis of this pointed-$R_{10}$ self-composition family.

The framework leaves the discrete choice of program matroid open. Its role is
to separate that choice from representation optimization, to organize
adversary calculations by source-matroid structure, and to make
regular-program optimization compatible with the decomposition theory of
regular matroids.

\part{Background}

\section{Matroids}
\label{sec:matroids}

Matroids were introduced by Whitney as an abstraction of linear dependence
~\cite{Whitney1935}. A matroid records which subsets of a finite ground set
are independent. From these independent sets one obtains bases, rank,
closure, flats, and circuits.

We follow Oxley for notation and terminology; for broader treatments, see
White and Oxley~\cite{White1986,Oxley2011}.

\begin{definition}
\label{def:matroid}
A \emph{matroid} is a pair $M=(E,\mathcal I)$, where $E$ is a finite set, called the \emph{ground set}, and $\mathcal I\subseteq 2^E$ is a collection of subsets of $E$, called \emph{independent sets}, satisfying the following independent-set axioms.
\begin{enumerate}
    \item[(I1)] $\varnothing\in\mathcal I$.
    \item[(I2)] If $I\in\mathcal I$ and $J\subseteq I$, then $J\in\mathcal I$.
    \item[(I3)] If $J,K\in\mathcal I$ and $|J|<|K|$, then there exists $e\in K\setminus J$ such that $J\cup\{e\}\in\mathcal I$.
\end{enumerate}
\end{definition}

The first two axioms make independence hereditary. The third is the exchange
axiom: if one independent set is larger than another, some element of the
larger set can be added to the smaller one.

A subset of $E$ that does not belong to $\mathcal I$ is called dependent. Among the independent sets, the maximal ones are the bases.

\begin{definition}
\label{def:matroid-basis}
Let $M=(E,\mathcal I)$ be a matroid. A \emph{basis} of $M$ is a maximal independent subset of $E$. More generally, if $A\subseteq E$, a basis of $A$ is a maximal independent subset of $A$.
\end{definition}

The exchange axiom implies that all bases of a fixed subset $A\subseteq E$
have the same cardinality. This common cardinality defines
the rank of $A$.

\begin{definition}
\label{def:matroid-rank}
Let $M=(E,\mathcal I)$ be a matroid. The rank function $r_M:2^E\to\mathbb Z_{\geq0}$ is defined by $r_M(A):=\max\{\,|I|:I\subseteq A,\ I\in\mathcal I\,\}$. Equivalently, $r_M(A)$ is the cardinality of any basis of $A$.
\end{definition}

Rank tells us whether adding an element increases the size of a largest
independent subset. This gives the closure operator.

\begin{definition}
\label{def:matroid-closure}
Let $M=(E,\mathcal I)$ be a matroid. For $A\subseteq E$, the \emph{closure} of $A$ is $\cl_M(A):=\{\,e\in E:r_M(A\cup\{e\})=r_M(A)\,\}$.
\end{definition}

Thus $e\in\cl_M(A)$ exactly when adding $e$ to $A$ does not increase the rank. A set equal to its closure is called a flat.

\begin{definition}
\label{def:matroid-flat}
Let $M=(E,\mathcal I)$ be a matroid. A subset $F\subseteq E$ is a \emph{flat} if $\cl_M(F)=F$.
\end{definition}

Equivalently, $F$ is a flat if every $e\in E\setminus F$ increases rank when it is added to $F$. The only flat containing a basis of $M$ is the ground set $E$ itself. A flat is \emph{proper} when it is not equal to $E$.

\begin{definition}
\label{def:matroid-circuit}
Let $M=(E,\mathcal I)$ be a matroid. A \emph{circuit} of $M$ is a minimal dependent subset of $E$. Thus $C\subseteq E$ is a circuit when $C\notin\mathcal I$, but every proper subset $C'\subsetneq C$ belongs to $\mathcal I$.
\end{definition}

\begin{proposition}
\label{prop:closure-circuit}
Let $M=(E,\mathcal I)$ be a matroid, let $A\subseteq E$, and let $e\notin A$. Then $e\in\cl_M(A)$ if and only if there is a circuit $C$ with $e\in C$ and $C\setminus\{e\}\subseteq A$.
\end{proposition}

The closure statement says that $e$ adds no new rank. The circuit statement exhibits a minimal reason why it adds no new rank.

\begin{definition}
\label{def:hyperplane-cocircuit}
A \emph{hyperplane} of a matroid $M$ on $E$ is a maximal proper flat. A \emph{cocircuit} is the complement $E\setminus H$ of a hyperplane $H$.
\end{definition}

Circuits and cocircuits give complementary descriptions of local matroid structure. A circuit is a minimal dependent set, while a cocircuit is a minimal set which meets every basis.

\subsection{Constructions}
\label{subsec:elementary-matroid-constructions}

We will use the following matroid constructions and certificates: the free
matroid, loops and coloops, minors, fundamental circuits, and the cocircuit
separation property.

\begin{definition}
\label{def:free-matroid}
The \emph{free matroid} on a finite set $E$ is the matroid in which every
subset of $E$ is independent. Equivalently, its unique basis is $E$ and
$\cl_M(A)=A$ for every $A\subseteq E$.
\end{definition}

\begin{definition}
\label{def:loop-coloop}
An element $e$ of a matroid is a \emph{loop} if
$e\in\cl_M(\varnothing)$, equivalently if $\{e\}$ is a circuit. It is a
\emph{coloop} if it belongs to every basis.
\end{definition}

\begin{definition}
\label{def:matroid-minors}
Let $M$ be a matroid on $E$ and let $A\subseteq E$.
The \emph{restriction} $M|A$ is the matroid on $A$ whose independent sets
are the independent subsets of $A$ in $M$. The \emph{deletion} of $A$ is
$M\setminus A:=M|(E\setminus A)$.

The \emph{contraction} $M/A$ is the matroid on $E\setminus A$ with rank
function, for $X\subseteq E\setminus A$:
\begin{equation}
    r_{M/A}(X)
    :=
    r_M(X\cup A)-r_M(A).
\end{equation}
A \emph{minor} of $M$ is any matroid obtained by a sequence of deletions and
contractions.
\end{definition}

\begin{definition}
\label{def:fundamental-circuit}
Let $B$ be a basis of $M$ and let $e\notin B$. The unique circuit contained in $B\cup\{e\}$ is the \emph{fundamental circuit of $e$ with respect to $B$}.
\end{definition}

The following standard separation fact is the cocircuit analogue of
Proposition~\ref{prop:closure-circuit}.

\begin{proposition}
\label{prop:cocircuit-separates-closure}
Let $M$ be a matroid on $E$, let $A\subseteq E$, and let
$e\notin\cl_M(A)$. Then there is a cocircuit $D$ such that
\begin{equation}
    e\in D,
    \qquad
    D\cap A=\varnothing.
\end{equation}
\end{proposition}

\begin{proof}
Choose a basis $B_A$ of $A$. Since $e\notin\cl_M(A)$,
$B_A\cup\{e\}$ is independent, so extend it to a basis $B$ of $M$.
Then
\begin{equation}
    H:=\cl_M(B\setminus\{e\})
\end{equation}
is a hyperplane. It contains $A$ and excludes $e$. Hence
$D:=E\setminus H$ is a cocircuit containing $e$ and disjoint from $A$.
\end{proof}

\subsection{Modular cuts and single-element extensions}
\label{subsec:modular-cuts}

A single-element extension adds one new element to an existing matroid. Suppose $N$ is a matroid on $E$ and we want a matroid on $E\cup\{\tau\}$ whose restriction to $E$ is $N$. For each flat $F$ of $N$, the extension must decide whether $\tau$ lies in $\cl(F)$. These choices must be compatible with the matroid axioms; modular cuts record the compatible choices.

\begin{definition}
\label{def:modular-pair}
Let $N$ be a matroid on a ground set $E$. Two flats $F$ and $G$ form a \emph{modular pair} if $r_N(F)+r_N(G)=r_N(F\cap G)+r_N(\cl_N(F\cup G))$~\cite{Crapo1965}.
\end{definition}

\begin{definition}
\label{def:modular-cut}
Let $N$ be a matroid on a ground set $E$. A \emph{modular cut} of $N$ is a collection $\mathcal C$ of flats of $N$ satisfying the following two conditions~\cite{Crapo1965}.
\begin{enumerate}
    \item If $F\in\mathcal C$, $F\subseteq G$, and $G$ is a flat of $N$, then $G\in\mathcal C$.
    \item If $F,G\in\mathcal C$ form a modular pair, then $F\cap G\in\mathcal C$.
\end{enumerate}
\end{definition}

In the extension interpretation, the first requirement says that if $\tau$ lies in the closure of a flat $F$, then it also lies in the closure of every larger flat. The second requirement is the additional compatibility needed for these closure choices to come from a matroid extension. Crapo showed that single-element extensions are classified by modular cuts~\cite{Crapo1965}; we use Oxley for the modern formulation~\cite{Oxley2011}.

\begin{definition}
\label{def:single-element-extension}
Let $N$ be a matroid on $E$, and let $\tau\notin E$. A \emph{single-element extension} of $N$ by $\tau$ is a matroid $M$ on $E\cup\{\tau\}$ whose restriction to $E$ is $N$.
\end{definition}

\begin{proposition}
\label{prop:modular-cut-extension}
Let $N$ be a matroid on $E$. Single-element extensions of $N$ by an element $\tau$ are in bijection with modular cuts of $N$. For an extension $M$, the corresponding modular cut is $\mathcal C_M:=\{\,F:F\text{ is a flat of }N\text{ and }\tau\in\cl_M(F)\,\}$. Conversely, every modular cut of $N$ determines a unique single-element extension in this way~\cite{Crapo1965}.
\end{proposition}

This correspondence includes the two extreme cases. The empty modular cut adjoins $\tau$ as a coloop, while the collection of all flats adjoins $\tau$ as a loop.

Proposition~\ref{prop:target-extension-criterion} is an immediate consequence
of Crapo's single-element extension theorem~\cite{Crapo1965}. It gives the
criterion for target extensions in which positive inputs span the added element and
negative inputs do not.

\begin{proposition}
\label{prop:target-extension-criterion}
Let $N$ be a matroid on $E$, let $f:D\to\{0,1\}$, and let $A:D\to 2^E$. Put $F:=\cl_N\circ A$, and define $\mathcal F_+:=F(f^{-1}(1))$ and $\mathcal F_-:=F(f^{-1}(0))$. Let $\mathcal C_+$ be the smallest modular cut of $N$ containing $\mathcal F_+$.

There is a single-element extension $M$ of $N$ by $\tau$ such that $f(x)=1$ if and only if $\tau\in\cl_M(A(x))$ for every $x\in D$ if and only if $\mathcal C_+\cap\mathcal F_-=\varnothing$. Such extensions correspond to the modular cuts $\mathcal C$ satisfying $\mathcal C_+\subseteq\mathcal C$ and $\mathcal C\cap\mathcal F_-=\varnothing$.
\end{proposition}

\begin{proof}
The modular cut $\mathcal C_+$ exists because the collection of all flats of $N$ is a modular cut, and intersections of modular cuts are modular cuts.

Let $M$ be a single-element extension of $N$ by $\tau$, with associated modular cut $\mathcal C_M$. Since $F(x)=\cl_N(A(x))$ and $M$ restricts to $N$ on $E$, the sets $A(x)$ and $F(x)$ have the same rank in $M$. Hence $\tau\in\cl_M(A(x))$ if and only if $\tau\in\cl_M(F(x))$, which by Proposition~\ref{prop:modular-cut-extension} holds exactly when $F(x)\in\mathcal C_M$.

Thus $M$ has the required closure rule exactly when $\mathcal C_M$ contains $\mathcal F_+$ and is disjoint from $\mathcal F_-$. The first condition is equivalent to $\mathcal C_+\subseteq\mathcal C_M$. Therefore such an extension exists exactly when $\mathcal C_+\cap\mathcal F_-=\varnothing$, and the final characterization follows.
\end{proof}

The closure data decide whether a target can be adjoined; representability is a further condition on the chosen extension.

\subsection{Representable matroids}
\label{subsec:representable-matroids}

A representation realizes an abstract matroid by a vector-space configuration
whose linear independence relations reproduce the matroid. Span programs use
these vectors, while the underlying matroid records only their dependence
structure.

\begin{definition}
\label{def:matroid-representability}
Let $\mathbb K$ be a field. A matroid $M=(E,\mathcal I)$ is \emph{representable over $\mathbb K$} if there are a vector space $V$ over $\mathbb K$ and a map $\phi:E\to V$ such that, for every $I\subseteq E$, the set $I$ is independent in $M$ if and only if the indexed family $(\phi(e))_{e\in I}$ is linearly independent over $\mathbb K$. The map $\phi$ is called a \emph{representation} of $M$ over $\mathbb K$.
\end{definition}

\begin{definition}
\label{def:vector-matroid}
Let $E$ be a finite set, let $V$ be a vector space over a field $\mathbb K$, and let $\phi:E\to V$. The \emph{vector matroid of $\phi$} is the matroid $M(\phi)$ on $E$ whose independent sets are exactly those subsets $I\subseteq E$ for which the indexed family $(\phi(e))_{e\in I}$ is linearly independent over $\mathbb K$.
\end{definition}

Thus $\phi$ represents a matroid $M$ if and only if $M(\phi)=M$.

A representation turns the abstract dependence relation into linear
dependence. In particular,
$e\in\cl_M(A)$ if and only if
$\phi(e)\in\operatorname{span}\{\phi(a):a\in A\}$.
Also, $r_M(A)=\dim\operatorname{span}\phi(A)$, and the circuits of $M$ are
the minimally linearly dependent indexed sets of represented vectors. The closure $\cl_M(A)$ remains a subset of the ground set
$E$; the representation supplies a linear test for membership.

Different representations of the same matroid need not be related by a change of coordinates. A representation is a linear realization of the matroid's independence structure, not data intrinsic to the abstract matroid. Rank, closure, circuits, and bases do not depend on a chosen representation.

\begin{definition}
\label{def:projective-equivalence}
Let $\phi:E\to V$ and $\psi:E\to W$ be representations of the same matroid over the same field. They are \emph{projectively equivalent} if there is a linear isomorphism $T:\operatorname{span}\phi(E)\to\operatorname{span}\psi(E)$ and, for each nonloop $e\in E$, a nonzero scalar $c_e$ such that $\psi(e)=c_eT\phi(e)$~\cite{BrylawskiLucas1976,Oxley2011}.
\end{definition}

The next proposition identifies matroid minors with the corresponding
operations on represented vector families~\cite{Oxley2011}.

\begin{proposition}
\label{prop:represented-minors}
Let $\phi:E\to V$ represent a matroid $M$.
For $A\subseteq E$, restriction and deletion are represented by retaining the corresponding vectors. Contraction $M/A$ is represented by the images of the remaining vectors in $V/\operatorname{span}\{\phi(a):a\in A\}$.
\end{proposition}

\begin{proof}
The statements for restriction and deletion are immediate. Put
$U:=\operatorname{span}\{\phi(a):a\in A\}$. For
$X\subseteq E\setminus A$, the span of the images of $X$ in $V/U$ has
dimension
\begin{equation}
    \dim\bigl(U+\operatorname{span}\phi(X)\bigr)-\dim U
    =
    r_M(X\cup A)-r_M(A),
\end{equation}
which is exactly $r_{M/A}(X)$.
\end{proof}

Linear functionals produce flats in represented matroids.

\begin{proposition}
\label{prop:functional-zero-flat}
Let $\phi:E\to V$ represent $M$. For $\ell\in V^*$, the set $F_\ell:=\{e\in E:\ell(\phi(e))=0\}$ is a flat of $M$.
\end{proposition}

\begin{proof}
If $e\in\cl_M(F_\ell)$, then $\phi(e)$ lies in the span of vectors on which
$\ell$ vanishes. Hence $\ell(\phi(e))=0$, so $e\in F_\ell$.
\end{proof}

\subsection{Regular matroids}
\label{subsec:regular-matroids}

Regular matroids admit totally unimodular real representations, and their
real representations are projectively equivalent. We use Tutte's totally
unimodular characterization. A matrix is \emph{totally unimodular} if every
square minor has determinant in $\{0,1,-1\}$.

\begin{definition}
\label{def:regular-matroid}
A matroid is \emph{regular} if it is representable over every field
~\cite{Oxley2011}.
\end{definition}

Tutte's characterization states that $M$ is regular if and only if it has a
real representation $\phi:E\to V$ for which, after choosing an ordering
$E=\{e_1,\ldots,e_n\}$ and a basis $\mathcal B$ of
$\operatorname{span}\phi(E)$, the coordinate matrix
\begin{equation}
    \begin{bmatrix}
        [\phi(e_1)]_{\mathcal B} & \cdots & [\phi(e_n)]_{\mathcal B}
    \end{bmatrix}
\end{equation}
is totally unimodular~\cite{Tutte1958,Oxley2011}.

\begin{remark}
As an example, let $E=\{e_1,e_2,e_3\}$ and
$\mathcal I=\{J\subseteq E:|J|\leq 2\}$. The matroid
$M=(E,\mathcal I)$ is regular. Define
$\phi:E\to\mathbb R^2$ by
\begin{equation}
    \phi(e_1)=
    \begin{bmatrix}1\\0\end{bmatrix},
    \qquad
    \phi(e_2)=
    \begin{bmatrix}0\\1\end{bmatrix},
    \qquad
    \phi(e_3)=
    \begin{bmatrix}1\\1\end{bmatrix}.
\end{equation}
With respect to the standard basis of $\mathbb R^2$, the corresponding
coordinate matrix is
\begin{equation}
    A=
    \begin{bmatrix}
        1 & 0 & 1\\
        0 & 1 & 1
    \end{bmatrix}.
\end{equation}
Its columns realize exactly the independent sets in $\mathcal I$, so $A$
represents $M$. The matrix $A$ is totally unimodular. This also certifies
regularity: more generally, a set of $k$ columns of a totally unimodular
matrix is linearly independent if and only if some $k\times k$ minor on
those columns is nonzero. Such a minor has determinant $1$ or $-1$, and
remains nonzero over every field. The vanishing and nonvanishing
minors, and hence the independent column sets, are consequently unchanged
when the matrix is regarded over another field. A totally unimodular
coordinate matrix thus represents the same matroid over every field.
\end{remark}

\begin{proposition}
\label{prop:regular-basic-properties}
Regular matroids are closed under deletion and contraction
~\cite{Oxley2011}. Moreover, all real representations of a regular matroid
are projectively equivalent~\cite{BrylawskiLucas1976,Oxley2011}. In
particular, every real representation is projectively equivalent to a fixed
totally unimodular representation.
\end{proposition}

Higgs showed that every matroid quotient can be realized by an extension
followed by a contraction~\cite{Higgs1968}. Since every matroid $M$ on $E$
is a quotient of the free matroid on $E$, Higgs's theorem gives a matroid
$\widehat M$ and a set $R$ with
$\widehat M\setminus R$ free on $E$ and
$\widehat M/R\cong M$. Theorem~\ref{thm:regular-quotient-form} shows that
when $M$ is regular, $\widehat M$ can also be chosen regular.

\begin{theorem}
\label{thm:regular-quotient-form}
Let $M$ be a finite matroid on a ground set $E$. Then $M$ is regular if and only if there exist a finite set $R$ disjoint from $E$ and a regular matroid $\widehat M$ on $E\sqcup R$ such that
\(
    \widehat M\setminus R
\)
is the free matroid on $E$ and
\(
    M\cong \widehat M/R,
\)
where the isomorphism is the identity on $E$.
\end{theorem}

\begin{proof}
If such a matroid $\widehat M$ exists, then $M$ is regular by
Proposition~\ref{prop:regular-basic-properties}.

Conversely, suppose $M$ is regular, has $n$ elements, and has rank $r$.
Choose a totally unimodular representation of $M$. After permuting columns
and pivoting on a basis, we may keep the representation totally unimodular
and write it in standard form
\begin{equation}
    A=
    \begin{pmatrix}
        I_r & D
    \end{pmatrix}.
\end{equation}
The matrix $D$ is totally unimodular. Put
\begin{equation}
    K=
    \begin{pmatrix}
        -D\\
        I_{n-r}
    \end{pmatrix}.
\end{equation}
Its columns form a basis of $\ker A$. Appending identity rows or columns
preserves total unimodularity, so both $K$ and
\begin{equation}
    \widehat A=
    \begin{pmatrix}
        I_n & K
    \end{pmatrix}
\end{equation}
are totally unimodular.

Let $\widehat M$ be the vector matroid of $\widehat A$. Index the first $n$
columns by $E$ and the remaining columns by a set $R$. Deleting $R$ leaves
$I_n$, so $\widehat M\setminus R$ is free on $E$.

By Proposition~\ref{prop:represented-minors}, contracting $R$ is represented
by quotienting $\mathbb R^n$ by the span of the columns of $K$, which is
$\ker A$. The map
\begin{equation}
    \mathbb R^n/\ker A\longrightarrow\mathbb R^r,
    \qquad
    [v]\longmapsto Av
\end{equation}
is an isomorphism and sends the image of the $e$th identity column to the
$e$th column of $A$. Thus the remaining elements represent $M$.
\end{proof}

\begin{corollary}
\label{cor:regular-extension-chain}
Every regular matroid can be obtained from a free matroid by a sequence of
single-element extensions, each of which is regular, followed by contraction
of the added elements.
\end{corollary}

\begin{proof}
Use Theorem~\ref{thm:regular-quotient-form} and order
$R=\{\rho_1,\ldots,\rho_k\}$. For $0\leq i\leq k$, let
\begin{equation}
    M_i
    :=
    \widehat M\setminus
    \{\rho_{i+1},\ldots,\rho_k\}.
\end{equation}
Then $M_0=\widehat M\setminus R$ is free, and $M_i$ is obtained from
$M_{i-1}$ by adjoining $\rho_i$. Every $M_i$ is a deletion minor of the
regular matroid $\widehat M$, hence is regular by
Proposition~\ref{prop:regular-basic-properties}. After all extensions have
been made, contracting $R$ gives $M$.
\end{proof}

\subsection{Sums and decomposition of regular matroids}
\label{subsec:regular-matroid-sums}

Regular matroids admit a recursive decomposition into graphic matroids,
cographic matroids, and one exceptional ten-element matroid. The decomposition
uses three operations that join regular matroids along no common element, one
common element, or a common triangle.

\begin{definition}[Graphic and cographic matroids]
\label{def:graphic-cographic}
Let $G$ be a finite graph. The \emph{graphic matroid} $M(G)$ has ground set
$E(G)$, with a set of edges independent exactly when it contains no graph
cycle. The \emph{cographic matroid} $M^*(G)$ has the same ground set, with
circuits given by the minimal nonempty edge cuts of $G$~\cite{Oxley2011}.
\end{definition}

\begin{definition}
\label{def:matroid-cycle}
A \emph{cycle} of a regular matroid is a disjoint union of circuits, with the
empty set allowed.
\end{definition}

A regular matroid is representable over $\mathbb F_2$, so the symmetric
difference of two cycles is again a cycle~\cite{Seymour1980,Oxley2011}. For
sets $A$ and $B$, write
$A\triangle B:=(A\setminus B)\cup(B\setminus A)$.

\begin{definition}[Regular matroid sums]
\label{def:regular-matroid-sums}
Let $M_1$ and $M_2$ be regular matroids on ground sets $E_1$ and $E_2$.
For $k\in\{1,2,3\}$, the \emph{$k$-sum} $M_1\oplus_k M_2$ is the matroid on
$E_1\triangle E_2$ whose cycles are exactly the sets
$C_1\triangle C_2\subseteq E_1\triangle E_2$, where $C_i$ is a cycle of
$M_i$, under the following hypotheses~\cite{Seymour1980,Oxley2011}:
\begin{enumerate}
    \item for a $1$-sum, $E_1\cap E_2=\varnothing$ and both ground sets are
    nonempty;
    \item for a $2$-sum, $E_1\cap E_2=\{e\}$, the common element $e$ is
    neither a loop nor a coloop of either matroid, and $|E_i|\geq3$;
    \item for a $3$-sum, $E_1\cap E_2=T$ has three elements, $T$ is a
    circuit of both matroids, $T$ contains no cocircuit of either matroid,
    and $|E_i|\geq7$.
\end{enumerate}
\end{definition}

A $1$-sum is the ordinary direct sum. In a $2$-sum, the common element $e$
is used to join the two matroids and is absent from the resulting ground
set. In a $3$-sum, the common triangle $T$ plays the same role and all three
of its elements are absent from the result. Each operation preserves
regularity~\cite{Seymour1980,Oxley2011}.

\begin{definition}[The matroid \texorpdfstring{$R_{10}$}{R10}]
\label{def:r10}
The matroid $R_{10}$ is the rank-five regular matroid on
$E_{10}:=\{0,1,\ldots,9\}$ represented by the matrix in
Figure~\ref{fig:r10-k33}~\cite{Seymour1980,Oxley2011}.
\end{definition}

\begin{figure}[t]
\centering
\begin{minipage}[c]{0.60\textwidth}
\centering
\begin{equation}
    R_{10}
    =
    M\!\left(
    \left[
    \begin{array}{ccccc|ccccc}
        1&0&0&0&0&-1& 1& 0& 0& 1\\
        0&1&0&0&0& 1&-1& 1& 0& 0\\
        0&0&1&0&0& 0& 1&-1& 1& 0\\
        0&0&0&1&0& 0& 0& 1&-1& 1\\
        0&0&0&0&1& 1& 0& 0& 1&-1
    \end{array}
    \right]
    \right).
\end{equation}
\small
The columns are indexed by $0,1,\ldots,9$ from left to right.
\end{minipage}
\hfill
\begin{minipage}[c]{0.32\textwidth}
\centering
\begin{tikzpicture}[scale=0.9]
    \coordinate (u1) at (0,2);
    \coordinate (u2) at (0,1);
    \coordinate (u3) at (0,0);
    \coordinate (v1) at (2.4,2);
    \coordinate (v2) at (2.4,1);
    \coordinate (v3) at (2.4,0);
    \foreach \u in {u1,u2,u3}
        \foreach \v in {v1,v2,v3}
            \draw (\u) -- (\v);
    \foreach \p in {u1,u2,u3,v1,v2,v3}
        \fill (\p) circle (2pt);
    \node[left]  at (u1) {$u_1$};
    \node[left]  at (u2) {$u_2$};
    \node[left]  at (u3) {$u_3$};
    \node[right] at (v1) {$v_1$};
    \node[right] at (v2) {$v_2$};
    \node[right] at (v3) {$v_3$};
    \node at (1.2,-0.55) {$K_{3,3}$};
\end{tikzpicture}

\small
Deleting element $0$ gives a copy of $M(K_{3,3})$.
\end{minipage}
\caption{A totally unimodular representation of $R_{10}$ and a graphic
representation of $R_{10}\setminus0$.}
\label{fig:r10-k33}
\end{figure}
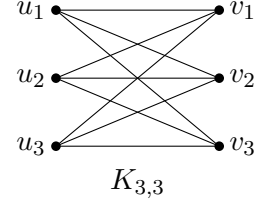

\begin{definition}
\label{def:seymour-decomposition}
A \emph{$\{1,2,3\}$-decomposition} of a regular matroid $M$ is a rooted
binary tree whose nodes are labeled by regular matroids. The root is labeled
by $M$; every internal node is the $1$-, $2$-, or $3$-sum of its two
children; and every leaf is graphic, cographic, or isomorphic to $R_{10}$.
\end{definition}

\begin{theorem}
\label{thm:seymour-decomposition}
Every regular matroid admits a $\{1,2,3\}$-decomposition
~\cite{Seymour1980,Oxley2011}.
\end{theorem}

Thus every regular matroid has a recursive description in which an internal
node combines two regular matroids across zero, one, or three shared elements,
and every leaf is graphic, cographic, or isomorphic to $R_{10}$. The shared
elements used by a $2$- or $3$-sum are absent from the ground set of the
resulting matroid.

\subsection{Colored matroids}
\label{subsec:colored-matroids}

We use the term \emph{colored matroid} in the element-labeling sense used
by Zaslavsky~\cite{Zaslavsky1992}. The colors are auxiliary labels; they do
not change the matroid.

\begin{definition}
\label{def:colored-matroid}
Let $M$ be a matroid on a finite ground set $E$, and let $\mathcal C$ be any
set. A \emph{$\mathcal C$-colored matroid} is a pair $(M,\kappa)$, where
$\kappa:E\to\mathcal C$ assigns a color to each element of the ground set.
\end{definition}

The underlying matroid remains $M$. In particular, its independent sets,
rank function, closure operator, flats, and circuits are unchanged by the
coloring. The map $\kappa$ records additional information attached to the
elements.

Here ``color'' means only a label. It may be a number, sign, set, relation,
or any other object used by the application.

\section{Queries}
\label{sec:query-models}
\label{sec:queries}

A deterministic query partitions the input domain into answer classes: two
inputs lie in the same block exactly when the query returns the same answer.
This is the standard decision-tree view
~\cite{BuhrmanDeWolf2002,KothariRacicotDeslogesSantha2015}. We use this
partition, rather than a particular encoding of the answers, as the query's
information content. Several fixed queries together induce the common
refinement of their answer partitions. Throughout, $D$ is a finite,
nonempty domain and all query gate sets are finite.

\subsection{Partitions}
\label{subsec:partitions}

We use the standard set-theoretic notion of a partition and the usual
refinement order.

\begin{definition}
\label{def:partition}
A \emph{partition} of $D$ is a collection $\pi\subseteq 2^D$ such that every
block $B\in\pi$ is nonempty, distinct blocks are disjoint, and
$\bigcup_{B\in\pi}B=D$.  We write $\Part(D)$ for the set of all partitions
of $D$.
\end{definition}

A block records a set of inputs which have not yet been distinguished.
Learning more information may split a block into smaller blocks.

\begin{definition}
\label{def:partition-refinement}
For $\pi,\sigma\in\Part(D)$, write $\pi\leq\sigma$ when every block of
$\pi$ is contained in a block of $\sigma$.  We say that $\pi$ \emph{refines}
$\sigma$.
\end{definition}

The finer partition is the smaller element of the order.  Write
$\top_D:=\{D\}$ for the indiscrete partition and
$\bot_D:=\{\{x\}:x\in D\}$ for the discrete partition.

\begin{definition}
\label{def:fiber-partition}
For $f:D\to Y$, the \emph{fiber partition} of $f$ is $\Pi_f:=\{f^{-1}(y):y\in Y,\ f^{-1}(y)\neq\varnothing\}$.
\end{definition}

Thus two inputs lie in the same block of $\Pi_f$ if and only if they have
the same target value. The partition records the fibers but not their labels.
If a computation produces a partition $\pi\leq\Pi_f$, then the block of
$\pi$ containing an input $x$ determines $f(x)$.

\begin{proposition}
\label{prop:partition-computation}
Let $f:D\to Y$ and $\pi\in\Part(D)$, and let
$q_\pi:D\to\pi$ be the map sending each $x\in D$ to its block $B_x$.
Then $\pi\leq\Pi_f$ if and only if there is a map
$\lambda_\pi:\pi\to Y$ making the diagram
\[
\begin{tikzcd}
D \arrow[rr,"f"] \arrow[dr,"q_\pi"'] && Y \\
& \pi \arrow[ur,"\lambda_\pi"'] &
\end{tikzcd}
\]
commute.
\end{proposition}

The map $\lambda_\pi$ is the decoder: once the information block
$q_\pi(x)$ is known, $\lambda_\pi$ returns the corresponding value $f(x)$.

\begin{proof}
If $\pi\leq\Pi_f$, then $f$ is constant on every block of $\pi$, so
$\lambda_\pi(B):=f(x)$ for any $x\in B$ is well-defined and satisfies
$f=\lambda_\pi\circ q_\pi$.  Conversely, if the diagram commutes, then
$f$ is constant on every block of $\pi$, so every block lies in one fiber
of $f$.  Hence $\pi\leq\Pi_f$.
\end{proof}

\subsection{Circuits and computation}
\label{subsec:circuits-and-computation}

We first separate circuit syntax from the information produced by an
implementation. Write $\mathcal G^*$ for the set of finite words over
$\mathcal G$, including the empty word.

\begin{definition}
\label{def:circuit-word}
Let $\mathcal G$ be a finite set of symbols, called elementary gates. A
\emph{circuit} over $\mathcal G$ is a word
$C=g_Tg_{T-1}\cdots g_1\in\mathcal G^*$.
Its length is $|C|:=T$.  The empty word is denoted $\varepsilon$.
\end{definition}

For the deterministic information model used in this section, an implementation
is specified by the partition of $D$ produced by each circuit word.

\begin{definition}
\label{def:circuit-implementation}
An \emph{implementation} of circuits over $\mathcal G$ on $D$ is a map $\llbracket-\rrbracket_D:\mathcal G^*\to\Part(D)$. We omit the subscript when the domain is clear.
\end{definition}

\begin{definition}
\label{def:circuit-computes}
A circuit $C$ \emph{computes} $f:D\to Y$ when
$\llbracket C\rrbracket\leq\Pi_f$.
\end{definition}

Definition~\ref{def:circuit-computes} is the usual deterministic
decision-tree correctness condition, written in partition notation
~\cite{BuhrmanDeWolf2002,KothariRacicotDeslogesSantha2015}.
Proposition~\ref{prop:partition-computation} gives the corresponding output
map. The circuit may distinguish inputs that have the same value of $f$.

\subsection{Query circuits}
\label{subsec:query-circuits}

In a deterministic decision tree, a query branches according to its answer
~\cite{BuhrmanDeWolf2002}. Its answer classes form a partition of
the input domain. We take that partition as the query gate data: the labels
used to encode the answers are discarded, while the pairs of inputs that the
query distinguishes are retained.

\begin{definition}
\label{def:query-gate}
A \emph{query gate} over $D$ is a formal symbol $q$ together with a
partition $\llbracket q\rrbracket\in\Part(D)$.  
\end{definition}

For a finite query set $\mathcal Q$, write
$\mathcal B_q:=\llbracket q\rrbracket$ and let $B_q(x)$ be the block of
$\mathcal B_q$ containing $x$. The global \emph{response-label set} is
$\mathsf R_{\mathcal Q}:=\{(q,B):q\in\mathcal Q,\ B\in\mathcal B_q\}$. We write $\mathsf R$ when $\mathcal Q$ is fixed.

The answers to two queries determine the common refinement of their two
answer partitions.

\begin{definition}
\label{def:common-refinement}
For $\pi,\sigma\in\Part(D)$, define
\begin{equation}
    \pi\wedge\sigma
    :=
    \{B\cap C:B\in\pi,\ C\in\sigma,\ B\cap C\neq\varnothing\}.
\end{equation}
For a query word $Q=q_T\cdots q_1\in\mathcal Q^*$, set
\begin{equation}
    \llbracket Q\rrbracket_D
    :=
    \llbracket q_T\rrbracket\wedge\cdots\wedge\llbracket q_1\rrbracket,
\end{equation}
and assign the empty word $\top_D$.
\end{definition}

The operation $\wedge$ is associative, commutative, and idempotent. Thus a
fixed query word records the information from a nonadaptive set of queries:
repetition and order do not change the induced partition.

\begin{definition}
\label{def:complete-query-set}
A finite query set $\mathcal Q$ is \emph{complete over $D$} if
\begin{equation}
    \bigwedge_{q\in\mathcal Q}\llbracket q\rrbracket=\bot_D.
\end{equation}
\end{definition}

A complete query set separates every pair of distinct inputs, so it computes
every function on $D$ by refinement.

Completeness is deliberately stronger than computation of one target.  A
complete set may reveal the whole input when the target requires much less.

\subsection{Separation}
\label{subsec:query-separation}

For a query partition, the separation relation records which input pairs lie
in different blocks.

\begin{definition}
\label{def:query-separation}
For a query $q$ over $D$, define $\operatorname{Sep}(q):=\{(x,y)\in D\times D:B_q(x)\neq B_q(y)\}$. For a function $f:D\to Y$, define $\operatorname{Sep}(f):=\{(x,y)\in D\times D:f(x)\neq f(y)\}$.
\end{definition}

Thus $\operatorname{Sep}(q)$ records the pairs distinguished by the answer to
$q$, while $\operatorname{Sep}(f)$ records the pairs that must eventually be
distinguished in order to compute $f$.

\begin{definition}
\label{def:query-set-separates-function}
For $S\subseteq\mathcal Q$, we say that $S$ \emph{separates $f$} when
\begin{equation}
    \operatorname{Sep}(f)
    \subseteq
    \bigcup_{q\in S}\operatorname{Sep}(q).
\end{equation}
\end{definition}

From $\operatorname{Sep}(q)$ we obtain the entrywise query mask used in
Section~\ref{sec:general-adversary-bound}. For coordinate queries, this is
the standard difference mask of the general adversary method
~\cite{HoyerLeeSpalek2007,LeeMittalReichardtSpalekSzegedy2011}. We use the
same construction for an arbitrary query partition.

Given a matrix $A=(A(x,y))_{x,y\in D}$ whose rows and columns are labeled by inputs $x,y\in D$, define
$\Delta_qA$ by
\begin{equation}
    (\Delta_qA)(x,y)
    :=
    \begin{cases}
        A(x,y), & (x,y)\in\operatorname{Sep}(q),\\
        0, & (x,y)\notin\operatorname{Sep}(q).
    \end{cases}
\end{equation}
Define $\Delta_fA$ in the same way using $\operatorname{Sep}(f)$.  When no
confusion can arise, we also write $\Delta_q$ and $\Delta_f$ for these
entrywise maps.

\subsection{Queries on product domains}
\label{subsec:product-query-lifts}

A query on one factor of a product domain extends to the product without
inspecting the other factor.

Let $q$ be a query over $X$, let $Z$ be finite, and let
$D\subseteq X\times Z$.  Its \emph{lift to $D$} is the query whose partition
is
\begin{equation}
    \llbracket q\rrbracket_D
    :=
    \{(B\times Z)\cap D:
        B\in\llbracket q\rrbracket,\ (B\times Z)\cap D\neq\varnothing\}.
\end{equation}
A query over $Z$ is lifted in the same way, using $X\times C$ for its
blocks.  When the domain makes the lift clear, we use the original query
symbol.

The lift does not change what the query asks.  In particular, for
$(x,z),(x',z')\in D$,
\begin{equation}
    ((x,z),(x',z'))\in\operatorname{Sep}(q)
    \quad\Longleftrightarrow\quad
    (x,x')\in\operatorname{Sep}(q)
\end{equation}
for a query lifted from $X$, and similarly for a query lifted from $Z$.

\subsection{Quantum query circuits}
\label{sec:quantum-query-circuits}
\label{subsec:quantum-query-computation}

A standard quantum query algorithm alternates input-independent reversible
transformations with calls to an input-dependent oracle
~\cite{BealsBuhrmanCleveMoscaDeWolf2001}. We use an equivalent real
phase-query form adapted to the query partitions above. The model is stated
directly in the circuit notation used throughout the paper; Appendix
~\ref{app:quantum-model-equivalence} proves exact equivalence in query count
with a standard finite response-oracle model.

\begin{definition}
\label{def:configuration-arena}
A \emph{configuration arena} is a finite nonempty set $\Omega$. Its elements
are called \emph{configurations}.
\end{definition}

\begin{definition}
\label{def:configuration-amplitude}
An \emph{amplitude assignment} on $\Omega$ is a function
$a:\Omega\to\mathbb R$ satisfying
\begin{equation}
    \sum_{c\in\Omega}a(c)^2=1.
\end{equation}
The probability reported at configuration $c$ is
$p_a(c):=a(c)^2$.
\end{definition}

The sign of $a(c)$ is invisible when a configuration is read directly, but it
matters when contributions are later recombined. Appendix
~\ref{app:square-root-carrier} records an elementary reversible-mixing
argument motivating the quadratic carrier. We otherwise take Definition
~\ref{def:configuration-amplitude} as part of the query model.

\subsubsection{Recombination}
\label{subsec:recombination}

The free gates are the input-independent linear maps preserving total squared
amplitude.

\begin{definition}
\label{def:recombination}
A \emph{recombination} on $\Omega$ is an invertible linear map
\(
    R:\mathbb R^\Omega\longrightarrow\mathbb R^\Omega
\)
which is independent of the unknown input and satisfies
\begin{equation}
    \sum_{c\in\Omega}(Ra)(c)^2
    =
    \sum_{c\in\Omega}a(c)^2
\end{equation}
for every $a\in\mathbb R^\Omega$.
\end{definition}

A recombination may move amplitude among configurations, but it does not
learn anything about the input. After ordering $\Omega$, the recombinations
are exactly the real orthogonal transformations.

\subsubsection{Phases from query partitions}
\label{subsec:phase-queries}

Phase-oracle access is a standard equivalent form of reversible black-box
access~\cite{BealsBuhrmanCleveMoscaDeWolf2001}. We use the same construction
for the blocks of an arbitrary query partition: each block receives a sign.

\begin{definition}
\label{def:phase-assignment}
Let $q\in\mathcal Q$. A \emph{phase assignment} for $q$ is a map
$\varphi:\llbracket q\rrbracket\to\{-1,+1\}$.
\end{definition}

A phase assignment does not give the query any additional information. If
two inputs lie in the same block of $\llbracket q\rrbracket$, every phase
assignment treats them in the same way.

\begin{definition}
\label{def:phase-address}
A \emph{phase address} on $\Omega$ assigns to each configuration $c\in\Omega$
a query $q_c\in\mathcal Q$ and a phase assignment
$\varphi_c:\llbracket q_c\rrbracket\to\{-1,+1\}$.
\end{definition}

The phase address is fixed before the unknown input is given. The input
determines only which block of $\llbracket q_c\rrbracket$ is selected.

\begin{definition}
\label{def:phase-query}
For an input $x\in D$, the \emph{phase query} is the linear map
$O_x:\mathbb R^\Omega\to\mathbb R^\Omega$ defined by
\begin{equation}
    (O_xa)(c):=\varphi_c(B_{q_c}(x))a(c),
\end{equation}
where $B_{q_c}(x)$ is the unique block of $\llbracket q_c\rrbracket$
containing $x$.
\end{definition}

Since every phase factor is $\pm1$, $O_x$ is an involution and preserves total
squared amplitude. It is the only operation in the circuit whose action
depends on the unknown input.

\subsubsection{Readout}
\label{subsec:quantum-readout}

\begin{definition}
\label{def:amplitude-readout}
Let $a$ be a normalized amplitude assignment on $\Omega$, and let
$\lambda:\Omega\to Y$ be a readout map. The probability of output $y\in Y$
is
\begin{equation}
    o_a(y):=\sum_{\lambda(c)=y}a(c)^2.
\end{equation}
\end{definition}

Because the fibers of $\lambda$ partition $\Omega$, $o_a$ is a probability
distribution on $Y$. Any final input-independent recombination may be
included in the circuit before this readout.

\subsubsection{Quantum query circuits}
\label{subsec:quantum-query-circuit-definition}

\begin{definition}
\label{def:quantum-query-circuit}
Fix a finite domain $D$, a finite nonempty query gate set $\mathcal Q$, and a
depth $T$. Let $\Omega$ be a configuration arena equipped with a phase
address, and let $\lambda:\Omega\to Y$ be a readout map.

A \emph{quantum query circuit of depth $T$} consists of a normalized initial
amplitude assignment $a^{\mathrm{in}}$ and input-independent recombinations
$R_0,R_1,\ldots,R_T:\mathbb R^\Omega\to\mathbb R^\Omega$.

For input $x\in D$, set $a_x^0:=R_0a^{\mathrm{in}}$. For
$0\leq t<T$, define
\begin{equation}
    a_x^{t+1}:=R_{t+1}O_xa_x^t.
\end{equation}
The circuit uses $T$ queries. Its probability of outputting $y\in Y$ is
$o_x(y):=\sum_{\lambda(c)=y}a_x^T(c)^2$.
\end{definition}

Every amplitude assignment occurring in the circuit remains normalized,
since recombinations and phase queries both preserve total squared amplitude.
We use the standard bounded-error convention, with error $1/3$ when no
superscript is shown~\cite{BealsBuhrmanCleveMoscaDeWolf2001}.

\begin{definition}
\label{def:quantum-query-complexity}
Let $0\leq\epsilon<1$. A quantum query circuit
$\epsilon$-\emph{computes} $f:D\to Y$ if
$o_x(f(x))\geq1-\epsilon$ for every $x\in D$.

The \emph{quantum query complexity} of $f$ relative to $\mathcal Q$ is
\begin{equation}
    Q_{\mathcal Q}^{\epsilon}(f)
    :=
    \min\{
        T:
        \text{a depth-$T$ quantum query circuit $\epsilon$-computes $f$}
    \}.
\end{equation}
If no such circuit exists, set $Q_{\mathcal Q}^{\epsilon}(f):=+\infty$.
We write $Q_{\mathcal Q}(f):=Q_{\mathcal Q}^{1/3}(f)$.
\end{definition}

The phase address allows the circuit to address different queries coherently
without introducing an adaptive classical selector. The information available
from each query remains exactly the information in its answer partition.

\section{The General Adversary Bound}
\label{sec:general-adversary-bound}

We use the general adversary bound with the separation filters from
Section~\ref{sec:query-models}. The bound is due to H{\o}yer, Lee, and
{\v S}palek~\cite{HoyerLeeSpalek2007}; Lee, Mittal, Reichardt,
{\v S}palek, and Szegedy give its filtered-$\gamma_2$ form
~\cite{LeeMittalReichardtSpalekSzegedy2011}.

Throughout this section all matrices are real.  For a matrix $A$, write
$\|A\|$ for its operator norm and $\|A\|_*$ for its nuclear norm, the sum of
its singular values.  The matrix pairing is
$\langle A,B\rangle:=\operatorname{Tr}(A^{\mathsf T}B)$.

\subsection{The general adversary bound}
\label{subsec:general-adversary-bound}

In the notation of Section~\ref{sec:query-models}, the weighted general
adversary bound is the following optimization
~\cite{HoyerLeeSpalek2007,LeeMittalReichardtSpalekSzegedy2011}.

\begin{definition}
\label{def:general-adversary-bound-weighted}
\label{def:general-adversary-bound}
Let $f:D\to Y$, let $\mathcal Q$ be a finite query set over $D$, and assign
each query a cost $c_q>0$. The \emph{weighted general adversary bound}
~\cite{HoyerLeeSpalek2007,LeeMittalReichardtSpalekSzegedy2011} is
\begin{equation}
    \operatorname{Adv}^{\pm}_{\mathcal Q,c}(f)
    :=
    \sup_\Gamma \|\Gamma\|,
\end{equation}
where $\Gamma$ ranges over symmetric matrices satisfying
\begin{equation}
    \Delta_f\Gamma=\Gamma,
    \qquad
    \|\Delta_q\Gamma\|\leq c_q
    \quad(q\in\mathcal Q).
\end{equation}
If these norms are unbounded, the value is $+\infty$.
\end{definition}

When every $c_q=1$, we write
$\operatorname{Adv}^{\pm}_{\mathcal Q}(f)$. Lee et al. give the
corresponding finiteness and attainment condition for filtered $\gamma_2$
~\cite{LeeMittalReichardtSpalekSzegedy2011}. In our notation it reduces to
query separation.

\begin{proposition}
\label{prop:adversary-attainment}
If $\mathcal Q$ separates $f$, then
$\operatorname{Adv}^{\pm}_{\mathcal Q,c}(f)$ is finite and attained.
\end{proposition}

\begin{proof}
The constraint set is closed.  If $f(x)\neq f(y)$, choose a query $q$ which
separates $x$ and $y$.  Then
$|\Gamma(x,y)|=|(\Delta_q\Gamma)(x,y)|\leq c_q$.  Entries on equal-output
pairs vanish.  Thus the constraint set is bounded, hence compact, and the
operator norm attains its maximum.
\end{proof}

For a Boolean target, adversary matrices have the usual bipartite block form
between the $0$- and $1$-inputs~\cite{HoyerLeeSpalek2007}.

\begin{lemma}
\label{lem:bipartite-spectrum}
Let $f:D\to\{0,1\}$ and let $A$ be symmetric with
$\Delta_fA=A$.  With $D=D_0\sqcup D_1$,
\begin{equation}
    A=
    \begin{pmatrix}
        0&G\\
        G^{\mathsf T}&0
    \end{pmatrix}
\end{equation}
for a real $D_0\times D_1$ matrix $G$, and $\|A\|=\|G\|$.  In
particular, the spectrum of $A$ is symmetric about zero.
\end{lemma}

\begin{proof}
The block form follows from the support condition.  Squaring gives diagonal
blocks $GG^{\mathsf T}$ and $G^{\mathsf T}G$, so the nonzero eigenvalues of
$A$ are the positive and negative singular values of $G$.
\end{proof}

Lee et al.\ prove that filtered $\gamma_2$ is unchanged when matching input
rows and columns are duplicated~\cite[Lemma~A.2(5)]{LeeMittalReichardtSpalekSzegedy2011}.
Lemma~\ref{lem:adversary-copy-invariance} states this invariance in the
adversary notation of Definition~\ref{def:general-adversary-bound-weighted}.

\begin{lemma}
\label{lem:adversary-copy-invariance}
Replace each input $x\in D$ by a nonempty finite set of copies.  Give every
copy the same target value as $x$, and let every query ignore the copy
label.  Then the weighted general adversary bound is unchanged.
\end{lemma}

\begin{proof}
One inequality follows by placing an adversary matrix on one chosen copy of
each input and setting all other entries to zero.

For the reverse inequality, let $\widehat\Gamma$ be an adversary matrix on
the copied domain and choose a unit eigenvector $w$ for an eigenvalue of
absolute value $\|\widehat\Gamma\|$.  Write $w_x$ for its component in the
copy space over $x$ and put $a_x:=\|w_x\|$.  Map the basis vector $e_x$ to
$w_x/a_x$ when $a_x\neq0$, choosing any unit copy vector when $a_x=0$.
This defines an isometry $R$.  Query masks and the target mask are constant
on each pair of copy fibers, so
\begin{equation}
    \Delta_q(R^{\mathsf T}\widehat\Gamma R)=R^{\mathsf T}(\Delta_q\widehat\Gamma)R,
    \qquad
    \Delta_f(R^{\mathsf T}\widehat\Gamma R)=R^{\mathsf T}(\Delta_f\widehat\Gamma)R.
\end{equation}
The compressed matrix is thus an adversary matrix on $D$.  Moreover
$\|a\|=1$, $Ra=w$, and its Rayleigh quotient at $a$ has absolute value
$\|\widehat\Gamma\|$.
\end{proof}

\subsection{Routing}
\label{subsec:adversary-routing}

Filtered $\gamma_2$ already packages arbitrary families of query filters
~\cite{LeeMittalReichardtSpalekSzegedy2011}, while the ordinary
$\gamma_2$ norm has a rank-one trace-norm formulation
~\cite{LeeShraibmanSpalek2008}. We use finite-dimensional norm duality to
write the filtered operator-norm constraints in the corresponding
decomposition form. We call that decomposition \emph{routing}.

Assume in this subsection that $\mathcal Q$ separates $f$. The routing norm
will be used for the source-matroid calculations in
Section~\ref{sec:source-adversary}.

\begin{definition}
\label{def:routing-norm}
\label{def:source-routing-norm}
A \emph{routing} of a matrix $U$ is a family of matrices
$(R_q)_{q\in\mathcal Q}$ satisfying
\begin{equation}
    \Delta_fU=\Delta_f\sum_{q\in\mathcal Q}\Delta_qR_q.
\end{equation}
Its cost is $\sum_{q\in\mathcal Q}c_q\|R_q\|_*$.  Let
$\operatorname{Route}_{f,c}(U)$ be the infimum of this cost over all routings of $U$.
\end{definition}

Only entries in $\operatorname{Sep}(f)$ are prescribed. On the subspace of
matrices supported on $\operatorname{Sep}(f)$, $\operatorname{Route}_{f,c}$ is a norm.
Proposition~\ref{prop:adversary-routing-duality} is the filtered trace-norm
decomposition dual to the adversary operator-norm constraints. Its proof
records the quotient used here and applies finite-dimensional norm duality to
Definition~\ref{def:routing-norm}.

\begin{proposition}
\label{prop:adversary-routing-duality}
\label{thm:source-routing-duality}
For every $U$,
\begin{equation}
    \operatorname{Route}_{f,c}(U)
    =
    \sup\left\{\langle A,U\rangle:\Delta_fA=A,\ \|\Delta_qA\|\leq c_q\ \text{for all }q\right\}.
\end{equation}
Consequently,
\begin{equation}
    \operatorname{Adv}^{\pm}_{\mathcal Q,c}(f)=\sup_{\|u\|=\|v\|=1}\operatorname{Route}_{f,c}(uv^{\mathsf T}).
\end{equation}
\end{proposition}

\begin{proof}
For any routing and any matrix $A$ in the displayed supremum,
\begin{equation}
    \langle A,U\rangle=\sum_q\langle\Delta_qA,R_q\rangle\leq\sum_q c_q\|R_q\|_*.
\end{equation}
For the reverse inequality, use finite-dimensional norm duality. The nuclear
and operator norms are dual, and the adjoint of the routing map applies the
query filters $\Delta_q$. The map
$(R_q)_q\mapsto\Delta_f\sum_q\Delta_qR_q$ is onto the matrices supported on
$\operatorname{Sep}(f)$ because $\mathcal Q$ separates $f$.  The norm
$\sum_qc_q\|R_q\|_*$ has dual $\max_q\|C_q\|/c_q$, and the adjoint sends
$A$ to $(\Delta_qA)_q$.  This gives the first formula.

Taking the supremum over unit $u,v$ gives the largest operator norm among
all, not necessarily symmetric, matrices $A$ with the displayed constraints.
Symmetric adversary matrices are included.  Conversely, for any such $A$,
form the symmetric dilation
\begin{equation}
    \widehat A:=\begin{pmatrix}0&A\\A^{\mathsf T}&0\end{pmatrix}
\end{equation}
on two copies of the input set.  Its norm is $\|A\|$, every filtered norm is
$\|\Delta_qA\|$, and its support lies in the lifted target-separation set.
Lemma~\ref{lem:adversary-copy-invariance} consequently bounds $\|A\|$ by the
original general adversary bound.
\end{proof}

For a nonconstant Boolean target, put $D_b:=f^{-1}(b)$. If $U$ has rows in $D_0$ and
columns in $D_1$, then
\begin{equation}
\label{eq:boolean-rectangular-routing}
    \operatorname{Route}_{f,c}(U)=\inf\left\{\sum_qc_q\|R_q\|_*:U=\sum_q\Delta_qR_q\right\},
\end{equation}
where the masks are restricted to $D_0\times D_1$.  Compressing any square
routing to this rectangular block cannot increase nuclear norm.  Conversely,
a rectangular routing embeds into the square problem by setting the other
blocks to zero.  Hence
\begin{equation}
    \operatorname{Adv}^{\pm}_{\mathcal Q,c}(f)=\sup_{\|u\|=\|v\|=1}\operatorname{Route}_{f,c}(uv^{\mathsf T}),
\end{equation}
with $u$ supported on $D_0$ and $v$ on $D_1$. If $f$ is constant,
$\operatorname{Adv}^{\pm}_{\mathcal Q,c}(f)=0$, and the rectangular
unit-vector formula is unnecessary.

\section{Span Programs}
\label{sec:span-programs}

Karchmer and Wigderson introduced span programs
~\cite{KarchmerWigderson1993}. We use the indexed-vector presentation used
by Reichardt in the quantum query setting~\cite{Reichardt2009}, but take the
availability map $A:D\to2^I$ as primitive on an arbitrary finite domain.
We work over $\mathbb R$.

\begin{definition}
\label{def:span-program}
Let $D$ be a finite set. A \emph{real span program on $D$} is a tuple $P=(V,t,I,\phi,A)$, where $V$ is a finite-dimensional real vector space, $t\in V$ is the target vector, $I$ is a finite index set, $\phi:I\to V$ assigns a vector to each index, and $A:D\to 2^I$ is the availability map. The program accepts $x\in D$ when $t\in\operatorname{span}\{\phi(i):i\in A(x)\}$, and it computes $f:D\to\{0,1\}$ when it accepts exactly the inputs with $f(x)=1$.
\end{definition}

We allow $t=0$. In that case the program accepts every input. In
Proposition~\ref{prop:span-program-matroid-shadow}, the element representing
the target is then a loop.

Write $\Phi_\phi:\mathbb R^I\to V$ for the coefficient map $\Phi_\phi(a):=\sum_{i\in I}a(i)\phi(i)$. For $S\subseteq I$, write $\mathbb R^S$ for the functions in $\mathbb R^I$ supported on $S$.

\begin{definition}
\label{def:reconstruction-set}
For a span program $P=(V,t,I,\phi,A)$ and input $x$, the \emph{reconstruction set} is $\mathcal R_x(\phi):=\{a\in\mathbb R^{A(x)}:\Phi_\phi(a)=t\}$.
\end{definition}

\begin{definition}
\label{def:separation-set}
For a span program $P=(V,t,I,\phi,A)$ and input $x$, the \emph{separation set} is $\mathcal S_x(\phi):=\{\ell\in V^*:\ell(t)=1,\ \ell(\phi(i))=0\text{ for every }i\in A(x)\}$.
\end{definition}

The conditions in Definitions~\ref{def:reconstruction-set} and
\ref{def:separation-set} are the standard positive- and negative-witness
conditions~\cite{ReichardtSpalek2008,Reichardt2009}. We call their full
solution sets the reconstruction and separation sets. Thus elements of
$\mathcal R_x(\phi)$ are positive witnesses and elements of
$\mathcal S_x(\phi)$ are negative witnesses.

Proposition~\ref{prop:reconstruction-separation-acceptance} is the standard
positive/negative witness alternative~\cite{Reichardt2009}.

\begin{proposition}
\label{prop:reconstruction-separation-acceptance}
For every input $x$, the program accepts $x$ if and only if $\mathcal R_x(\phi)\neq\varnothing$, and it rejects $x$ if and only if $\mathcal S_x(\phi)\neq\varnothing$.
\end{proposition}

\begin{proof}
The first statement is the definition of acceptance written in coefficients. If $\ell\in\mathcal S_x(\phi)$, then $\ell$ vanishes on the span of the available vectors but takes value $1$ on the target, so the target is not in that span. Conversely, if the target is not in the available span, finite-dimensional linear algebra gives a functional which vanishes on that span and is nonzero on the target; rescaling it gives an element of $\mathcal S_x(\phi)$.
\end{proof}

\begin{proposition}
\label{prop:span-program-matroid-shadow}
Let $P=(V,t,I,\phi,A)$ be a real span program. Choose a new element
$\tau\notin I$, define
\begin{equation}
    \widetilde\phi:I\sqcup\{\tau\}\to V,
    \qquad
    \widetilde\phi(i):=\phi(i)\ (i\in I),
    \qquad
    \widetilde\phi(\tau):=t,
\end{equation}
and let $M_P:=M(\widetilde\phi)$. Then, for every $x\in D$,
\begin{equation}
    P\text{ accepts }x
    \quad\Longleftrightarrow\quad
    \tau\in\cl_{M_P}(A(x)).
\end{equation}

Conversely, let $M$ be a real-representable matroid on
$I\sqcup\{\tau\}$, let $A:D\to2^I$, and let
$\psi:I\sqcup\{\tau\}\to W$ be a real representation of $M$. Set
$V:=\operatorname{span}\psi(I\sqcup\{\tau\})$. The data
\begin{equation}
    \bigl(V,\psi(\tau),I,\psi|_I,A\bigr)
\end{equation}
define a real span program, and it accepts $x$ if and only if
$\tau\in\cl_M(A(x))$.
\end{proposition}

\begin{proof}
For the first statement, Definitions~\ref{def:vector-matroid} and
\ref{def:matroid-closure} give
\begin{equation}
    \tau\in\cl_{M_P}(A(x))
    \quad\Longleftrightarrow\quad
    t\in\operatorname{span}\{\phi(i):i\in A(x)\},
\end{equation}
which is exactly the acceptance condition in
Definition~\ref{def:span-program}.

For the converse, the same represented-matroid identity gives
\begin{equation}
    \tau\in\cl_M(A(x))
    \quad\Longleftrightarrow\quad
    \psi(\tau)\in\operatorname{span}\{\psi(i):i\in A(x)\}.
\end{equation}
This is the acceptance condition for the displayed span program.
\end{proof}

Linear span in a representation is matroid closure in the associated vector matroid. The availability map remains independent data.

\part{Query Geometry and Matroidal Structure}

This part develops the adversary quantities determined by the query problem
before any program matroid is chosen. Section~\ref{sec:correlation-bounds}
introduces the correlation quantities used by the adversary method for the
quantum query circuits of Section~\ref{sec:quantum-query-circuits}. Section
~\ref{sec:matroidal-query-models} then packages functional dependence, target
determination, and pairwise separation data into the source matroid used
throughout the remainder of the paper.

\section{Correlations and the General Adversary Bound}
\label{sec:correlation-bounds}

The adversary method can be expressed through the inner products of the input-dependent quantum states~\cite{BarnumSaksSzegedy2003,HoyerLeeSpalek2007}. In the partition-query model, these inner products give the correlation picture below. The Boolean-output minimization SDP dual to the general adversary bound is standard~\cite{Reichardt2009}; we use the separation filters of Section~\ref{sec:query-models} and call its feasible families correlation transports. We retain the dual variables explicitly.

Throughout this section, $f:D\to\{0,1\}$ is Boolean and $\mathcal Q$ is a
finite query gate set.  No matroidal assumption is made.

\subsection{Pairwise correlation}
\label{subsec:correlation-separation}

Before tracking correlations quantitatively, we record the support
condition: a bounded-error circuit can distinguish two inputs only if some
allowed query separates them.

\begin{proposition}
\label{prop:quantum-separation-cover}
If some quantum query circuit $\epsilon$-computes $f$ for
$\epsilon<1/2$, then $\mathcal Q$ separates $f$.
\end{proposition}

\begin{proof}
Suppose $f(x)\neq f(y)$ but no query in $\mathcal Q$ separates $x$ and
$y$. Then $B_q(x)=B_q(y)$ for every $q\in\mathcal Q$.

Consequently every phase assignment gives the same sign on $x$ and $y$ at
every configuration, so $O_x=O_y$. Since the initial amplitude assignment
and every recombination are input independent, the entire computation is the
same on $x$ and $y$. The two inputs consequently have the same output
distribution and cannot both be computed correctly with error less than
$1/2$.
\end{proof}

Thus failure of $\mathcal Q$ to separate $f$ prevents bounded-error computation over the query set.

Fix a quantum query circuit of depth $T$. For each input $x$, let $a_x^t:\Omega\to\mathbb R$ denote its amplitude assignment after $t$ queries. The adversary method tracks the Gram matrix of these input-dependent states~\cite{BarnumSaksSzegedy2003,HoyerLeeSpalek2007}. In our real circuit model this Gram matrix is the following kernel.

\begin{definition}
\label{def:correlation-kernel}
The \emph{correlation kernel} at time $t$ is the function
$C_t:D\times D\to\mathbb R$ given by
$C_t(x,y):=\sum_{c\in\Omega}a_x^t(c)a_y^t(c)$.
\end{definition}

Since every amplitude assignment is normalized, $C_t(x,x)=1$ for every
input $x$.

The correlations arising from a quantum circuit satisfy a positivity
property which will also appear in the static optimization.

\begin{definition}
\label{def:positive-kernel}
A \emph{positive semidefinite kernel} on $D$ is a symmetric function
$K:D\times D\to\mathbb R$ such that
\begin{equation}
    \sum_{x,y\in D}
        \alpha(x)\alpha(y)K(x,y)
    \geq0
\end{equation}
for every function $\alpha:D\to\mathbb R$.
\end{definition}

The next lemma is the two-coordinate Cauchy--Schwarz inequality for positive
semidefinite kernels.

\begin{lemma}
\label{lem:positive-kernel-entry-bound}
If $K$ is a positive semidefinite kernel, then
$|K(x,y)|^2\leq K(x,x)K(y,y)$ for every $x,y\in D$. In particular, if $K(x,x)=0$, then $K(x,y)=0$ for every $y\in D$.
\end{lemma}

\begin{proof}
For fixed $x,y$, positivity applied to functions supported on $\{x,y\}$ says that
$K(x,x)+2tK(x,y)+t^2K(y,y)\geq0$ for every real $t$. Its discriminant is nonpositive, which gives the stated inequality. The final claim follows immediately.
\end{proof}

\begin{proposition}
\label{prop:correlation-positive}
For every $t$, the correlation kernel $C_t$ is positive semidefinite.
\end{proposition}

\begin{proof}
For any $\alpha:D\to\mathbb R$,
\begin{equation}
    \sum_{x,y}\alpha(x)\alpha(y)C_t(x,y)=\sum_c\left(\sum_x\alpha(x)a_x^t(c)\right)^2\geq0.
\end{equation}
\end{proof}

Because recombinations preserve the Euclidean inner product, they preserve the
entire correlation kernel.

\begin{proposition}
\label{prop:recombination-preserves-correlation}
Let $R:\mathbb R^\Omega\to\mathbb R^\Omega$ be a recombination. Then
$\sum_c(Ra)(c)(Rb)(c)=\sum_c a(c)b(c)$ for every
$a,b\in\mathbb R^\Omega$.
\end{proposition}

\begin{proof}
Apply preservation of squared amplitude to $a+b$ and $a-b$ and subtract the
two resulting identities.
\end{proof}

Accordingly, a phase query is the only operation which can change the correlation of two inputs. This is the usual support property in an adversary argument: one query can change the correlation of a pair only when that query gives different answers on the two inputs~\cite{HoyerLeeSpalek2007}. In our model, coordinate disagreement is replaced by query separation.

\begin{proposition}
\label{prop:query-correlation-support}
Let $c\in\Omega$ address the query $q_c$. If $q_c$ does not separate
$x$ and $y$, then the phase query leaves the contribution of $c$ to
$C_t(x,y)$ unchanged.
\end{proposition}

\begin{proof}
If $q_c$ does not separate $x$ and $y$, then
$B_{q_c}(x)=B_{q_c}(y)$. Hence
$\varphi_c(B_{q_c}(x))=\varphi_c(B_{q_c}(y))$, and the product of the two
signs is $1$. The contribution $a_x(c)a_y(c)$ is unchanged.
\end{proof}

\subsection{Correlation transport}
\label{subsec:static-correlation-transport}

Proposition~\ref{prop:query-correlation-support} leads to a standard static SDP. For coordinate queries, Reichardt showed that the Boolean-output general adversary bound is equal to a minimization over positive semidefinite query blocks~\cite{Reichardt2009}. Lee, Mittal, Reichardt, {\v S}palek, and Szegedy give the corresponding arbitrary-filter framework~\cite{LeeMittalReichardtSpalekSzegedy2011}. We write the same SDP using the separation filters $\Delta_q$ and call a feasible family a correlation transport. The kernels in Definition~\ref{def:correlation-transport} need not be intermediate correlation kernels of a circuit; the name records the interpretation, not an additional dynamical constraint.

\begin{definition}
\label{def:correlation-transport}
A \emph{correlation transport} for $f$ is a family of positive semidefinite kernels
$(X_q)_{q\in\mathcal Q}$ satisfying
\begin{equation}
    \sum_{\substack{q\in\mathcal Q\\
                    (x,y)\in\operatorname{Sep}(q)}}
        X_q(x,y)
    =
    1
\end{equation}
for every $(x,y)\in\operatorname{Sep}(f)$.

For an input $x$, its load is
$d_X(x):=\sum_{q\in\mathcal Q}X_q(x,x)$. The cost of the transport is
$\max_{x\in D}d_X(x)$.
\end{definition}

The kernels $X_q$ are not required to vanish on pairs which $q$ does not
separate. Those entries do not appear in the transport equations.
They cannot in general be deleted, because positivity relates the different
entries of each kernel.

\begin{definition}
\label{def:correlation-transport-value}
Define
\begin{equation}
    \operatorname{CT}_{\mathcal Q}(f)
    :=
    \inf_{(X_q)}
    \max_{x\in D}
    \sum_{q\in\mathcal Q}X_q(x,x),
\end{equation}
where the infimum is over all correlation transports for $f$. If no
correlation transport exists, set
$\operatorname{CT}_{\mathcal Q}(f):=+\infty$.
\end{definition}

This is the usual support condition for a filtered factorization; we give a direct construction in our notation~\cite{LeeMittalReichardtSpalekSzegedy2011}.

\begin{proposition}
\label{prop:correlation-transport-feasible}
If $\mathcal Q$ separates $f$, then a correlation transport for $f$ exists.
\end{proposition}

\begin{proof}
For each unordered pair $\{x,y\}$ with $f(x)\neq f(y)$, choose a query
$q_{xy}$ which separates $x$ and $y$. Such a query exists because
$\mathcal Q$ separates $f$.

Let $K_{xy}$ be the kernel which takes value $1$ on
$\{x,y\}\times\{x,y\}$ and $0$ elsewhere. It is positive semidefinite, since for every
$\alpha:D\to\mathbb R$,
\begin{equation}
    \sum_{u,v}\alpha(u)\alpha(v)K_{xy}(u,v)
    =
    \bigl(\alpha(x)+\alpha(y)\bigr)^2
    \geq0.
\end{equation}

For each $q$, let $X_q$ be the sum of the kernels $K_{xy}$ for which
$q_{xy}=q$. Sums of positive semidefinite kernels are positive semidefinite. For each
$(x,y)\in\operatorname{Sep}(f)$, exactly one kernel assigned to a query
separating that pair contributes $1$ to its transport equation. Hence
$(X_q)_{q\in\mathcal Q}$ is a correlation transport.
\end{proof}

Proposition~\ref{prop:correlation-transport-feasible} shows that $\operatorname{CT}_{\mathcal Q}(f)$ is finite whenever $\mathcal Q$ separates $f$. The standard adversary SDP attains its optimum~\cite{Reichardt2009,LeeMittalReichardtSpalekSzegedy2011}; we record the compactness argument in the present notation.

\begin{proposition}
\label{prop:correlation-transport-attainment}
If $\mathcal Q$ separates $f$, there is a correlation transport whose cost
is $\operatorname{CT}_{\mathcal Q}(f)$.
\end{proposition}

\begin{proof}
By Proposition~\ref{prop:correlation-transport-feasible}, choose one
correlation transport and let $L$ be its cost. It is enough to minimize over
transports of cost at most $L$.

By Lemma~\ref{lem:positive-kernel-entry-bound},
$|X_q(x,y)|^2\leq X_q(x,x)X_q(y,y)$. If a transport has cost at most $L$, every diagonal entry $X_q(x,x)$ lies between $0$ and $L$, and hence $|X_q(x,y)|\leq L$ for all $q,x,y$.

The transports of cost at most $L$ hence form a bounded subset of a
finite-dimensional space. The positivity constraints and transport equations define closed sets, so the bounded constraint set is compact. The cost is continuous and consequently
attains its infimum.
\end{proof}

\subsection{The correlation Lagrangian}
\label{subsec:correlation-lagrangian}

Assume that $\mathcal Q$ separates $f$. Direct Lagrange-multiplier derivations of adversary dual formulations appear in \v Spalek and Szegedy~\cite{SpalekSzegedy2006}; Reichardt's Boolean-output minimization formulation gives the SDP used here~\cite{Reichardt2009}. We retain the multipliers explicitly. By Propositions~\ref{prop:correlation-transport-feasible} and~\ref{prop:correlation-transport-attainment}, the transport problem has an attained optimum.

Introduce a scalar $t$ and minimize $t$ subject to
$\sum_qX_q(x,x)\leq t$ for every input $x$, together with the transport
equations and positivity of each $X_q$.

For kernels $K,L:D\times D\to\mathbb R$, write
$K\mathbin{\bullet}L:=\sum_{x,y}K(x,y)L(x,y)$ for their entrywise pairing.

Let $p:D\to[0,\infty)$ assign a multiplier to each load inequality. The
transport equations for $(x,y)$ and $(y,x)$ are identical because every
$X_q$ is symmetric. Replacing any pair of multipliers by their average
leaves the Lagrangian unchanged. We may consequently use a
symmetric multiplier $\Gamma:D\times D\to\mathbb R$ supported on
$\operatorname{Sep}(f)$. For each query $q$, define
$\Gamma_q(x,y):=\Gamma(x,y)$ when $(x,y)\in\operatorname{Sep}(q)$ and
$\Gamma_q(x,y):=0$ otherwise. Finally, let $D_p$ be the diagonal kernel with
$D_p(x,x)=p(x)$.

\begin{definition}
\label{def:correlation-transport-lagrangian}
The \emph{correlation Lagrangian} is
\begin{equation}
\begin{aligned}
    \mathcal L(t,(X_q),p,\Gamma)
    &:=
    t
    +
    \sum_x
        p(x)
        \left(
            \sum_qX_q(x,x)-t
        \right)\\
    &\quad+
    \sum_{(x,y)\in\operatorname{Sep}(f)}
        \Gamma(x,y)
        \left(
            1-
            \sum_{\substack{q\in\mathcal Q\\
                            (x,y)\in\operatorname{Sep}(q)}}
                X_q(x,y)
        \right).
\end{aligned}
\end{equation}
\end{definition}

Collecting the terms gives
\begin{equation}
    \mathcal L
    =
    t\left(1-\sum_xp(x)\right)
    +
    \sum_{x,y}\Gamma(x,y)
    +
    \sum_{q\in\mathcal Q}
        (D_p-\Gamma_q)\mathbin{\bullet}X_q.
\end{equation}

The first term is bounded below as a function of the unrestricted scalar
$t$ exactly when $\sum_xp(x)=1$.

To minimize the remaining terms over positive semidefinite kernels, we need to know when
$A\mathbin{\bullet}X$ can be bounded below as $X$ ranges over such kernels.

\begin{lemma}
\label{lem:positive-kernel-self-duality}
Let $A:D\times D\to\mathbb R$ be symmetric. Then
$\inf_{X\text{ positive semidefinite}}A\mathbin{\bullet}X$ is finite if and only if
$A$ is positive semidefinite. In that case the infimum is zero.
\end{lemma}

\begin{proof}
Suppose first that $A$ is positive semidefinite. Every positive semidefinite kernel $X$ can be written
as a finite sum $X=\sum_i u_i u_i^{\mathsf T}$. Hence
$A\mathbin{\bullet}X=\sum_i u_i^{\mathsf T}Au_i\geq0$, while $X=0$ gives
value zero.

If $A$ is not positive semidefinite, choose $\alpha:D\to\mathbb R$ with
$\sum_{x,y}\alpha(x)\alpha(y)A(x,y)<0$. The kernel
$X(x,y):=\alpha(x)\alpha(y)$ is positive semidefinite. Multiplying $X$ by arbitrarily
large positive scalars sends $A\mathbin{\bullet}X$ to $-\infty$.
\end{proof}

Applying the lemma separately to each $X_q$ gives the dual problem. The resulting dual is the partition-query form of the standard adversary dual appearing in Reichardt's proof~\cite{Reichardt2009}.

\begin{proposition}
\label{prop:correlation-transport-dual}
Suppose $\mathcal Q$ separates $f$. The dual value of the correlation
transport problem is
\begin{equation}
    \sup_{p,\Gamma}
        \sum_{x,y}\Gamma(x,y),
\end{equation}
where $p(x)\geq0$, $\sum_xp(x)=1$, $\Gamma$ is symmetric and supported on
$\operatorname{Sep}(f)$, and $D_p-\Gamma_q$ is positive semidefinite for every
$q\in\mathcal Q$. This dual value equals
$\operatorname{CT}_{\mathcal Q}(f)$.
\end{proposition}

\begin{proof}
The expansion following Definition~\ref{def:correlation-transport-lagrangian} shows that minimizing over the unrestricted
$t$ requires $\sum_xp(x)=1$. By
Lemma~\ref{lem:positive-kernel-self-duality}, minimizing over each positive semidefinite
kernel $X_q$ has a finite value exactly when $D_p-\Gamma_q$ is positive semidefinite.
When $\sum_xp(x)=1$ and every $D_p-\Gamma_q$ is positive semidefinite, the infimum of the Lagrangian is $\sum_{x,y}\Gamma(x,y)$.

It remains to compare the primal and dual values. Because $\mathcal Q$ separates $f$, Proposition~\ref{prop:correlation-transport-feasible} gives a primal point satisfying the transport constraints. For the dual, choose $p(x)>0$ for every input, normalize $\sum_xp(x)=1$, and take $\Gamma=0$. Then, for every nonzero $\alpha:D\to\mathbb R$, one has $\sum_xp(x)\alpha(x)^2>0$, so every kernel $D_p-\Gamma_q=D_p$ is positive definite.

After ordering the finite set $D$, positive semidefinite kernels are positive semidefinite matrices. The primal and dual are finite-dimensional semidefinite programs, and the constructed point satisfies the dual inequalities strictly. Slater strong duality gives equality of the primal and dual values~\cite[Theorem~3.4]{Lovasz2003}. The primal value is $\operatorname{CT}_{\mathcal Q}(f)$ by Definition~\ref{def:correlation-transport-value}.
\end{proof}

The function $p$ weights the load constraint of each input. The multiplier $\Gamma(x,y)$ belongs to the transport equation for the pair $(x,y)$. For a query $q$, the kernel $\Gamma_q$ keeps only the multipliers of pairs separated by $q$.

\subsection{Identifying the transport value}
\label{subsec:correlation-adversary-identification}

For coordinate queries, the identity below is Reichardt's Boolean-output dual formulation of the general adversary bound~\cite{Reichardt2009}. The argument depends only on the separation filters, so the same proof applies to the query partitions used here; this is also the arbitrary-filter setting of filtered-$\gamma_2$~\cite{LeeMittalReichardtSpalekSzegedy2011}. We include the argument to fix the normalization of $p$ and $\Gamma$.

\begin{theorem}
\label{thm:correlation-transport-adversary}
For every Boolean $f:D\to\{0,1\}$ and finite query gate set $\mathcal Q$,
\begin{equation}
    \operatorname{CT}_{\mathcal Q}(f)
    =
    \operatorname{Adv}^{\pm}_{\mathcal Q}(f),
\end{equation}
with both sides allowed to take the value $+\infty$.
\end{theorem}

\begin{proof}
Suppose first that $\mathcal Q$ does not separate $f$. Choose
$x,y\in D$ with $f(x)\neq f(y)$ such that no query separates $x$ and $y$.

The transport equation for $(x,y)$ contains no terms, so no correlation
transport exists. By
Definition~\ref{def:correlation-transport-value},
$\operatorname{CT}_{\mathcal Q}(f)=+\infty$.

For $a>0$, define a symmetric kernel $\Gamma^{(a)}$ by
$\Gamma^{(a)}(x,y)=\Gamma^{(a)}(y,x)=a$ and set all other entries to zero.
No query separates $x$ and $y$, so $\Gamma_q^{(a)}=0$ for every $q$.
Thus every $\Gamma^{(a)}$ satisfies the constraints of
Definition~\ref{def:general-adversary-bound}, while
$\|\Gamma^{(a)}\|=a$. Hence
$\operatorname{Adv}^{\pm}_{\mathcal Q}(f)=+\infty$.

Now suppose that $\mathcal Q$ separates $f$. By
Proposition~\ref{prop:correlation-transport-dual},
$\operatorname{CT}_{\mathcal Q}(f)$ equals the supremum of the transport
dual.

Let $(p,\Gamma)$ satisfy the dual constraints. Suppose $p(x)=0$. For every query $q$, the positive semidefinite kernel $D_p-\Gamma_q$ has zero diagonal entry at $x$. Lemma~\ref{lem:positive-kernel-entry-bound} shows that its entire $x$-row and $x$-column vanish. Hence $\Gamma_q(x,y)=0$ for every $y$.

Because $\mathcal Q$ separates $f$, every pair in
$\operatorname{Sep}(f)$ belongs to $\operatorname{Sep}(q)$ for some query
$q$. Since $\Gamma$ is supported on $\operatorname{Sep}(f)$, its entire
$x$-row and $x$-column hence vanish whenever $p(x)=0$.

Set $s(x):=\sqrt{p(x)}$. Define
\begin{equation}
    \widehat\Gamma(x,y)
    :=
    \begin{cases}
        \dfrac{\Gamma(x,y)}{s(x)s(y)},
        &p(x)p(y)>0,\\[1.2ex]
        0,
        &p(x)p(y)=0.
    \end{cases}
\end{equation}
On the support of $p$, positive semidefiniteness of $D_p-\Gamma_q$ is equivalent, by
diagonal rescaling, to positive semidefiniteness of $I-\widehat\Gamma_q$. Outside that
support the corresponding rows and columns of $\widehat\Gamma_q$ vanish, so
$I-\widehat\Gamma_q$ is positive semidefinite on all of $\mathbb R^D$.

By Lemma~\ref{lem:bipartite-spectrum}, the eigenvalues of
$\widehat\Gamma_q$ occur in opposite pairs. Since
$I-\widehat\Gamma_q$ is positive semidefinite, every eigenvalue of
$\widehat\Gamma_q$ is at most $1$. The opposite-pair symmetry then gives
$\|\widehat\Gamma_q\|\leq1$.

Thus $\widehat\Gamma$ satisfies the constraints of
Definition~\ref{def:general-adversary-bound}. Since
$\sum_xs(x)^2=1$,
\begin{equation}
    \sum_{x,y}\Gamma(x,y)
    =
    \sum_{x,y}s(x)\widehat\Gamma(x,y)s(y)
    \leq
    \|\widehat\Gamma\|
    \leq
    \operatorname{Adv}^{\pm}_{\mathcal Q}(f).
\end{equation}
Taking the supremum over pairs satisfying the transport-dual constraints gives
$\operatorname{CT}_{\mathcal Q}(f)
\leq\operatorname{Adv}^{\pm}_{\mathcal Q}(f)$.

For the reverse inequality, let $\Gamma_0$ be an optimal kernel from
Proposition~\ref{prop:adversary-attainment}. By
Lemma~\ref{lem:bipartite-spectrum}, its largest eigenvalue equals
$\|\Gamma_0\|$. Choose a real unit eigenvector $v$ with
$\Gamma_0v=\|\Gamma_0\|v$.

Define $p(x):=v(x)^2$ and
$\Gamma(x,y):=v(x)\Gamma_0(x,y)v(y)$. Then $p(x)\geq0$,
$\sum_xp(x)=1$, and $\Gamma$ is symmetric and supported on
$\operatorname{Sep}(f)$.

For every query $q$, the inequality $\|\Gamma_{0,q}\|\leq1$ implies that
$I-\Gamma_{0,q}$ is positive semidefinite. Hence $D_p-\Gamma_q$ is positive semidefinite because,
for every $\alpha:D\to\mathbb R$,
\begin{equation}
    \sum_{x,y}\alpha(x)\alpha(y)(D_p-\Gamma_q)(x,y)=\sum_{x,y}\alpha(x)v(x)(I-\Gamma_{0,q})(x,y)v(y)\alpha(y)\geq0.
\end{equation}
Therefore $(p,\Gamma)$ satisfies the transport-dual constraints. Its objective
value is
\begin{equation}
    \sum_{x,y}\Gamma(x,y)
    =
    v^{\mathsf T}\Gamma_0v
    =
    \|\Gamma_0\|
    =
    \operatorname{Adv}^{\pm}_{\mathcal Q}(f).
\end{equation}
Proposition~\ref{prop:correlation-transport-dual} now gives
$\operatorname{CT}_{\mathcal Q}(f)
\geq\operatorname{Adv}^{\pm}_{\mathcal Q}(f)$, completing the proof.
\end{proof}

H{\o}yer, Lee, and {\v S}palek prove the same Boolean-output lower bound in the standard query model~\cite[Theorem~2]{HoyerLeeSpalek2007}. Appendix~\ref{app:adversary-lower-bound} gives the direct proof for the phase-query model of Section~\ref{sec:quantum-query-circuits}. In particular, for $0\leq\epsilon<1/2$,
\begin{equation}
    Q_{\mathcal Q}^{\epsilon}(f)
    \geq
    \frac{1-2\sqrt{\epsilon(1-\epsilon)}}{2}
    \operatorname{Adv}^{\pm}_{\mathcal Q}(f).
\end{equation}
The identity in Theorem~\ref{thm:correlation-transport-adversary} is independent of the constant in this comparison.

\section{Matroidal Query Models}
\label{sec:matroidal-query-models}

A set of query answers may determine the answer to another query. After the answer blocks are labeled, this is the standard notion of functional dependency~\cite{Armstrong1974,Matus1991}. Functional-dependency relations have also been combined with matroidal structure in network-computation models with demanded functions~\cite{GuptaRajan2016}. Here the queries are the ground elements and closure is derived from their answer partitions. We restrict to query systems for which this closure satisfies matroid exchange, adjoin a distinguished target element recording which query sets determine $f$, and retain the pairwise separation relations used by the adversary bound. The resulting colored targeted matroid is the $f$-separation matroid.

Throughout this section, $D$ is a finite nonempty domain and $\mathcal Q$ is
a finite set of query gates over $D$. For $S\subseteq\mathcal Q$, write $\llbracket S\rrbracket:=\bigwedge_{q\in S}\llbracket q\rrbracket$, with $\llbracket\varnothing\rrbracket:=\top_D$.

\subsection{Functional dependence among queries}
\label{subsec:functional-query-dependence}

After labeling the blocks of each query partition, the condition $\llbracket S\rrbracket\leq\llbracket q\rrbracket$ is exactly the functional dependency $S\to q$~\cite{Armstrong1974,Matus1991}. In partition language, every block of $\llbracket S\rrbracket$ then lies inside a block of $\llbracket q\rrbracket$.

\begin{definition}
\label{def:query-closure}
For $S\subseteq\mathcal Q$, define $\operatorname{cl}_{\mathcal Q}(S):=\{q\in\mathcal Q:\llbracket S\rrbracket\leq\llbracket q\rrbracket\}$. We say that the queries in $\operatorname{cl}_{\mathcal Q}(S)$ are
\emph{determined by} $S$.
\end{definition}

Thus $q\in\operatorname{cl}_{\mathcal Q}(S)$ means that once the answers to
the queries in $S$ are known, asking $q$ cannot further refine the current
information partition. The partition order makes the usual
functional-dependence closure axioms immediate~\cite{Armstrong1974,Matus1991}.

\begin{proposition}
\label{prop:query-closure-operator}
The map
$\operatorname{cl}_{\mathcal Q}:2^{\mathcal Q}\to2^{\mathcal Q}$
is a closure operator: it is extensive, monotone, and idempotent.
\end{proposition}

\begin{proof}
If $q\in S$, then $\llbracket S\rrbracket\leq\llbracket q\rrbracket$, so
$S\subseteq\operatorname{cl}_{\mathcal Q}(S)$.

If $S\subseteq T$, then
$\llbracket T\rrbracket\leq\llbracket S\rrbracket$. Hence every query
determined by $S$ is also determined by $T$.

Finally, every query in $\operatorname{cl}_{\mathcal Q}(S)$ is already
determined by $S$. Therefore
$\llbracket\operatorname{cl}_{\mathcal Q}(S)\rrbracket
=\llbracket S\rrbracket$, and applying the closure a second time adds no new
queries.
\end{proof}

A functional-dependence closure need not satisfy matroid exchange. We single out the query systems for which it does.

\begin{definition}
\label{def:matroidal-query-system}
The query system $\mathcal Q$ is \emph{matroidal} if
$\operatorname{cl}_{\mathcal Q}$ satisfies
\begin{equation}
    q\in
    \operatorname{cl}_{\mathcal Q}(S\cup\{r\})
    \setminus
    \operatorname{cl}_{\mathcal Q}(S)
    \quad\Longrightarrow\quad
    r\in\operatorname{cl}_{\mathcal Q}(S\cup\{q\})
\end{equation}
for every $S\subseteq\mathcal Q$ and $q,r\in\mathcal Q$.

When this holds, let $N_{\mathcal Q}$ denote the matroid on
$\mathcal Q$ whose closure operator is
$\operatorname{cl}_{\mathcal Q}$. We call $N_{\mathcal Q}$ the
\emph{query matroid}.
\end{definition}

The terms \emph{matroidal query system} and \emph{query matroid} refer to
this matroidal special case of functional-dependence closure. General
functional-dependency structures need not be matroidal~\cite{Matus1991}; the
definition does not yet refer to a target function.

\subsection{Adding a target}
\label{subsec:target-query-matroid}

Fix a function $f:D\to Y$. A distinguished element of a matroid determines
the standard notion of a matroid port: its minimal spanning sets are the
circuits through the distinguished element~\cite{Lehman1964,Seymour1976Ports}.
Here the fixed structure is the query matroid $N_{\mathcal Q}$, and the
monotone family consists of the query sets that compute $f$. We ask whether
that family is the port of a single-element extension of $N_{\mathcal Q}$.
\begin{definition}
\label{def:target-query-matroid}
Suppose $\mathcal Q$ is matroidal, with query matroid $N_{\mathcal Q}$. We
say that $f$ is \emph{matroidally realizable over $\mathcal Q$} if there is a
single-element extension $M_f$ of $N_{\mathcal Q}$ by an element $\tau_f$
such that, for every $S\subseteq\mathcal Q$,
\begin{equation}
    \llbracket S\rrbracket\leq\Pi_f
    \quad\Longleftrightarrow\quad
    \tau_f\in\cl_{M_f}(S).
\end{equation}
We call $M_f$ the \emph{targeted query matroid} of $f$.
\end{definition}

This imposes more than the closure/computation equivalence in
Definition~\ref{def:target-query-matroid}: the extension must also restrict
on $\mathcal Q$ to the query matroid already determined by functional
dependence. Crapo's single-element-extension theorem gives the criterion
~\cite{Crapo1965}.

\begin{theorem}
\label{thm:target-modular-cut}
Let $\mathcal Q$ be matroidal, with query matroid $N_{\mathcal Q}$, and put
\begin{equation}
    \mathcal C_f
    :=
    \{F:F\text{ is a flat of }N_{\mathcal Q}
    \text{ and }\llbracket F\rrbracket\leq\Pi_f\}.
\end{equation}
Then $f$ is matroidally realizable over $\mathcal Q$ if and only if
$\mathcal C_f$ is a modular cut of $N_{\mathcal Q}$.
\end{theorem}

\begin{proof}
Apply Proposition~\ref{prop:target-extension-criterion} to the finite domain
$2^{\mathcal Q}$ with availability map $S\mapsto S$ and Boolean target
$g(S)=1$ exactly when $\llbracket S\rrbracket\leq\Pi_f$. Since $S$ and
$\cl_{N_{\mathcal Q}}(S)$ induce the same query partition, the positive flats
for this extension problem are the members of $\mathcal C_f$.
The proposition thus asks for a modular cut containing every member of
$\mathcal C_f$ and no flat outside $\mathcal C_f$. Since every flat is either
computing or noncomputing, such a cut must be exactly $\mathcal C_f$.
Hence the required target extension exists if and only if $\mathcal C_f$ is
itself a modular cut.
\end{proof}

When the extension exists, its associated modular cut is $\mathcal C_f$. By
Crapo's correspondence, Proposition~\ref{prop:modular-cut-extension} hence
determines a unique single-element extension of $N_{\mathcal Q}$ on
$\mathcal Q\sqcup\{\tau_f\}$, and the minimal fixed query sets computing
$f$ form the port of $M_f$ at $\tau_f$~\cite{Lehman1964,Seymour1976Ports}.

\begin{proposition}
\label{prop:minimal-query-circuits}
Let $M_f$ be a targeted query matroid. A set $S\subseteq\mathcal Q$ is a
fixed query set computing $f$ which is minimal under inclusion if and only if
$S\cup\{\tau_f\}$ is a circuit of $M_f$ containing $\tau_f$.
\end{proposition}

\begin{proof}
The set $S$ computes $f$ exactly when
$\tau_f\in\cl_{M_f}(S)$. By the circuit characterization of closure, this
holds exactly when some circuit containing $\tau_f$ is contained in
$S\cup\{\tau_f\}$. Minimality forces that circuit to be
$S\cup\{\tau_f\}$, and the converse follows in the same way.
\end{proof}

\subsection{Coordinate queries}
\label{subsec:coordinate-query-matroids}

Let $G$ be a finite set with at least two elements, let $D=G^k$, and let
$\mathcal Q=\{Q_1,\ldots,Q_k\}$, where $Q_i$ reveals the $i$th coordinate.

\begin{proposition}
\label{prop:coordinate-query-matroid}
The query matroid $N_{\mathcal Q}$ of the full coordinate-query system is the
free matroid on $\mathcal Q$.
\end{proposition}

\begin{proof}
If $i\notin S$, two inputs can agree on every coordinate in $S$ while taking
different values in coordinate $i$. Thus $Q_i$ is not determined by $S$.
Hence $\operatorname{cl}_{\mathcal Q}(S)=S$ for every
$S\subseteq\mathcal Q$, which is the closure operator of the free matroid.
\end{proof}

For a function $f:G^k\to Y$, a coordinate $i\in[k]$ is
\emph{essential for $f$} if there are inputs $x,x'$ differing only in
coordinate $i$ with $f(x)\neq f(x')$. Write
$\operatorname{Ess}(f):=\{i\in[k]:i\text{ is essential for }f\}$.

\begin{proposition}
\label{prop:coordinate-target-matroid}
Every function $f:G^k\to Y$ is matroidally realizable over the coordinate
queries. A fixed set $S$ of coordinate queries computes $f$ if and only if
$Q_i\in S$ for every $i\in\operatorname{Ess}(f)$.
\end{proposition}

\begin{proof}
If $S$ contains every essential coordinate, then two inputs agreeing on $S$
can be connected by changing only nonessential coordinates one at a time.
None of these changes affects $f$, so $f$ is constant on every block of
$\llbracket S\rrbracket$.

Conversely, if an essential coordinate $i$ is absent from $S$, there are two
inputs agreeing on every queried coordinate but differing only in coordinate
$i$ on which $f$ takes different values. Thus $S$ does not compute $f$.

Since $N_{\mathcal Q}$ is free, every subset is a flat. The computing flats
are exactly those containing
$\{Q_i:i\in\operatorname{Ess}(f)\}$. They form a modular cut, so
Theorem~\ref{thm:target-modular-cut} gives the targeted query matroid.
\end{proof}

In the Boolean case the targeted query matroid is also representable over $\mathbb F_2$.

\begin{proposition}
\label{prop:boolean-coordinate-representation}
Let $D=\{0,1\}^k$. Represent $Q_i$ by the $i$th standard basis vector
$e_i\in\mathbb F_2^k$, and represent the target by $v_{\tau_f}:=\sum_{i\in\operatorname{Ess}(f)}e_i$. These vectors represent the targeted query matroid of $f$.
\end{proposition}

\begin{proof}
The query vectors are independent, so their restriction is the free query
matroid. Moreover, $v_{\tau_f}$ lies in the span of
$\{e_i:Q_i\in S\}$ exactly when every essential coordinate belongs to $S$.
By Proposition~\ref{prop:coordinate-target-matroid}, this is the condition
that $S$ compute $f$.
\end{proof}

\subsection{The \texorpdfstring{$f$}{f}-separation matroid}
\label{subsec:f-separation-matroid}

The targeted query matroid $M_f$ records dependence among the queries and
which fixed query sets span the target. This is enough to recover fixed,
nonadaptive query complexity, but it does not retain all of the information
used by a quantum computation.

The missing data is the separation information already attached to the queries in Section~\ref{sec:query-models}. The relation $\operatorname{Sep}(q)$ records the pairs on which query $q$ gives different answers, while $\operatorname{Sep}(f)$ records the pairs the target itself separates. In the language of decision tables, the same pairwise data is recorded by the discernibility matrix~\cite{SkowronRauszer1992}. Dependency closure and discernibility are also treated together in generalized dependency-relation models~\cite{ChiaselottiGentileInfusino2017,ChiaselottiInfusino2018}. In our notation, the discernibility set of a pair $(x,y)$ is $\{q\in\mathcal Q:(x,y)\in\operatorname{Sep}(q)\}$.

The complementary object is standard in functional-dependency theory: the
attributes on which two tuples agree form their \emph{agree set} (or
equality set), and agree sets are tied to functional-dependency closure
~\cite{GinsburgHull1983}. Maximal agree sets which fail to determine a
specified target attribute are usually treated as maximal nondetermining sets
~\cite{LopesPetitLakhal2001}. Section~\ref{sec:source-adversary} uses the
matroidal specialization of these agree sets.

\begin{proposition}
\label{prop:query-separation-cover}
Let $M_f$ be a targeted query matroid and let $S\subseteq\mathcal Q$. The following are equivalent:
\begin{enumerate}
    \item $S$ computes $f$, meaning $\llbracket S\rrbracket\leq\Pi_f$;
    \item $\tau_f\in\cl_{M_f}(S)$;
    \item $S$ separates $f$ in the sense of Definition~\ref{def:query-set-separates-function}.
\end{enumerate}
\end{proposition}

\begin{proof}
The equivalence of the first two statements is the defining property of the
targeted query matroid.

For the first and third statements, two inputs $x$ and $y$ lie in the same
block of $\llbracket S\rrbracket$ exactly when every query in $S$ gives the
same answer on them. Thus $\llbracket S\rrbracket$ refines $\Pi_f$ exactly
when every pair with $f(x)\neq f(y)$ is separated by at least one query in
$S$.
\end{proof}

For the first and third conditions in Proposition~\ref{prop:query-separation-cover}, this is the standard discernibility criterion written in the query notation of Section~\ref{sec:query-models}~\cite{SkowronRauszer1992,ChiaselottiInfusino2018}. The target-closure condition adds the pointed-matroid description. The definition below packages these established dependence, target-extension, and pairwise-separation data for the adversary analysis.

\begin{definition}
\label{def:f-separation-matroid}
Let $M_f$ be the targeted query matroid of $f$ on
$\mathcal Q\sqcup\{\tau_f\}$. Define $\Sigma_f:\mathcal Q\sqcup\{\tau_f\}\to 2^{D\times D}$ by $\Sigma_f(q):=\operatorname{Sep}(q)$ for $q\in\mathcal Q$ and
$\Sigma_f(\tau_f):=\operatorname{Sep}(f)$.

The \emph{$f$-separation matroid} is the colored matroid
\begin{equation}
    \mathfrak M_f:=(M_f,\Sigma_f).
\end{equation}
\end{definition}

We call $\mathfrak M_f$ the \emph{$f$-separation matroid} or \emph{source matroid}. Its underlying matroid $M_f$ is the targeted query matroid. The underlying matroid records functional dependence and target determination, while the colors retain the pairwise separation relations used by the adversary bound. The coloring does not change the independent sets, rank, closure, flats, or circuits of $M_f$.

For Boolean $f$, exchanging the two output labels leaves both the fiber partition and target separation relation unchanged: $\Pi_f=\Pi_{1-f}$ and $\operatorname{Sep}(f)=\operatorname{Sep}(1-f)$. We retain $f$ together with its $f$-separation matroid.

Proposition~\ref{prop:query-separation-cover} can now be written as
\begin{equation}
    \tau_f\in\cl_{M_f}(S)
    \quad\Longleftrightarrow\quad
    \Sigma_f(\tau_f)
    \subseteq
    \bigcup_{q\in S}\Sigma_f(q)
\end{equation}
for every $S\subseteq\mathcal Q$.

The two sides retain different information. Closure records which query sets determine the target. The colors record which input pairs are distinguished by each individual query; these are the filters used by the general adversary bound of Section~\ref{sec:general-adversary-bound}. A response-labeled representable matroid can implement the same function while allowing the program dependence structure to differ from this source data.

\section{Agreement Flats in the General Adversary Bound}
\label{sec:flat-adversary}
\label{sec:source-adversary}

The $f$-separation matroid records both functional dependence among the
queries and the pairwise separation filters used by the adversary bound.
Pairwise agree sets are standard in functional-dependency theory
~\cite{GinsburgHull1983}; in a matroidal query system they become flats of
the query matroid. We use these agreement flats to reorganize the standard
Boolean general-adversary optimization
~\cite{HoyerLeeSpalek2007,LeeMittalReichardtSpalekSzegedy2011}, without
choosing a program matroid.

Throughout, let $f:D\to\{0,1\}$ be nonconstant, let
$\mathfrak M_f=(M_f,\Sigma_f)$ be its $f$-separation matroid, and put
$N_{\mathcal Q}:=M_f\setminus\tau_f$.  Write
$D_+:=f^{-1}(1)$ and $D_-:=f^{-1}(0)$.

\subsection{Agreement flats}
\label{subsec:agreement-flats}

For opposite-output inputs, let
$F_{xy}:=\{q\in\mathcal Q:(x,y)\notin\Sigma_f(q)\}$. Viewing the query
responses as attributes, this is the standard agree set of the two inputs
~\cite{GinsburgHull1983}; its complement is the discernibility set from
Section~\ref{sec:matroidal-query-models}. Agree sets are closed under
functional dependency, and maximal sets which fail to determine a specified
attribute are standard in database dependency theory
~\cite{LopesPetitLakhal2001}. The matroidal specialization is the following. Call a flat $F$ of $N_{\mathcal Q}$ \emph{computing} if $\tau_f\in\cl_{M_f}(F)$, and \emph{noncomputing} otherwise.

\begin{proposition}
\label{prop:agreement-set-flat}
\label{prop:noncomputing-flat-realized-above}
\label{prop:maximal-noncomputing-target-cocircuit}
For every $(x,y)\in D_+\times D_-$, the set $F_{xy}$ is a flat of
$N_{\mathcal Q}$ and does not span the target.  Every noncomputing flat is
contained in some $F_{xy}$.  The maximal noncomputing flats are the
complements in $\mathcal Q$ of the cocircuits of $M_f$ which contain
$\tau_f$.
\end{proposition}

\begin{proof}
If $q\in\cl_{\mathcal Q}(F_{xy})$, then the joint answers to $F_{xy}$
determine the answer to $q$.  Since $x$ and $y$ agree on $F_{xy}$, they also
agree on $q$.  Hence $q\in F_{xy}$, so $F_{xy}$ is closed.  It cannot
determine $f$, since $x$ and $y$ lie in the same block of its joint-response
partition but have different target values.

Now let $F$ be any noncomputing flat.  One block of the joint-response
partition of $F$ contains inputs $x,y$ with different target values.  Then
$F\subseteq F_{xy}$.

Finally, let $F$ be maximal among the flats which do not span $\tau_f$.
Adjoining any query outside $F$ spans the target.  Closure exchange then
shows that $F$ is a hyperplane of $M_f$ avoiding $\tau_f$, so its complement
is a cocircuit containing $\tau_f$.  The converse is the usual
hyperplane--cocircuit correspondence~\cite[Chapter~2]{Oxley2011}.
\end{proof}

Let $\mathcal L_f^-$ be the noncomputing flats of $N_{\mathcal Q}$. For a realized flat $F$, put $\mathcal A_F:=\{(x,y)\in D_+\times D_-:F_{xy}=F\}$. On $\mathcal A_F$, query $q$ separates the pair if and only if $q\notin F$.
The adversary filters depend only on this incidence rule.

\subsection{Adversary matrices on agreement flats}
\label{subsec:flat-adversary-fibers}
\label{subsec:source-routing-norm}

Let $A$ be a real matrix with rows indexed by $D_+$ and columns indexed by
$D_-$.  Write $A_F$ for its restriction to the entries in $\mathcal A_F$.
Then $A=\sum_F A_F$. If $\Delta_q$ denotes the Schur projection onto pairs separated by query $q$, then $\Delta_qA=\sum_{F:q\notin F}A_F$.

Start from the standard rectangular form of the Boolean general adversary
bound~\cite{HoyerLeeSpalek2007}. Grouping its entries by agreement flat
gives the following equivalent agreement-flat formulation.

\begin{theorem}
\label{thm:flat-adversary-decomposition}
Let $f:D\to\{0,1\}$ be nonconstant, let $\mathcal Q$ be a finite query set
whose query problem admits the $f$-separation matroid
$\mathfrak M_f=(M_f,\Sigma_f)$, and write
$D_+:=f^{-1}(1)$ and $D_-:=f^{-1}(0)$. For each $q\in\mathcal Q$, let
$\Delta_q$ be the Schur projection onto pairs separated by $q$. Then
\begin{equation}
    \operatorname{Adv}^{\pm}_{\mathcal Q}(f)
    =
    \sup_{A\in\mathbb R^{D_+\times D_-}}\|A\|
\end{equation}
subject to $\|\Delta_qA\|\leq1$ for every $q\in\mathcal Q$. Equivalently,
\begin{equation}
    \operatorname{Adv}^{\pm}_{\mathcal Q}(f)
    =
    \sup_{(A_F)}
        \left\|\sum_FA_F\right\|,
\end{equation}
where $F$ ranges over the realized agreement flats,
$A_F\in\mathbb R^{D_+\times D_-}$ is supported on $\mathcal A_F$, and
\begin{equation}
    \left\|\sum_{F:q\notin F}A_F\right\|\leq1
    \qquad(q\in\mathcal Q).
\end{equation}
\end{theorem}

\begin{proof}
Order the inputs as $D_+\sqcup D_-$.  Every adversary matrix has the form
\begin{equation}
    \Gamma=
    \begin{pmatrix}0&A\\A^{\mathsf T}&0\end{pmatrix},
\end{equation}
so $\|\Gamma\|=\|A\|$.  Its $q$-filtered block is $\Delta_qA$.  Splitting
$A$ over the disjoint agreement fibers gives the second form.
\end{proof}

The second optimization is the agreement-flat reindexing used below; its
query constraints depend only on the incidence rule $q\notin F$ for realized
agreement flats.

For the routing-norm discussion, assume that $\mathcal Q$ separates $f$; otherwise $\operatorname{Adv}^{\pm}_{\mathcal Q}(f)=+\infty$. The routing cost from Subsection~\ref{subsec:adversary-routing} gives the same constraints a second interpretation.  Let $\mathsf V:=\mathbb R^{D_+\times D_-}$.  In this
Boolean, unit-cost setting write $\operatorname{Route}_f:=\operatorname{Route}_{f,\mathbf 1}$, using the
rectangular form~\eqref{eq:boolean-rectangular-routing}.

If $\Pi_F$ denotes the projection onto the agreement fiber $\mathcal A_F$, then any routing $U=\sum_{q\in\mathcal Q}\Delta_qB_q$ satisfies $\Pi_FU=\sum_{q\notin F}\Pi_FB_q$ for every realized agreement flat $F$.  Thus demand on $F$ can leave only
through queries outside $F$. This part of the routing problem is controlled
directly by the agreement-flat data of the source matroid $\mathfrak M_f$.

\subsection{Routing and lower certificates}
\label{subsec:source-routing-certificates}

We use routings for upper certificates and feasible adversary matrices for
lower certificates. A related pair-to-query assignment appears in the
index-selection functions of Laplante, Lee, and
Szegedy~\cite{LaplanteLeeSzegedy2006}: each opposite-output pair is assigned
to a coordinate on which it differs. We use a fractional analogue constrained
to depend only on the agreement flat.

One source-defined routing starts with numbers $p_{Fq}\geq0$ on the
realized agreement flats, with $p_{Fq}=0$ for $q\in F$ and
$\sum_{q\notin F}p_{Fq}=1$.  Define
$M_q(x,y):=p_{F_{xy},q}$.  Then $\sum_qM_q$ is the all-ones matrix on
$D_+\times D_-$, and $M_q$ is supported on pairs separated by $q$.

\begin{proposition}
\label{prop:scalar-flat-routing-certificate}
For any such choice of $p_{Fq}$, let
\begin{equation}
    \kappa(p)
    :=
    \sup_{\|u\|=\|v\|=1}
        \sum_q\|M_q\circ(uv^{\mathsf T})\|_*.
\end{equation}
Then $\operatorname{Adv}^{\pm}_{\mathcal Q}(f)\leq\kappa(p)$.
\end{proposition}

\begin{proof}
For $U=uv^{\mathsf T}$, put $B_q=M_q\circ U$.  Since $B_q$ is supported on
pairs separated by $q$, $\Delta_qB_q=B_q$, and the routing identities give
$U=\sum_q\Delta_qB_q$.  Hence
$\operatorname{Route}_f(U)\leq\sum_q\|B_q\|_*$. Take the supremum over unit $u,v$.
\end{proof}

\begin{proposition}
\label{prop:routing-potential-optimality}
Suppose a routing argument proves
$\operatorname{Adv}^{\pm}_{\mathcal Q}(f)\leq L$.  If there is a matrix
$A\in\mathsf V$ with $\|\Delta_qA\|\leq1$ for every query and $\|A\|=L$,
then $\operatorname{Adv}^{\pm}_{\mathcal Q}(f)=L$.
\end{proposition}

\begin{proof}
The matrix $A$ gives the matching lower bound in
Theorem~\ref{thm:flat-adversary-decomposition}.
\end{proof}

\subsection{M\"obius decomposition on the flat lattice}
\label{subsec:flat-potential}

The same decomposition has a lattice form. The cumulative matrices below
are the zeta transform of the agreement-fiber blocks on the lattice of flats;
the fiber blocks are recovered by ordinary M\"obius inversion on a finite
poset~\cite{Rota1964}. Join shifts and their differences are standard
operators on meet/join lattices~\cite{PuschelSeifertWendler2021}, and
analogous differences indexed by queried information occur in learning-graph
dual formulations~\cite{BelovsRosmanis2014LearningGraphs}. The proposition
identifies the query Schur mask with the corresponding join difference of the
cumulative matrices.

Put
$\widehat0:=\cl_{N_{\mathcal Q}}(\varnothing)$.  For a flat $G$ of
$N_{\mathcal Q}$, let $\mathsf E_G$ keep the opposite-output pairs which
agree on every query in $G$, and put $C_G:=\mathsf E_GA$.  These matrices are determined by $A=C_{\widehat0}$, and
$\mathsf E_G\mathsf E_H=\mathsf E_{G\vee H}$.  When using the full flat
lattice below, set $A_F:=0$ for unrealized agreement flats.

\begin{proposition}
\label{prop:mobius-source-potential}
\label{prop:source-potential-query-derivative}
If $A=\sum_FA_F$ is split over agreement fibers, then
\begin{equation}
    C_G=\sum_{F\supseteq G}A_F.
\end{equation}
Thus the $A_F$ are the M\"obius increments of $C$.  If
$H_q:=\cl_{N_{\mathcal Q}}(\{q\})$, then
\begin{equation}
    \Delta_qA=C_{\widehat0}-C_{H_q},
    \qquad
    \mathsf E_G\Delta_qA=C_G-C_{G\vee H_q}.
\end{equation}
In particular, an adversary witness satisfies
$\|C_G-C_{G\vee H_q}\|\leq1$ for every $G$ and $q$.
\end{proposition}

\begin{proof}
The first formula follows because $\mathsf E_G$ keeps the fiber
$\mathcal A_F$ when $G\subseteq F$; M\"obius inversion recovers
the increments.  The remaining identities follow from
$\Delta_q=I-\mathsf E_{H_q}$ and
$\mathsf E_G\mathsf E_{H_q}=\mathsf E_{G\vee H_q}$.  For fixed $G$, $\mathsf E_G$ is the block-diagonal projection associated with the joint response partition of $G$, so it is operator-norm contractive.  Applying it to $\Delta_qA$ gives the final norm bound.
\end{proof}

Actual block-diagonal structure permits blockwise operator-norm estimates; disjoint entry supports alone do not decouple the spectral norm. When one target fiber is a singleton, the matrix geometry collapses further to an exact fractional packing and covering problem.

\subsection{Singleton-fiber adversary linear program}
\label{subsec:singleton-fiber-scalarization}

Suppose one target fiber consists of a single input. Then every adversary
block is a column vector, and the agreement fibers have disjoint row supports.
The general adversary problem reduces to a scalar packing problem.
For ordinary coordinate queries, the corresponding linear program is the
fractional-block-sensitivity program at the singleton input
~\cite{KulkarniTal2016}; related fractional certificate measures occur in
quantum certificate complexity~\cite{AaronsonQuantumCertificate}. In matroid
language, circuits through a distinguished element and their blocker
coverings are standard port structures
~\cite{Seymour1977Ports,ChenDingZang2008}. The theorem below identifies this
packing/covering problem exactly for an arbitrary finite query set. By
exchanging the roles of $D_+$ and $D_-$, assume $D_-=\{y_0\}$.

Let \(\mathcal F_{\mathrm{ag}}:=\{F_{xy}:f(x)\neq f(y)\}\) denote the realized agreement flats.

\begin{theorem}
\label{thm:singleton-fiber-scalarization}
Assume that $\mathcal Q$ separates $f$.  If $|D_-|=1$ or $|D_+|=1$, then
\begin{equation}
    \operatorname{Adv}^{\pm}_{\mathcal Q}(f)^2
    =
    \max_{b_F\geq0}
        \sum_{F\in\mathcal F_{\mathrm{ag}}}b_F
\end{equation}
subject to
\begin{equation}
    \sum_{\substack{F\in\mathcal F_{\mathrm{ag}}\\q\notin F}}b_F
    \leq1
    \qquad(q\in\mathcal Q).
\end{equation}
The dual linear program is
\begin{equation}
    \operatorname{Adv}^{\pm}_{\mathcal Q}(f)^2
    =
    \min_{\lambda_q\geq0}
        \sum_{q\in\mathcal Q}\lambda_q
\end{equation}
subject to
\begin{equation}
    \sum_{q\notin F}\lambda_q\geq1
    \qquad(F\in\mathcal F_{\mathrm{ag}}).
\end{equation}
Moreover, the dual constraints may be restricted to maximal noncomputing
flats, equivalently to cocircuits of \(M_f\) containing \(\tau_f\).
\end{theorem}

\begin{proof}
Assume \(D_-=\{y_0\}\).  Each \(A_F\) in
Theorem~\ref{thm:flat-adversary-decomposition} is a column vector supported
on the rows \(x\) with \(F_{x y_0}=F\).  Distinct fibers have disjoint supports. Therefore $\|\sum_FA_F\|^2=\sum_F\|A_F\|^2$, and for each query $q$, $\|\sum_{F:q\notin F}A_F\|^2=\sum_{F:q\notin F}\|A_F\|^2$. Setting $b_F=\|A_F\|^2$ gives the displayed packing program.  Conversely,
any nonnegative $(b_F)$ is realized by choosing one vector of norm
$\sqrt{b_F}$ on each nonempty fiber, so no information is lost.  Since
$\mathcal Q$ separates $f$, every realized agreement flat omits some query;
hence each $b_F$ is bounded by one in the packing region.  The packing optimum
is finite and attained.  The covering region is nonempty, for
example by taking $\lambda_q=1$ for every query, and its minimum is likewise
attained.  Ordinary finite-dimensional linear-programming duality gives the
covering program and equality of the two optimum values.

If \(F\) is a nonmaximal realized agreement flat, choose a maximal
noncomputing flat \(F'\supseteq F\).  Proposition~\ref{prop:maximal-noncomputing-target-cocircuit}
shows that \(F'\) is realized.  Since
\(\mathcal Q\setminus F'\subseteq\mathcal Q\setminus F\), the nonnegative
covering constraint for \(F'\) implies the constraint for \(F\).  Hence only
maximal noncomputing flats are needed.  The cocircuit description follows
from Proposition~\ref{prop:maximal-noncomputing-target-cocircuit}.
\end{proof}

For coordinate queries, this recovers the fractional-block-sensitivity
packing and fractional-certificate covering programs at the singleton input
~\cite{KulkarniTal2016}. In matroid language, the irredundant constraints are
the standard port/blocker data
~\cite{Seymour1977Ports,ChenDingZang2008}. Thus
\(\operatorname{Adv}^{\pm}(f)^2\) is exactly the fractional matching
number, and hence the fractional vertex-cover number, of the hypergraph whose
edges are the non-target parts of the target cocircuits.

\part{Matroidal Span Programs}

An abstract matroid span program records the dependence structure used to
span its target separately from any numerical representation. A real
representation assigns reconstruction and separation energies to that
structure. For regular program matroids, the representation freedom reduces
to positive element weights, and the resulting weighted basis ratios respect
Seymour decomposition. The only exceptional leaf beyond the graphic and
cographic cases is the fixed matroid $R_{10}$.

\section{Response-Labeled Matroid Span Programs}
\label{sec:matroid-span-programs}

Standard span programs label input vectors by the input responses that make
them available, and several vectors may carry the same response label
~\cite{KarchmerWigderson1993,Reichardt2009}. Non-binary span programs use
the same form with larger response alphabets~\cite{BeigiTaghavi2019}.

Forget the numerical coefficients of such a vector configuration and retain
the dependence matroid on its target and input elements; we call this matroid
$M_P$ the \emph{program matroid}. The target designation and response labels
remain separate data of the program. The program matroid is an auxiliary
object and need not equal the targeted query matroid $M_f$. When the program matroid is real-representable,
Proposition~\ref{prop:span-program-matroid-shadow} turns a representation
back into an ordinary span program. Throughout this section, fix a Boolean function $f:D\to\{0,1\}$ together
with its $f$-separation matroid
$\mathfrak M_f=(M_f,\Sigma_f)$. The function $f$ remains part of the data because $\mathfrak M_f=\mathfrak M_{1-f}$. Since $\Sigma_f(q)=\operatorname{Sep}(q)$, the coloring recovers the response partition $\mathcal B_q$: two inputs lie in the same block exactly when their pair is not in $\Sigma_f(q)$. Thus the response-label set $\mathsf R$ is determined by $\mathfrak M_f$ together with its query labels.

\subsection{Abstract matroid span programs}
\label{subsec:abstract-matroid-span-programs}

Here ``over $\mathfrak M_f$'' means that the response labels and target
function come from the query problem. We impose no matroid morphism, minor,
quotient, or extension relation between the targeted query matroid $M_f$ and
the program matroid $M_P$.

In the representable case, $M_P$ is the vector matroid of the usual labeled
span-program configuration. We use the same target-closure rule before
choosing a representation.

\begin{definition}
\label{def:f-matroid-span-program}
An \emph{$f$-matroid span program over $\mathfrak M_f$} is a tuple
\(
P=(M_P,\tau_P,\lambda_P),
\)
where $M_P=(E_P,\mathcal I_P)$ is a finite matroid, $\tau_P\in E_P$ is a
distinguished target element, and
\begin{equation}
    \lambda_P:
    E_P\setminus\{\tau_P\}
    \longrightarrow
    \mathsf R
\end{equation}
assigns a query-response label to every non-target element.

For $x\in D$, define the set of elements available on $x$ by
\begin{equation}
    A_P(x)
    :=
    \{
        e\in E_P\setminus\{\tau_P\}:
        \lambda_P(e)=(q,B_q(x))
        \text{ for some }q\in\mathcal Q
    \}.
\end{equation}

We require that the program compute $f$, meaning that, for every $x\in D$,
\begin{equation}
    f(x)=1
    \quad\Longleftrightarrow\quad
    \tau_P\in\cl_{M_P}(A_P(x)).
\end{equation}
\end{definition}

The target placement can be separated from the rest of the program. Given $P$, write $N_P:=M_P\setminus\tau_P$ and $F_P(x):=\cl_{N_P}(A_P(x))$.
Proposition~\ref{prop:target-extension-criterion} says that the possible
placements of a target over the fixed non-target matroid $N_P$, with its
inherited response labels, are exactly the modular cuts which contain every $F_P(x)$ with $f(x)=1$ and
contain no $F_P(y)$ with $f(y)=0$. In particular, if the smallest modular cut generated by the positive flats
already contains a negative flat, Proposition~\ref{prop:target-extension-criterion}
rules out that response matroid before any representation is chosen.

Proposition~\ref{prop:target-extension-criterion} does not assert that the smallest allowed modular cut gives a real-representable extension, so the program optimization still ranges over all allowed target cuts and real-representable extensions.

As in ordinary span programs, the labeling map need not be injective:
several elements may carry the same response label~\cite{Reichardt2009}. On
a given input, all of those elements are available simultaneously.

The targeted query matroid $M_f$ records functional dependence among the
allowed queries and target determination. The program matroid $M_P$ records
dependence among the elements of a particular abstract span program. The
response labels belong to the program $P$ and connect its elements to the
queries, but Definition~\ref{def:f-matroid-span-program} requires no matroid
map between $M_f$ and $M_P$.

\begin{definition}
\label{def:representable-matroid-span-program-class}
Let $\mathscr R_{\mathbb R}(f,\mathfrak M_f)$ denote the class of $f$-matroid span programs over $\mathfrak M_f$ whose program matroid is representable over $\mathbb R$. We identify two such programs when a matroid isomorphism carries target to target and preserves every response label.

Because $f$ is fixed throughout this section, we abbreviate this class by $\mathscr R(\mathfrak M_f)$.
\end{definition}

Definition~\ref{def:f-matroid-span-program} requires every non-target element
to carry a response label, so the definition contains no separate class of
unlabeled, always-available elements. If a query has the trivial one-block
partition, however, an element carrying that response label is available on
every input.

\subsection{Linear realizations}
\label{subsec:matroid-span-linear-realizations}

The standard span-program model is recovered after choosing a real
representation of $M_P$. Input vectors in the usual model are labeled by
the responses which make them available
~\cite{KarchmerWigderson1993,Reichardt2009}; for non-binary response
alphabets see~\cite{BeigiTaghavi2019}. We use the same labeling syntax for
the arbitrary query partitions of Section~\ref{sec:query-models}.

\begin{definition}
\label{def:response-labeled-span-program}
A \emph{real response-labeled span program over $\mathfrak M_f$} is a real span program $P=(V,t,I,\phi,A)$ on $D$ together with a label map $\lambda:I\to\mathsf R$ such that $A(x)=\{i\in I:\lambda(i)=(q,B_q(x))\text{ for some }q\in\mathcal Q\}$ for every $x\in D$.
\end{definition}

Thus the query model enters only through the label-induced availability map; the linear acceptance rule is the one fixed in Section~\ref{sec:span-programs}.

\begin{proposition}
\label{prop:matroid-program-to-linear-span-program}
Let
$P=(M_P,\tau_P,\lambda_P)\in\mathscr R(\mathfrak M_f)$, and choose any real
representation $\phi_P:E_P\to W_P$ of $M_P$. Set
\begin{equation}
    V_P:=\operatorname{span}_{\mathbb R}\{\phi_P(e):e\in E_P\}\subseteq W_P.
\end{equation}
Then, viewed in $V_P$, the target vector $\phi_P(\tau_P)$, the input vectors
$\phi_P(e)$ for $e\neq\tau_P$, and the labels $\lambda_P(e)$ form a real
response-labeled span program computing $f$.
\end{proposition}

\begin{proof}
On input $x$, the available vectors are those indexed by
$A_P(x)$. Since $\phi_P$ represents $M_P$, matroid closure is detected by
linear span. Therefore
\begin{equation}
\begin{aligned}
    f(x)=1
    &\Longleftrightarrow
    \tau_P\in\cl_{M_P}(A_P(x))\\
    &\Longleftrightarrow
    \phi_P(\tau_P)
    \in
    \operatorname{span}
    \{
        \phi_P(e):
        e\in A_P(x)
    \}.
\end{aligned}
\end{equation}
The last equivalence is the acceptance condition of the associated linear span program. Since $E_P$ is finite, $V_P$ is finite-dimensional; restricting to it changes none of the represented dependencies or spans.
\end{proof}

\begin{proposition}
\label{prop:linear-span-program-to-matroid-program}
Every finite real response-labeled span program computing $f$ determines a
member of $\mathscr R(\mathfrak M_f)$.
\end{proposition}

\begin{proof}
Let the span program have vector space $V$, target $t$, input index set $I$,
input-vector map $\phi:I\to V$, and label map
$\lambda:I\to\mathsf R$.

Introduce one new formal element $\tau_P$ and let
$E_P:=I\sqcup\{\tau_P\}$. Represent each $i\in I$ by $\phi(i)$ and
represent $\tau_P$ by $t$. These vectors determine a real-representable
matroid $M_P$ on $E_P$.

Let $\lambda_P$ agree with $\lambda$ on the non-target elements. For every
input $x$, the set $A_P(x)$ is the set of indices of the vectors
available to the original span program. Hence
$\tau_P\in\cl_{M_P}(A_P(x))$ exactly when $t$ lies in the span of the
available vectors. Since the original span program computes $f$, this occurs
exactly when $f(x)=1$.
\end{proof}

Thus $\mathscr R(\mathfrak M_f)$ may be viewed as the class of abstract
dependency types carried by real span programs compatible with the query
problem. Passing from a linear span program to its matroid forgets the
numerical coefficients in its dependencies. Choosing a representation puts
such numerical information back.

\subsection{Support}
\label{subsec:matroid-span-support}

Before measuring those numerical coefficients, we can ask how large the
required dependencies and separations must be using the abstract matroid
alone.

Suppose first that $f(x)=1$. Then
$\tau_P\in\cl_{M_P}(A_P(x))$, so some subset of the available elements
already spans the target.

\begin{definition}
\label{def:positive-support}
For $P\in\mathscr R(\mathfrak M_f)$ and $x$ with $f(x)=1$, define the
\emph{positive support} by
\begin{equation}
    s_+(P,x)
    :=
    \min
    \{
        |S|:
        S\subseteq A_P(x),\
        \tau_P\in\cl_{M_P}(S)
    \}.
\end{equation}
\end{definition}

The closure--circuit characterization of
Proposition~\ref{prop:closure-circuit} rewrites the same minimum in circuit
language.

\begin{proposition}
\label{prop:positive-support-circuits}
For every $x$ with $f(x)=1$,
\begin{equation}
    s_+(P,x)
    =
    \min
    \{
        |C|-1:
        C\text{ is a circuit of }M_P,\
        \tau_P\in C,\
        C\setminus\{\tau_P\}\subseteq A_P(x)
    \}.
\end{equation}
\end{proposition}

\begin{proof}
Apply Proposition~\ref{prop:closure-circuit} with $e=\tau_P$ and minimize
over the allowed subsets of $A_P(x)$.
\end{proof}

For a negative input, the dual certificate is a cocircuit which contains the
target and avoids every available element. Indeed,
$f(x)=0$ means
$\tau_P\notin\cl_{M_P}(A_P(x))$, so
Proposition~\ref{prop:cocircuit-separates-closure} gives such a cocircuit
directly.

\begin{definition}
\label{def:negative-support}
For $P\in\mathscr R(\mathfrak M_f)$ and $x$ with $f(x)=0$, define the
\emph{negative support} by
\begin{equation}
    s_-(P,x)
    :=
    \min
    \{
        |D|-1:
        D\text{ is a cocircuit of }M_P,\
        \tau_P\in D,\
        D\cap A_P(x)=\varnothing
    \}.
\end{equation}
\end{definition}

The subtraction of one removes the target itself. The two definitions use
the dual certificate families of matroid theory: a positive input asks for a
small target circuit contained in the available set together with $\tau_P$,
while a negative input asks for a small target cocircuit disjoint from the
available set away from $\tau_P$.

\begin{definition}
\label{def:matroid-span-support-profile}
Define $S_+(P):=\max_{f(x)=1}s_+(P,x)$ and $S_-(P):=\max_{f(x)=0}s_-(P,x)$, with a maximum over an empty target fiber understood to be zero. We call the pair $(S_+(P),S_-(P))$ the \emph{support profile} of $P$.
\end{definition}

The support profile depends only on the labeled abstract program. It records
the minimum number of coordinates needed for reconstruction or separation,
but not the sizes of the coefficients in a real representation. After a
representation is chosen, the same two problems are measured by squared
coefficient mass and squared functional evaluations, respectively
~\cite{ReichardtSpalek2008,Reichardt2009}. Section
~\ref{sec:quantitative-representations} makes this replacement and then
optimizes over the real representations of the fixed program $P$.

\section{Representation Costs and Optimization}
\label{sec:quantitative-representations}

This section sets up the two optimization levels used throughout the rest of
Part~III. For a fixed labeled program $P$ and a real representation of its
program matroid $M_P$, the reconstruction and separation problems carry the
standard positive and negative span-program costs
~\cite{ReichardtSpalek2008,Reichardt2009,Reichardt2014Equivalent}. We first
minimize these quadratic costs over real representations of the fixed program,
obtaining $\operatorname{lincost}(P)$, and then minimize over representable
labeled programs, obtaining $\operatorname{MSP}(f,\mathfrak M_f)$. Keeping
these two levels separate allows adversary lower certificates to distinguish
a poor representation from a poor choice of program matroid.

Throughout this section, fix
$P=(M_P,\tau_P,\lambda_P)\in\mathscr R(\mathfrak M_f)$ and a real
representation $\phi_P:E_P\to V_P$ of $M_P$. We take
$V_P=\operatorname{span}\{\phi_P(e):e\in E_P\}$ and write
$E_P^\circ:=E_P\setminus\{\tau_P\}$. By
Proposition~\ref{prop:matroid-program-to-linear-span-program}, $\phi_P$
gives a real response-labeled span program with reconstruction sets
$\mathcal R_x(\phi_P)$ and separation sets $\mathcal S_x(\phi_P)$.
We also write $\Phi_P:=\Phi_{\phi_P}$ for the coefficient map from
Section~\ref{sec:span-programs}.

\subsection{Support as coordinate support}
\label{subsec:support-as-coordinate-support}

For $a\in\mathbb R^{E_P^\circ}$, define
$\operatorname{supp}(a):=\{e\in E_P^\circ:a(e)\neq0\}$ and write
$\|a\|_0:=|\operatorname{supp}(a)|$. The notation $\|\cdot\|_0$ counts
nonzero coordinates; it is not a norm.

The representation also determines an evaluation map
$\operatorname{ev}_{\phi_P}:V_P^*\to\mathbb R^{E_P^\circ}$, defined by
$\operatorname{ev}_{\phi_P}(\ell)(e):=\ell(\phi_P(e))$.

The support quantities from Section~\ref{subsec:matroid-span-support} are the minimum coordinate supports of the reconstruction and separation problems.

\begin{proposition}
\label{prop:support-linear-characterization}
If $f(x)=1$, then
\begin{equation}
    s_+(P,x)
    =
    \min_{a\in\mathcal R_x(\phi_P)}
        \|a\|_0.
\end{equation}
If $f(x)=0$, then
\begin{equation}
    s_-(P,x)
    =
    \min_{\ell\in\mathcal S_x(\phi_P)}
        \|\operatorname{ev}_{\phi_P}(\ell)\|_0.
\end{equation}
\end{proposition}

\begin{proof}
Suppose first that $f(x)=1$. Let
$a\in\mathcal R_x(\phi_P)$ and set $S:=\operatorname{supp}(a)$. Then
$S\subseteq A_P(x)$ and
$\phi_P(\tau_P)\in\operatorname{span}\{\phi_P(e):e\in S\}$. Since
$\phi_P$ represents $M_P$, this is equivalent to
$\tau_P\in\cl_{M_P}(S)$. Thus every reconstruction gives a set admitted in
Definition~\ref{def:positive-support}.

Conversely, if $S\subseteq A_P(x)$ satisfies
$\tau_P\in\cl_{M_P}(S)$, then $\phi_P(\tau_P)$ lies in the span of the
vectors indexed by $S$. Choosing coefficients for such a linear combination
gives an element of $\mathcal R_x(\phi_P)$ supported on $S$. The first
identity follows.

Now suppose $f(x)=0$. For $\ell\in\mathcal S_x(\phi_P)$, let $F_\ell:=\{e\in E_P:\ell(\phi_P(e))=0\}$.
By Proposition~\ref{prop:functional-zero-flat}, this is a flat of $M_P$.
It contains $A_P(x)$ and excludes $\tau_P$. Choose a basis
$B_\ell$ of $F_\ell$. Since $\tau_P\notin F_\ell$,
$B_\ell\cup\{\tau_P\}$ is independent. Extend this set to a basis $B$ of
$M_P$ and set $H:=\cl_{M_P}(B\setminus\{\tau_P\})$. Then $H$ is a
hyperplane containing $F_\ell$ and excluding $\tau_P$. Its complement
$D:=E_P\setminus H$ is a cocircuit containing $\tau_P$ and disjoint from
$A_P(x)$. Moreover,
$D\setminus\{\tau_P\}$ is contained in
$\operatorname{supp}(\operatorname{ev}_{\phi_P}(\ell))$. Hence
\begin{equation}
    s_-(P,x)
    \leq
    \|\operatorname{ev}_{\phi_P}(\ell)\|_0.
\end{equation}

For the reverse inequality, choose a cocircuit $D$ attaining
$s_-(P,x)$ and let $H:=E_P\setminus D$. Then $H$ is a hyperplane containing
$A_P(x)$ and excluding $\tau_P$. Because the represented elements span
$V_P$, the subspace
$\operatorname{span}\{\phi_P(e):e\in H\}$ has codimension one in $V_P$.
Choose a nonzero linear functional whose kernel is this subspace. Since
$\tau_P\notin H$ and $H$ is a flat, its value on $\phi_P(\tau_P)$ is
nonzero; rescale it so that this value is $1$.

The resulting functional belongs to $\mathcal S_x(\phi_P)$. Moreover,
$\ell(\phi_P(e))=0$ exactly when $e\in H$. Hence
\begin{equation}
    \|\operatorname{ev}_{\phi_P}(\ell)\|_0
    =
    |D|-1
    =
    s_-(P,x).
\end{equation}
\end{proof}

Thus $s_+$ and $s_-$ count active coordinates in the reconstruction and separation problems, while the representation supplies their numerical values.

\subsection{Quadratic reconstruction and separation}
\label{subsec:quadratic-reconstruction-separation}

The same reconstruction and separation sets carry the standard span-program
quadratic costs. For a positive input, the reconstruction energy is the
minimum squared coefficient norm of a target reconstruction; for a negative
input, the separation energy is the minimum squared norm of the evaluations
of a normalized separating functional on the input vectors
~\cite{ReichardtSpalek2008,Reichardt2009}.

For $a\in\mathbb R^{E_P^\circ}$, write
$\|a\|_2^2:=\sum_{e\in E_P^\circ}a(e)^2$. This is the coordinate
sum of squares indexed by the program elements. In particular, it does not
place a norm or inner product on the representation space $V_P$.

\begin{definition}
\label{def:positive-reconstruction-energy}
For $x$ with $f(x)=1$, the \emph{reconstruction energy} of $x$ is
\begin{equation}
    c_+(\phi_P,x)
    :=
    \min_{a\in\mathcal R_x(\phi_P)}
        \|a\|_2^2.
\end{equation}
\end{definition}

This is the standard positive span-program quadratic cost, written here as
the reconstruction energy~\cite{ReichardtSpalek2008,Reichardt2009}.

\begin{definition}
\label{def:negative-separation-energy}
For $x$ with $f(x)=0$, the \emph{separation energy} of $x$ is
\begin{equation}
    c_-(\phi_P,x)
    :=
    \min_{\ell\in\mathcal S_x(\phi_P)}
        \|\operatorname{ev}_{\phi_P}(\ell)\|_2^2.
\end{equation}
\end{definition}

This is the standard negative span-program quadratic cost, written here as
the separation energy~\cite{ReichardtSpalek2008,Reichardt2009}. Equivalently,
it minimizes $\sum_{e\in E_P^\circ}\ell(\phi_P(e))^2$ subject to
$\ell(\phi_P(\tau_P))=1$ and $\ell(\phi_P(e))=0$ for every available
element $e$.

These definitions optimize over all possible supports. A reconstruction of
minimum squared size need not use a set attaining $s_+(P,x)$, and a
separating functional of minimum squared size need not have support equal to a
cocircuit attaining $s_-(P,x)$. The support quantities and the quadratic
quantities measure different features of the same constraint sets.

Because the spaces involved are finite dimensional, these two input-level
minima are attained.

\begin{proposition}
\label{prop:representation-energy-attainment}
For every $x$ with $f(x)=1$, the minimum defining
$c_+(\phi_P,x)$ is attained. For every $x$ with $f(x)=0$, the minimum
defining $c_-(\phi_P,x)$ is attained.
\end{proposition}

\begin{proof}
For a positive input, $\mathcal R_x(\phi_P)$ is a nonempty closed affine
subspace of the finite-dimensional coordinate space
$\mathbb R^{E_P^\circ}$. Intersecting it with any bounded sublevel set of
$\|a\|_2^2$ gives a compact set, so the minimum is attained.

For a negative input,
$\operatorname{ev}_{\phi_P}(\mathcal S_x(\phi_P))$ is a nonempty affine
subspace of $\mathbb R^{E_P^\circ}$ and hence is closed. The same compact
sublevel-set argument gives a minimum on this image, and every point of the
image is the evaluation vector of a member of $\mathcal S_x(\phi_P)$.
\end{proof}

Define the worst-case quantities
\begin{equation}
    C_+(\phi_P):=\max_{f(x)=1}c_+(\phi_P,x),
    \qquad
    C_-(\phi_P):=\max_{f(x)=0}c_-(\phi_P,x).
\end{equation}
These are the worst-case reconstruction and separation energies of the
represented span program.
As in Definition~\ref{def:matroid-span-support-profile}, a maximum over an
empty target fiber is understood to be zero.

\subsection{Removing arbitrary scale}
\label{subsec:representation-balancing}

The standard positive and negative span-program quadratic costs have a
one-parameter scaling freedom: rescaling the target by $c\neq0$ multiplies
the positive cost by $|c|^2$ and the negative cost by $|c|^{-2}$
~\cite{ReichardtSpalek2008,Reichardt2009}. The two-parameter representation
rescaling below is the same relative freedom together with a common scaling
of the whole represented configuration.

For $\alpha,\beta>0$, define $\phi_P^{(\alpha,\beta)}$ by
$\phi_P^{(\alpha,\beta)}(\tau_P):=\alpha\phi_P(\tau_P)$ and
$\phi_P^{(\alpha,\beta)}(e):=\beta\phi_P(e)$ for
$e\in E_P^\circ$. This is another representation of the same matroid.

\begin{proposition}
\label{prop:representation-scaling}
For every $\alpha,\beta>0$,
\begin{equation}
    C_+(\phi_P^{(\alpha,\beta)})
    =
    \frac{\alpha^2}{\beta^2}
    C_+(\phi_P),
    \qquad
    C_-(\phi_P^{(\alpha,\beta)})
    =
    \frac{\beta^2}{\alpha^2}
    C_-(\phi_P).
\end{equation}
\end{proposition}

\begin{proof}
A coefficient vector $a$ reconstructs $\phi_P(\tau_P)$ under $\phi_P$
exactly when $(\alpha/\beta)a$ reconstructs
$\alpha\phi_P(\tau_P)$ under $\phi_P^{(\alpha,\beta)}$. The squared
coefficient size is multiplied by $\alpha^2/\beta^2$.

For separation, if $\ell$ is normalized on $\phi_P(\tau_P)$, then
$\alpha^{-1}\ell$ is normalized on the rescaled target. Its evaluation on
each rescaled non-target vector is
$(\beta/\alpha)\ell(\phi_P(e))$. The separation energy is consequently
multiplied by $\beta^2/\alpha^2$.
\end{proof}

Thus the two worst-case energies change oppositely under an arbitrary common
relative scale between the target and the non-target vectors, and their
product is invariant. Optimizing over this relative scale gives the standard
geometric-mean balance of the positive and negative quadratic costs
~\cite{ReichardtSpalek2008,Reichardt2009}.

\begin{proposition}
\label{prop:representation-balance}
For every chosen representation $\phi_P$,
\begin{equation}
    \inf_{\lambda>0}
    \max
    \{
        \lambda^{-2}C_+(\phi_P),
        \lambda^2C_-(\phi_P)
    \}
    =
    \sqrt{
        C_+(\phi_P)C_-(\phi_P)
    }.
\end{equation}
\end{proposition}

\begin{proof}
Write $A:=C_+(\phi_P)$ and $B:=C_-(\phi_P)$. If $A,B>0$, the two terms
$\lambda^{-2}A$ and $\lambda^2B$ are equal when
$\lambda^4=A/B$. Their common value is $\sqrt{AB}$, and increasing either
term away from this balance cannot decrease their maximum.

If $A=0$ or $B=0$, the right side is zero, and the same value is approached
on the left by sending $\lambda$ toward zero or infinity as appropriate.
\end{proof}

\begin{definition}
\label{def:representation-cost}
The \emph{balanced cost} of the chosen representation $\phi_P$ is
\begin{equation}
    \mathcal E(\phi_P)
    :=
    \sqrt{
        C_+(\phi_P)C_-(\phi_P)
    }.
\end{equation}
\end{definition}

For the represented span program associated with $\phi_P$,
$\mathcal E(\phi_P)$ is the standard balanced span-program quadratic
complexity measure in the normalization used here
~\cite{ReichardtSpalek2008,Reichardt2009}. Propositions
~\ref{prop:representation-scaling} and~\ref{prop:representation-balance}
show explicitly how this removes the common relative scale between the target
and non-target vectors.

\subsection{Dependence on the chosen representation}
\label{subsec:representation-dependence}

The balanced cost is unaffected by a change of coordinates in the ambient
vector space.

\begin{proposition}
\label{prop:representation-coordinate-invariance}
Let $T:V_P\to W$ be an invertible linear map and define
$\psi_P:=T\circ\phi_P$. Then
$c_+(\psi_P,x)=c_+(\phi_P,x)$ on positive inputs and
$c_-(\psi_P,x)=c_-(\phi_P,x)$ on negative inputs. Consequently
$\mathcal E(\psi_P)=\mathcal E(\phi_P)$.
\end{proposition}

\begin{proof}
For reconstruction,
\begin{equation}
    \sum_e a(e)T\phi_P(e)
    =
    T\phi_P(\tau_P)
\end{equation}
holds exactly when
$\sum_e a(e)\phi_P(e)=\phi_P(\tau_P)$. Thus the coefficient vectors satisfying the reconstruction equation
are unchanged.

For separation, composition with $T$ gives a bijection between the relevant
linear functionals: a functional $\ell:W\to\mathbb R$ satisfies the separation constraints for
$\psi_P$ if and only if $\ell\circ T$ satisfies them for $\phi_P$. Corresponding
functionals have the same evaluations on the corresponding program vectors.
The two input-level energies are unchanged.
\end{proof}

This invariance does not make $\mathcal E(\phi_P)$ an invariant of the
abstract matroid $M_P$. In matroid representation theory, invertible row
operations and nonzero column scalings preserve the represented matroid and
generate projective equivalence over $\mathbb R$
~\cite[Sec.~6.3]{Oxley2011}. The reconstruction and separation energies are invariant under the former
but not under the latter: rescaling one represented element can change
reconstruction coefficients and separation evaluations. Some matroids also
admit projectively inequivalent real representations. Thus the optimization
below includes both column-weight choices inside a projective class and,
when present, genuine realization moduli of the same abstract matroid.

Let $\operatorname{Rep}_{\mathbb R}(M_P)$ denote the real representations
$\phi_P:E_P\to V_P$ of $M_P$ for which the represented elements span the
finite-dimensional ambient space $V_P$.

We now keep the abstract response-labeled program fixed and minimize the
representation cost $\mathcal E$ over all real representations of its program
matroid.

\begin{definition}
\label{def:linear-cost}
For $P\in\mathscr R(\mathfrak M_f)$, define
\begin{equation}
    \operatorname{lincost}(P)
    :=
    \inf_{\phi_P\in\operatorname{Rep}_{\mathbb R}(M_P)}
        \mathcal E(\phi_P).
\end{equation}
\end{definition}

The space of real representations of a fixed matroid is generally
noncompact, so we use an infimum; attainment of
$\operatorname{lincost}(P)$ is not asserted for an arbitrary fixed program.

The quantity $\operatorname{lincost}(P)$ keeps the abstract labeled matroid
program fixed while optimizing the numerical form of its real
representations. In particular, the support profile of $P$ does not change
during this optimization.

\subsection{Optimizing over labeled programs}
\label{subsec:matroid-span-program-cost}

For ordinary Boolean coordinate queries, Reichardt proved that minimizing the
standard span-program quadratic complexity measure over all span programs
computing $f$ gives the general adversary bound, with an optimum attained
~\cite{Reichardt2009,Reichardt2014Equivalent}. We reorganize that global
search by first fixing the abstract labeled program $P$ and optimizing real
representations of its program matroid $M_P$, and only then varying $P$.

Because $\mathfrak M_f$ does not by itself distinguish $f$ from $1-f$, we
retain $f$ in the full notation for this quantity.

\begin{definition}
\label{def:matroid-span-program-cost}
Define
\begin{equation}
    \operatorname{MSP}(f,\mathfrak M_f)
    :=
    \inf_{P\in\mathscr R_{\mathbb R}(f,\mathfrak M_f)}
        \operatorname{lincost}(P).
\end{equation}
If $\mathscr R_{\mathbb R}(f,\mathfrak M_f)$ is empty, set
$\operatorname{MSP}(f,\mathfrak M_f):=+\infty$.
\end{definition}

When $f$ is fixed and no ambiguity can arise, we may write
$\operatorname{MSP}(\mathfrak M_f)$.

The definition contains two optimizations. For a fixed abstract program
$P$, $\operatorname{lincost}(P)$ ranges over its real representations. The
outer infimum then ranges over all representable matroid span programs with
the prescribed response labels and acceptance rule. For ordinary Boolean
coordinate queries, global attainment is already part of the standard
span-program/adversary equivalence~\cite{Reichardt2014Equivalent}.
Appendix~\ref{app:universal-normal-form} gives the corresponding constructive
attainment statement for the present response-labeled quotient construction. No
fixed-program attainment claim for $\operatorname{lincost}(P)$ is made here.

For a program $P$, first choose its non-target restriction
$N_P=M_P\setminus\tau_P$ and an allowed target modular cut. A real
representation of the resulting program matroid is chosen only afterward,
and the balanced cost $\mathcal E(\phi_P)$ is evaluated at that stage.
Proposition~\ref{prop:target-extension-criterion} can discard a candidate
$N_P$ before any linear representation is considered whenever the positive
availability flats already force a negative flat into their generated
modular cut.

For finite response partitions, the quotient construction of
Appendix~\ref{app:universal-normal-form} gives a quantitative comparison
between this global program optimum and the general adversary bound. Put
\begin{equation}
    \kappa(\mathcal Q)
    :=
    \max_{q\in\mathcal Q}
        \bigl(|\mathcal B_q|-1\bigr).
\end{equation}

\begin{theorem}
\label{thm:multiresponse-msp-comparison}
If $f$ is nonconstant, $\mathcal Q$ is a finite query set that separates
$f$, each query has a finite response partition, and the query problem admits
the $f$-separation matroid $\mathfrak M_f$, then
\begin{equation}
    \operatorname{Adv}^{\pm}_{\mathcal Q}(f)
    \leq
    \operatorname{MSP}(f,\mathfrak M_f)
    \leq
    \sqrt{\kappa(\mathcal Q)}\,
    \operatorname{Adv}^{\pm}_{\mathcal Q}(f).
\end{equation}
For binary queries, $\kappa(\mathcal Q)=1$, so equality holds.
\end{theorem}

\begin{proof}
Appendix~\ref{app:universal-normal-form} expands an optimal computational
factorization by response labels and quotients by the positive-input
relations. Proposition~\ref{prop:quotient-normal-form-cost} gives the
$\sqrt{\kappa(\mathcal Q)}$ upper factor, while
Theorem~\ref{thm:adversary-lower-bounds-msp} gives the lower bound.
\end{proof}

\section{Regular Program Matroids: Representation Optimization and Certification}
\label{sec:regular-matroid-optimization}

For binary query sets,
Theorem~\ref{thm:binary-msp-adversary-equality} identifies the global
matroid-span-program optimum with the established span-program/adversary
optimum~\cite{Reichardt2009,Reichardt2014Equivalent}. For arbitrary finite
response partitions, Theorem~\ref{thm:multiresponse-msp-comparison} gives the
corresponding $\sqrt{\kappa(\mathcal Q)}$ comparison.

This section fixes a regular program matroid and optimizes only its real
representations. All real representations of a regular matroid are
projectively equivalent~\cite{BrylawskiLucas1976,Oxley2011}. After fixing one
totally unimodular representation, every other real representation differs
by an ambient change of basis and nonzero column scalings. Ambient basis
changes preserve reconstruction and separation energies, so the numerical
search reduces to positive element weights. Theorem
~\ref{thm:regular-matroid-energy} expresses each input cost as a ratio of
weighted basis sums on a restriction or contraction minor, and Corollary
~\ref{cor:regular-linear-cost} gives the resulting fixed-program objective.

Standard one-port regular-matroid theory develops current/voltage energy,
blocking duality, and determinant/basis formulations for regular matroids
with distinguished ports~\cite{Chaiken1987,Chaiken2006PortedTutte}. Here the
positive and negative span-program costs are evaluated on different
input-dependent minors while one common element-weight assignment is
optimized across all inputs.

\begin{definition}
\label{def:regular-matroid-span-program}
An $f$-matroid span program is \emph{regular} when its program matroid is
regular.
\end{definition}

Theorem~\ref{thm:regular-quotient-form} and
Corollary~\ref{cor:regular-extension-chain} show that a regular candidate can
be built by regular single-element extensions of a free matroid followed by
contraction of the added elements.  Proposition~\ref{prop:modular-cut-extension}
describes each extension by a modular cut.

Target placement uses Proposition~\ref{prop:target-extension-criterion}.  If
$N$ is a regular non-target matroid equipped with response labels, put
$F_x:=\cl_N(A(x))$ and let $\mathcal C_+$ be the smallest modular cut
containing the positive flats.  If $\mathcal C_+$ contains a negative flat,
no target extension of $N$ computes $f$.  Otherwise we keep the modular cuts
which contain $\mathcal C_+$, avoid every negative flat, and whose target
extension is regular.  Once the abstract program has been chosen, the
calculation below needs only its element weights.

\subsection{Representations reduce to element weights}
\label{subsec:regular-representation-weights}

Fix a regular program
$P=(M_P,\tau_P,\lambda_P)$ computing a nonconstant function. Choose once and for all a totally unimodular representation $e\mapsto v_e$ of $M_P$. Since the function is nonconstant, the target is not a loop, so
$v_{\tau_P}\neq0$.

Write $w=(w_e)_{e\in E_P\setminus\{\tau_P\}}>0$ for a positive weight assignment on the non-target elements, and define
\begin{equation}
    \phi_w(\tau_P):=v_{\tau_P},
    \qquad
    \phi_w(e):=\sqrt{w_e}\,v_e
    \quad(e\neq\tau_P).
\end{equation}
Independent nonzero rescaling does not change the represented matroid, so
$\phi_w$ is again a representation of $M_P$. By the standard projective
uniqueness theorem for real representations of regular matroids
~\cite{BrylawskiLucas1976,Oxley2011}, this accounts for every real
representation up to an ambient change of basis.

\begin{proposition}
\label{prop:regular-representations-are-weights}
For a regular program $P$,
\begin{equation}
    \operatorname{lincost}(P)
    =
    \inf_{w>0}
        \mathcal E(\phi_w).
\end{equation}
Thus optimizing over all real representations of $M_P$ is equivalent to
optimizing positive scalar weights on its non-target elements.
\end{proposition}

\begin{proof}
Every representation $\phi_w$ is a real representation of $M_P$, so the
right-hand side is at least $\operatorname{lincost}(P)$.

For the reverse inequality, use the unique representability recorded in
Proposition~\ref{prop:regular-basic-properties}. Any real representation
$\psi$ of $M_P$ has the form
\begin{equation}
    \psi(e)=c_e T v_e
\end{equation}
for an invertible linear map $T$ and nonzero real scalars $c_e$ on the
nonloop elements. Proposition~\ref{prop:representation-coordinate-invariance}
removes $T$. A common ambient rescaling then normalizes the target scalar to
$1$. The remaining non-target scalars may be replaced by their absolute
values: changing the sign of one represented vector only changes the sign of
the corresponding reconstruction coefficient and of the corresponding
functional evaluation, so neither squared energy changes. Setting
$w_e=c_e^2$ then gives a weighted representation $\phi_w$ with the
same balanced cost as $\psi$. Here the ambient spaces are restricted to the
spans of their columns. Loop columns are zero and may receive arbitrary
positive weights without affecting either energy; nonconstancy ensures that
the target is a nonloop, so its normalization is valid.
\end{proof}

For a regular program, Proposition~\ref{prop:regular-representations-are-weights}
reduces the real realization space to one positive scalar per non-target
element.

Regularity also turns the support profile from
Section~\ref{subsec:matroid-span-support} into a quantitative bound. A
totally unimodular representation carries the standard signed circuit and
cocircuit structure of a regular matroid: circuit dependencies and
cocircuit row-space representatives may be chosen with all nonzero
coefficients equal to $\pm1$~\cite{Oxley2011}. For a target cocircuit $D$,
let
$H:=E_P\setminus D$, choose a basis of $H$, and adjoin $\tau_P$ to obtain a
basis of $M_P$. Pivoting to this basis preserves total unimodularity. The row
corresponding to $\tau_P$ then has support $D$ and all nonzero entries in
$\{\pm1\}$~\cite{Oxley2011}. Thus target circuits and target cocircuits
produce reconstruction and separation witnesses with unit-magnitude
coefficients.

\begin{theorem}
\label{thm:regular-support-bound}
Let $P$ be a regular matroid span program computing a nonconstant function.
Then
\begin{equation}
    \operatorname{lincost}(P)
    \leq
    \sqrt{S_+(P)S_-(P)}.
\end{equation}
\end{theorem}

\begin{proof}
Use the fixed totally unimodular representation $e\mapsto v_e$ with all
non-target weights equal to $1$, and denote this unit-weight representation by
$\phi_{\mathbf 1}$.

For a positive input $x$, choose a circuit $C_x$ attaining
$s_+(P,x)$. Since $\tau_P\in C_x$, the primitive dependence on $C_x$ may
be normalized as
\begin{equation}
    v_{\tau_P}
    =
    \sum_{e\in C_x\setminus\{\tau_P\}}
        \sigma_e v_e,
    \qquad
    \sigma_e\in\{\pm1\}.
\end{equation}
This reconstructs the target with squared coefficient size
$|C_x|-1=s_+(P,x)$. Hence
\begin{equation}
    C_+(\phi_{\mathbf 1})\leq S_+(P).
\end{equation}

For a negative input $y$, choose a cocircuit $D_y$ attaining
$s_-(P,y)$. Its complement is a hyperplane containing $A_P(y)$. The
cocircuit form of total unimodularity gives a linear functional $\ell_y$
with
\begin{equation}
    \ell_y(v_{\tau_P})=1,
    \qquad
    \ell_y(v_e)=0\quad(e\notin D_y),
    \qquad
    \ell_y(v_e)\in\{\pm1\}\quad(e\in D_y\setminus\{\tau_P\}).
\end{equation}
Since $D_y\cap A_P(y)=\varnothing$, this functional vanishes on every
available element. Its squared evaluation size is hence
$|D_y|-1=s_-(P,y)$, so
\begin{equation}
    C_-(\phi_{\mathbf 1})\leq S_-(P).
\end{equation}
Thus
\begin{equation}
    \mathcal E(\phi_{\mathbf 1})
    \leq
    \sqrt{S_+(P)S_-(P)},
\end{equation}
and the definition of $\operatorname{lincost}(P)$ gives the result.
\end{proof}

A partition of the non-target elements gives a sharper bound. Let
\begin{equation}
    E_P\setminus\{\tau_P\}
    =
    E_1\sqcup\cdots\sqcup E_r.
\end{equation}

\begin{corollary}
\label{cor:regular-typed-support-bound}
Suppose there are numbers $a_j,b_j\geq0$ such that every positive input
$x$ has a target circuit $C_x$ with
\begin{equation}
    C_x\setminus\{\tau_P\}\subseteq A_P(x)
\end{equation}
and
\begin{equation}
    |C_x\cap E_j|\leq a_j
    \qquad(1\leq j\leq r),
\end{equation}
and every negative input $y$ has a target cocircuit $D_y$ satisfying
$D_y\cap A_P(y)=\varnothing$ and
\begin{equation}
    |D_y\cap E_j|\leq b_j
    \qquad(1\leq j\leq r).
\end{equation}
Then
\begin{equation}
    \operatorname{lincost}(P)
    \leq
    \sum_{j=1}^r\sqrt{a_jb_j}.
\end{equation}
\end{corollary}

\begin{proof}
Give every element of $E_j$ the same weight $w_j>0$. By the availability
hypothesis, the circuit relation used above is an admissible reconstruction
of the target, with squared coefficient size at most
\begin{equation}
    \sum_{j=1}^r\frac{a_j}{w_j},
\end{equation}
while the cocircuit functional has squared evaluation size at most
\begin{equation}
    \sum_{j=1}^r b_jw_j.
\end{equation}
Consequently
\begin{equation}
    \operatorname{lincost}(P)
    \leq
    \inf_{w_1,\ldots,w_r>0}
    \left(
        \sum_j\frac{a_j}{w_j}
    \right)^{1/2}
    \left(
        \sum_j b_jw_j
    \right)^{1/2}.
\end{equation}
For $a_j,b_j>0$, the choice
$w_j\propto\sqrt{a_j/b_j}$ makes the product equal to
$\bigl(\sum_j\sqrt{a_jb_j}\bigr)^2$. If $a_j=0<b_j$, let $w_j$ tend to
$0$; if $b_j=0<a_j$, let it tend to infinity; indices with
$a_j=b_j=0$ do not matter. Since $\operatorname{lincost}(P)$ is an
infimum, the same bound follows in all cases.
\end{proof}

Theorem~\ref{thm:regular-support-bound} is the one-class case of the
preceding argument. The typed corollary can be strictly stronger when
positive circuits and negative cocircuits use different parts of the ground
set.

\subsection{Weighted basis energies and the fixed-program objective}
\label{subsec:regular-basis-energy}

This is the main computational reduction for a fixed regular program. Once
the abstract program is chosen, projective equivalence reduces the remaining
representation freedom to positive element weights. The theorem below
evaluates every input cost from weighted bases of two input-dependent minors,
so the remaining optimization is intrinsic to the program matroid and those
weights.

The required determinant identities are standard. For a regular matroid,
Cauchy--Binet and a totally unimodular representation give the weighted
basis-generating polynomial in determinant form
~\cite{ChoeOxleySokalWagner2004}. These determinants also define the
determinantal probability measure on matroid bases, with positive column
weights multiplying the mass of a basis by the product of its element
weights~\cite{Lyons2003}. The one-port electrical/basis interpretation is standard for regular matroids~\cite{Chaiken1987,Chaiken2006PortedTutte}.

Let $N$ be a matroid containing a
distinguished nonloop element $\tau$. For a positive weight assignment $w$
on its non-target elements and a basis $B$ of $N$, put $w(B):=\prod_{e\in B\setminus\{\tau\}}w_e$.
We separate the weighted basis sum according to whether the target belongs to
the basis:
\begin{equation}
    Z_{\mathrm{in}}(N;w)
    :=
    \sum_{\substack{B\text{ basis of }N\\ \tau\in B}}w(B),
    \qquad
    Z_{\mathrm{out}}(N;w)
    :=
    \sum_{\substack{B\text{ basis of }N\\ \tau\notin B}}w(B).
\end{equation}
Whenever $N$ is a minor of $M_P$, the notation uses the restriction of the
program weight assignment $w$ to the surviving non-target elements.

For a positive input $x$, let $M_x:=M_P|(A_P(x)\cup\{\tau_P\})$, and for a negative input $y$, let $M_y:=M_P/A_P(y)$. Both are regular minors of $M_P$.

The reconstruction and separation energies are exactly the target-inclusion
odds on the corresponding restriction and contraction minors when the same
global weight assignment is used on every input.

\begin{theorem}
\label{thm:regular-matroid-energy}
Let $P$ be a regular matroid span program computing a nonconstant function,
and let $\phi_w$ be the weighted representation above. Then, for every
positive input $x$,
\begin{equation}
    c_+(\phi_w,x)
    =
    \frac{Z_{\mathrm{in}}(M_x;w)}
         {Z_{\mathrm{out}}(M_x;w)},
\end{equation}
and for every negative input $y$,
\begin{equation}
    c_-(\phi_w,y)
    =
    \frac{Z_{\mathrm{out}}(M_y;w)}
         {Z_{\mathrm{in}}(M_y;w)}.
\end{equation}
If $\tau_P$ is a coloop of $M_y$, then
$Z_{\mathrm{out}}(M_y;w)=0$ and both sides of the second identity are zero.
\end{theorem}

\begin{proof}
Fix a positive input $x$. Since the target lies in the closure of the
available set, the available elements have the same rank as $M_x$. Restrict
the fixed totally unimodular representation of $M_P$ to
$A_P(x)\cup\{\tau_P\}$, then delete dependent rows to obtain a full-row-rank
totally unimodular matrix. This compression preserves the reconstruction
problem: on the restricted columns, each deleted row is a linear combination
of the retained rows, so a coefficient vector reconstructs the compressed
target if and only if it reconstructs the original target. The prescribed
column factors $\sqrt{w_e}$ are unchanged. Write the resulting target column
as $t$, and let $C$ be the matrix whose other columns are the corresponding
weighted available columns. The available columns span the row space, so
$K:=CC^{\mathsf T}$ is positive definite. The minimum squared coefficient
size in a reconstruction $Ca=t$ is
\begin{equation}
    c_+(\phi_w,x)=t^{\mathsf T}K^{-1}t.
\end{equation}
By Cauchy--Binet and total unimodularity,
\begin{equation}
    \det K=Z_{\mathrm{out}}(M_x;w),
\end{equation}
while adding the target column gives
\begin{equation}
    \det(K+tt^{\mathsf T})
    =
    Z_{\mathrm{out}}(M_x;w)+Z_{\mathrm{in}}(M_x;w).
\end{equation}
The matrix determinant lemma gives
$\det(K+tt^{\mathsf T})=\det K(1+t^{\mathsf T}K^{-1}t)$, and comparison
proves the positive formula.

Now fix a negative input $y$ and choose a basis $B_y$ of $A_P(y)$.
Pivot the fixed totally unimodular representation on the columns of $B_y$.
Basis pivoting preserves total unimodularity. Deleting the pivot rows and the
contracted columns represents $M_P/B_y$; the elements of
$A_P(y)\setminus B_y$ become loops and may then be deleted. The resulting
matrix represents $M_y=M_P/A_P(y)$. Since the pivoting acts on rows while the
program weights are column scalars, every surviving non-target column keeps
its factor $\sqrt{w_e}$. Thus the contraction is regular with the same
surviving weights.

A separating functional vanishes on $\operatorname{span}A_P(y)$, so it
descends to the quotient by that span. Conversely, a functional on the
quotient with value $1$ on the target lifts to a separating functional on the
original representation. Contracted elements contribute zero to the
separation energy.

If $\tau_P$ is a coloop of $M_y$, the other elements lie in a hyperplane
excluding the target. A normalized functional can vanish on every
non-target element, so $c_-(\phi_w,y)=0$; equivalently no basis of $M_y$
avoids the target, and $Z_{\mathrm{out}}(M_y;w)=0$.

Otherwise the non-target elements span $M_y$. Use the contracted totally
unimodular matrix above, write its target column as $t$, let $C$ contain its
weighted non-target columns, and put $K=CC^{\mathsf T}$. Writing a linear
functional as $z^{\mathsf T}$, the separation problem is
\begin{equation}
    \min\{z^{\mathsf T}Kz:z^{\mathsf T}t=1\},
\end{equation}
whose value is $1/(t^{\mathsf T}K^{-1}t)$. The same determinant calculation
gives
\begin{equation}
    t^{\mathsf T}K^{-1}t
    =
    \frac{Z_{\mathrm{in}}(M_y;w)}
         {Z_{\mathrm{out}}(M_y;w)},
\end{equation}
and the negative formula follows.
\end{proof}

Equivalently, under the weighted random-basis measure on $M_x$,
\begin{equation}
    c_+(\phi_w,x)
    =
    \frac{\Pr[\tau_P\in B]}{\Pr[\tau_P\notin B]},
\end{equation}
while on $M_y$,
\begin{equation}
    c_-(\phi_w,y)
    =
    \frac{\Pr[\tau_P\notin B]}{\Pr[\tau_P\in B]}.
\end{equation}
Thus the positive and negative witness energies are opposite target-basis
odds on the two input minors~\cite{Lyons2003}.

The determinant and weighted-basis identities are standard; their
consequence here is that the representation optimization of a fixed regular
program is intrinsic to the abstract matroid and its element weights.

In graphic span programs, a single positive edge-weight assignment has
already been optimized against the worst positive and negative witness
energies~\cite[Thm.~12]{JefferyKimmel2017}. The corollary below extends this
shared-weight form to arbitrary regular program matroids, with the two
energies expressed by weighted basis ratios of the input restriction and
contraction minors.

\begin{corollary}
\label{cor:regular-linear-cost}
Let $P$ be a regular matroid span program computing a nonconstant function.
Then
\begin{equation}
\begin{aligned}
    \operatorname{lincost}(P)
    =
    \inf_{w>0}
    \Bigg(&
        \max_{f(x)=1}
        \frac{Z_{\mathrm{in}}(M_x;w)}
             {Z_{\mathrm{out}}(M_x;w)}
        \cdot
        \max_{f(y)=0}
        \frac{Z_{\mathrm{out}}(M_y;w)}
             {Z_{\mathrm{in}}(M_y;w)}
    \Bigg)^{1/2}.
\end{aligned}
\end{equation}
\end{corollary}

\begin{proof}
Combine Proposition~\ref{prop:regular-representations-are-weights},
Theorem~\ref{thm:regular-matroid-energy}, and
Definition~\ref{def:representation-cost}.
\end{proof}

For a fixed regular program, the continuous optimization thus has only
positive element weights as variables. The objective is expressed through
weighted basis sums of restrictions and contractions rather than through a
free Gram factorization or semidefinite variable. This is a reduction in
optimization variables, not by itself a claim that the basis sums are
efficiently computable in every input model.

\subsection{Comparison with the Adversary Optimum}
\label{subsec:regular-optimality-certification}

Every represented program gives an upper bound on the universal optimum, so
for any regular program $P$,
\begin{equation}
    \operatorname{Adv}^{\pm}_{\mathcal Q}(f)
    \leq
    \operatorname{lincost}(P).
\end{equation}

\begin{proposition}
\label{prop:regular-optimality-certificate}
Suppose
$\operatorname{Adv}^{\pm}_{\mathcal Q}(f)=L$, and let
$P\in\mathscr R_{\mathbb R}(f,\mathfrak M_f)$ be a regular program. If
$\operatorname{lincost}(P)=L$, equivalently if the weight optimization of
Corollary~\ref{cor:regular-linear-cost} has value $L$, then
\begin{equation}
    \operatorname{lincost}(P)
    =
    \operatorname{MSP}(f,\mathfrak M_f)
    =
    \operatorname{Adv}^{\pm}_{\mathcal Q}(f)
    =
    L.
\end{equation}
Thus $P$ is globally optimal.
\end{proposition}

\begin{proof}
The adversary certificate gives the first equality. Theorem~\ref{thm:adversary-lower-bounds-msp}
gives
\begin{equation}
    L
    =
    \operatorname{Adv}^{\pm}_{\mathcal Q}(f)
    \leq
    \operatorname{MSP}(f,\mathfrak M_f)
    \leq
    \operatorname{lincost}(P)
    =
    L.
\end{equation}
\end{proof}

Thus a source certificate of value $L$ and a regular program with
$\operatorname{lincost}(P)=L$ certify global optimality. Equality of the
fixed-program infimum does not imply that an optimizing weight assignment
exists.

\section{Regular Program Costs under Seymour Decomposition}
\label{sec:seymour-calculus}

Corollary~\ref{cor:regular-linear-cost} reduces the representation problem for
a fixed regular program to positive weights on its non-target elements. For a
fixed assignment $w$, the positive and negative energies are determined by
ratios of the weighted basis sums $Z_{\mathrm{in}}$ and
$Z_{\mathrm{out}}$ on regular restriction and contraction minors. Each such
minor admits a $\{1,2,3\}$-decomposition. Weighted basis-generating functions
and their behavior under ported matroid sums are classical
~\cite{Chaiken1989PortedMatroid,Chaiken2006PortedTutte}. Our use here is to
propagate the target-conditioned restriction/contraction ratios of a fixed
regular span program under one shared weight assignment. The basis ratio is
preserved by the following reductions of summands not containing the target.

Ratios are interpreted in the extended nonnegative reals: a target loop
gives $Z_{\mathrm{in}}/Z_{\mathrm{out}}=0$, and a target coloop gives
$Z_{\mathrm{in}}/Z_{\mathrm{out}}=\infty$, with reciprocal zero. At positive
weights the two basis sums never vanish simultaneously.

\subsection{One-sums}
\label{subsec:seymour-one-sum}

\begin{proposition}
\label{prop:seymour-one-sum}
Let $N=N_0\oplus_1N_1$ be a regular matroid with target
$\tau\in E(N_0)$, and let $w$ be a positive weight assignment on the
non-target elements. Then
\begin{equation}
    Z_{\mathrm{in}}(N;w)
    =Z_{\mathrm{in}}(N_0;w)Z(N_1;w),
    \qquad
    Z_{\mathrm{out}}(N;w)
    =Z_{\mathrm{out}}(N_0;w)Z(N_1;w),
    \label{eq:seymour-one-sum-factorization}
\end{equation}
and hence
\begin{equation}
    \frac{Z_{\mathrm{in}}(N;w)}{Z_{\mathrm{out}}(N;w)}
    =
    \frac{Z_{\mathrm{in}}(N_0;w)}{Z_{\mathrm{out}}(N_0;w)}.
    \label{eq:seymour-one-sum-ratio}
\end{equation}
\end{proposition}

\begin{proof}
Every basis of a direct sum is the disjoint union of a basis of $N_0$ and a
basis of $N_1$. Since the target belongs to $N_0$, separating the bases of
$N$ according to whether they contain $\tau$ gives
\eqref{eq:seymour-one-sum-factorization}. The common factor $Z(N_1;w)$ then
cancels.
\end{proof}

For either input minor in Theorem~\ref{thm:regular-matroid-energy}, a
$1$-sum component not containing the target changes neither the
reconstruction energy nor the separation energy.

\subsection{Two-sums}
\label{subsec:seymour-two-sum}

Let $N$ be a regular matroid with distinguished element $e$ and positive
weights on $E(N)\setminus\{e\}$. Define
\begin{equation}
    P_0(N,e;w)
    :=\sum_{\substack{B\text{ basis of }N\\ e\notin B}}w(B),
    \qquad
    P_1(N,e;w)
    :=\sum_{\substack{B\text{ basis of }N\\ e\in B}}
        w(B\setminus\{e\}).
    \label{eq:seymour-two-sum-P}
\end{equation}
When $P_1(N,e;w)>0$, put
$m_e(N;w):=P_0(N,e;w)/P_1(N,e;w)$.
\begin{theorem}
\label{thm:seymour-two-sum}
Let $M=H\oplus_2N$, where $H\cap N=\{e\}$ and the target $\tau$ belongs to
$H$. Give all elements of $M$ positive weights, inherited by
$H\setminus\{e\}$ and $N\setminus\{e\}$. For
$\sigma\in\{\mathrm{in},\mathrm{out}\}$, let $Z_\sigma(H;m)$ denote the
weighted basis sum of $H$ after assigning weight $m$ to $e$. If
$P_1(N,e;w)>0$, then
\begin{equation}
    Z_\sigma(M;w)
    =P_1(N,e;w)\,
      Z_\sigma\!\left(H;m_e(N;w)\right).
    \label{eq:seymour-two-sum-factorization}
\end{equation}
Consequently,
\begin{equation}
    \frac{Z_{\mathrm{in}}(M;w)}{Z_{\mathrm{out}}(M;w)}
    =
    \frac{Z_{\mathrm{in}}(H;m_e)}{Z_{\mathrm{out}}(H;m_e)}.
    \label{eq:seymour-two-sum-ratio}
\end{equation}
\end{theorem}

\begin{proof}
Write $Z_\sigma(H;m)=A_{\sigma,0}+mA_{\sigma,1}$, separating bases of $H$
according to whether they omit or contain $e$. A basis of a proper $2$-sum
uses the common element on exactly one side. Therefore
\begin{equation}
    Z_\sigma(M;w)
    =A_{\sigma,0}P_1+A_{\sigma,1}P_0
    =P_1Z_\sigma(H;m_e),
\end{equation}
which proves both identities.
\end{proof}

If $P_0=0$, the same identity is interpreted by deleting $e$ from $H$; if
$P_1=0$, it is interpreted by contracting $e$. Thus a $2$-sum component not containing the target
in an input minor is replaced, for the basis ratio in
Theorem~\ref{thm:regular-matroid-energy}, by the common element $e$ with
weight $m_e$.

\begin{proposition}
\label{prop:seymour-two-sum-input-minors}
Let $P=(N,e,\lambda)$ be a regular $f$-matroid span program over
$\mathfrak M_f$, let $w$ be positive weights on
$E(N)\setminus\{e\}$, and let $\phi_w$ be the corresponding weighted real
representation. For an input $z$, put
\[
    N_z^+:=N|(A_P(z)\cup\{e\}),
    \qquad
    N_z^-:=N/A_P(z).
\]
Then
\begin{equation}
    m_e(N_z^+;w)
    =
    \begin{cases}
        c_+(\phi_w,z)^{-1}, & f(z)=1,\\
        0, & f(z)=0,
    \end{cases}
    \label{eq:two-sum-positive-input-value}
\end{equation}
and
\begin{equation}
    m_e(N_z^-;w)
    =
    \begin{cases}
        \infty, & f(z)=1,\\
        c_-(\phi_w,z), & f(z)=0.
    \end{cases}
    \label{eq:two-sum-negative-input-value}
\end{equation}
\end{proposition}

\begin{proof}
If $f(z)=1$, then $e$ is spanned by the available elements. In $N_z^+$,
Theorem~\ref{thm:regular-matroid-energy} gives
$c_+(\phi_w,z)=P_1/P_0$, hence $m_e=P_0/P_1=c_+(\phi_w,z)^{-1}$.
If $f(z)=0$, then $e$ is a coloop of $N_z^+$, so $P_0=0$ and $m_e=0$.

For $N_z^-$, if $f(z)=0$, Theorem~\ref{thm:regular-matroid-energy} gives
$c_-(\phi_w,z)=P_0/P_1=m_e$. If $f(z)=1$, contraction of the available elements
makes $e$ a loop, so $P_1=0$ and $m_e=\infty$.
\end{proof}

If all non-target weights of $N$ are multiplied by $\lambda>0$, then
$P_0$ and $P_1$ are homogeneous of degrees $r(N)$ and $r(N)-1$,
respectively. Hence
\begin{equation}
    m_e(N;\lambda w)=\lambda m_e(N;w).
    \label{eq:two-sum-homogeneity}
\end{equation}
The quadratic formulas of Section~\ref{sec:regular-matroid-optimization}
also give the monotonicity used later: increasing an available positive
weight cannot increase reconstruction energy, while increasing a finite
weight in a negative contraction minor cannot decrease separation energy.

\subsection{Three-sums}
\label{subsec:seymour-three-sum}

For the local replacements in this subsection, $\oplus_3$ also denotes the
same cycle-space gluing along a coindependent triangle when a replacement
piece has fewer than seven elements. In particular, gluing $M(K_4)$ is the standard $\Delta$--$Y$ operation
~\cite{Seymour1980,Oxley2011}; see also
\cite[Sec.~2]{BercziMatravolgyiSchwarcz2023}.
The size requirement in Definition~\ref{def:regular-matroid-sums} remains in
force for proper sums in a Seymour decomposition.

Let $N$ be a regular matroid with distinguished triangle
$T=\{t_1,t_2,t_3\}$ and positive weights on $E(N)\setminus T$. For each
$L\subseteq T$, define
\begin{equation}
    \Phi_L(N,T;w)
    :=
    \sum_{\substack{S\subseteq E(N)\setminus T\text{ independent}\\
                     S\cup T\text{ spans }N\\
                     \cl_N(S)\cap T=L}}
        w(S).
    \label{eq:seymour-three-sum-phi}
\end{equation}
Since $T$ is a circuit, if $\cl_N(S)$ contains two elements of $T$, it
contains the third. The only possible intersections are
$\varnothing$, $\{t_1\}$, $\{t_2\}$, $\{t_3\}$, and $T$. Put
\begin{equation}
    a:=\Phi_T,
    \qquad
    p_i:=\Phi_{\{t_i\}}\quad(i=1,2,3),
    \qquad
    c:=\Phi_\varnothing.
    \label{eq:seymour-three-sum-five}
\end{equation}

\begin{lemma}
\label{lem:seymour-three-sum-five}
Let $M=H\oplus_3N$, where $H\cap N=T$ and the target belongs to $H\setminus T$.
For fixed weights on $H\setminus T$, each of
$Z_{\mathrm{in}}(M;w)$ and $Z_{\mathrm{out}}(M;w)$ is a linear function of
$a,p_1,p_2,p_3,c$, with coefficients depending only on $H$ and its weights.
\end{lemma}

\begin{proof}
Use binary representations with the common triangle spanning a common
two-dimensional space $U$. For an independent set $S$ on either side with
$S\cup T$ spanning that side, put $U_S=\langle S\rangle\cap U$. Its dimension
is $0$, $1$, or $2$. Over $\mathbb F_2$, the only one-dimensional subspaces
of $U$ are the three lines represented by $t_1,t_2,t_3$; thus these are exactly
the five classes in~\eqref{eq:seymour-three-sum-five}.

A union $S_H\cup S_N$ is a basis of the glued matroid precisely when the
sets on each side are independent, span their side after adjoining $T$, and
$U_{S_H}\oplus U_{S_N}=U$. Indeed independence requires their intersection
to be zero, and the basis size $r(H)+r(N)-2$ requires their dimensions to
sum to two. The compatible class pairs are therefore $(T,\varnothing)$,
$(\varnothing,T)$, and $(\{t_i\},\{t_j\})$ with $i\ne j$; this is the
usual $3$-sum basis correspondence~\cite{Seymour1980,Oxley2011}; see also
\cite[Sec.~2]{BercziMatravolgyiSchwarcz2023}.

For $\sigma\in\{\mathrm{in},\mathrm{out}\}$, restrict the $H$-side class
sums according to target inclusion, omitting the target from the weight
product, and denote them by $a_H^\sigma,p_{H,i}^\sigma,c_H^\sigma$. Then
\[
 Z_\sigma(H\oplus_3N)
 =a_H^\sigma c+c_H^\sigma a+\sum_{i\ne j}p_{H,i}^\sigma p_j.
\]
This proves the stated linear dependence.
\end{proof}

\begin{proposition}
\label{prop:regular-triangle-identity}
For a regular matroid $N$ and a distinguished triangle $T$ satisfying the
$3$-sum hypotheses, the quantities in
\eqref{eq:seymour-three-sum-five} satisfy
\begin{equation}
    ac=p_1p_2+p_1p_3+p_2p_3.
    \label{eq:regular-triangle-identity}
\end{equation}
\end{proposition}

\begin{proof}
Choose a full-row-rank totally unimodular real representation whose columns for $T$ are
$t_1,t_2,t_3$. After changing column signs, assume $t_3=t_1+t_2$. Since $T$
contains no cocircuit, $E(N)\setminus T$ spans $N$. If $C$ is the matrix of
non-triangle columns and $W$ is the diagonal matrix of their positive
weights, then $K:=CWC^{\mathsf T}$ is positive definite. By Cauchy--Binet,
\begin{equation}
\begin{aligned}
    \det\!\left(K+\sum_{i=1}^3x_it_it_i^{\mathsf T}\right)
    ={}&a
      +(p_2+p_3)x_1+(p_1+p_3)x_2+(p_1+p_2)x_3\\
     &+c(x_1x_2+x_1x_3+x_2x_3).
    \label{eq:triangle-basis-polynomial}
\end{aligned}
\end{equation}
The determinant representation of weighted basis-generating polynomials is
standard for regular matroids~\cite{ChoeOxleySokalWagner2004}; the use of
Pl\"ucker coordinates for distinguished elements is developed in
\cite{Chaiken2006PortedTutte}.

Put $u_i=K^{-1/2}t_i$, so $u_3=u_1+u_2$, and write
$A=\langle u_1,u_1\rangle$, $B=\langle u_2,u_2\rangle$, and
$s=\langle u_1,u_2\rangle$. Since $\det K=a$, the matrix determinant lemma
and comparison with~\eqref{eq:triangle-basis-polynomial} give
\begin{equation}
    \frac{p_1}{a}=B+s,
    \qquad
    \frac{p_2}{a}=A+s,
    \qquad
    \frac{p_3}{a}=-s,
    \qquad
    \frac{c}{a}=AB-s^2.
\end{equation}
In particular, $s=-p_3/a\leq0$, with strict inequality when $p_3>0$;
no sign assumption on $s$ is needed. Therefore
\begin{equation}
    \frac{p_1p_2+p_1p_3+p_2p_3}{a^2}
    =AB-s^2
    =\frac ca,
\end{equation}
and~\eqref{eq:regular-triangle-identity} follows.
\end{proof}

\begin{theorem}
\label{thm:seymour-three-sum}
Let $M=H\oplus_3N$, where $H\cap N=T=\{t_1,t_2,t_3\}$ and the target
belongs to $H\setminus T$. Assume $p_1p_2p_3>0$, and put
$\alpha_i:=a/p_i$ for $i=1,2,3$. Let $K$ be a copy of $M(K_4)$ whose common
triangle with $H$ is $T$, and give its three non-triangle edges weights
$\alpha_1,\alpha_2,\alpha_3$, where spoke $i$ meets the vertex opposite
triangle edge $t_i$. Then
\begin{equation}
    \frac{Z_{\mathrm{in}}(H\oplus_3N;w)}
         {Z_{\mathrm{out}}(H\oplus_3N;w)}
    =
    \frac{Z_{\mathrm{in}}(H\oplus_3K;w)}
         {Z_{\mathrm{out}}(H\oplus_3K;w)}.
    \label{eq:seymour-three-sum-ratio}
\end{equation}
\end{theorem}

\begin{proof}
For the weighted $K_4$ piece, the five quantities in
\eqref{eq:seymour-three-sum-five} are
\begin{equation}
    \left(
        \alpha_1\alpha_2\alpha_3,
        \alpha_2\alpha_3,
        \alpha_1\alpha_3,
        \alpha_1\alpha_2,
        \alpha_1+\alpha_2+\alpha_3
    \right).
    \label{eq:k4-five-sums}
\end{equation}
With $\alpha_i=a/p_i$, the first four entries are
$\lambda(a,p_1,p_2,p_3)$, where
$\lambda:=a^2/(p_1p_2p_3)$. Proposition~\ref{prop:regular-triangle-identity}
gives $\alpha_1+\alpha_2+\alpha_3=\lambda c$, so all five quantities for
$K$ are $\lambda$ times those for $N$. Lemma~\ref{lem:seymour-three-sum-five}
thus multiplies both $Z_{\mathrm{in}}$ and $Z_{\mathrm{out}}$ by the
same factor $\lambda$ when $N$ is replaced by $K$.
\end{proof}

The zero cases have an explicit replacement as well. The hypotheses give
$a>0$ because $N\setminus T$ spans, and $c>0$ because an independent pair
in $T$ extends to a basis. Thus~\eqref{eq:regular-triangle-identity} permits
at most one $p_i$ to vanish. If $p_3=0$, then $p_1p_2=ac>0$.
Contract spoke $3$ of $K_4$ and give the remaining spokes weights
$\alpha_1=a/p_1$ and $\alpha_2=a/p_2$. Its five class sums are
\[
 (\alpha_1\alpha_2,\alpha_2,\alpha_1,0,1)
 =\frac{a}{p_1p_2}(a,p_1,p_2,0,c).
\]
The same linearity proof preserves the ratio. Permuting the indices covers
the other zero cases. Hence every such component has a finite weighted
graphic replacement, possibly with one spoke contracted.

\subsection{Optimization on a decomposition}
\label{subsec:seymour-fixed-weight}

Fix a positive weight assignment $w$. For each input $z$, choose a
$\{1,2,3\}$-decomposition of the regular minor $M_z$ from
Theorem~\ref{thm:regular-matroid-energy}. Apply
Proposition~\ref{prop:seymour-one-sum}, Theorem~\ref{thm:seymour-two-sum},
and Theorem~\ref{thm:seymour-three-sum} successively to summands not
containing $\tau_P$. If $(\widetilde M_z(w),\widetilde w_z(w))$ is the
result, define
\begin{equation}
    \mathscr R_z(w)
    :=
    \frac{Z_{\mathrm{in}}(\widetilde M_z(w);\widetilde w_z(w))}
         {Z_{\mathrm{out}}(\widetilde M_z(w);\widetilde w_z(w))}.
    \label{eq:seymour-reduced-ratio}
\end{equation}
Each reduction preserves the target-basis ratio, so
\begin{equation}
    \mathscr R_z(w)
    =
    \frac{Z_{\mathrm{in}}(M_z;w)}{Z_{\mathrm{out}}(M_z;w)}.
    \label{eq:seymour-ratio-invariant}
\end{equation}

\begin{theorem}
\label{thm:seymour-linear-cost}
Let $P$ be a regular matroid span program computing a nonconstant function.
Then
\begin{equation}
    \operatorname{lincost}(P)
    =
    \inf_{w>0}
    \left(
        \max_{f(x)=1}\mathscr R_x(w)
        \max_{f(y)=0}\mathscr R_y(w)^{-1}
    \right)^{1/2}.
    \label{eq:seymour-linear-cost}
\end{equation}
\end{theorem}

\begin{proof}
By~\eqref{eq:seymour-ratio-invariant}, the quantities
$\mathscr R_x(w)$ and $\mathscr R_y(w)^{-1}$ are exactly the reconstruction
and separation energies in Theorem~\ref{thm:regular-matroid-energy}.
Substituting them into Corollary~\ref{cor:regular-linear-cost} gives
\eqref{eq:seymour-linear-cost}.
\end{proof}

The optimization variables are the same as in
Section~\ref{sec:regular-matroid-optimization}: one positive weight is
assigned to each non-target program element, and the same assignment is used
for every input. The decomposition changes how the objective is evaluated.
A $1$-sum factor not containing the target cancels, a $2$-sum component is replaced by the
single weight $P_0/P_1$, and a $3$-sum component is replaced by three weights
on $M(K_4)$. Thus the fixed-program optimization from Section~\ref{sec:regular-matroid-optimization} can be
evaluated through the smaller matroids in a Seymour decomposition without
changing the shared weight assignment being optimized.

\section{The \texorpdfstring{$R_{10}$}{R10} Summand}
\label{sec:r10-summand}

Theorem~\ref{thm:seymour-linear-cost} reduces the fixed-weight objective of a
regular program to target-basis ratios on the summands of a
$\{1,2,3\}$-decomposition. Graphic and cographic summands retain the
structure of graphs and their duals. The only other leaves allowed by
Seymour's theorem are copies of $R_{10}$.

Unlike the graphic and cographic families, $R_{10}$ is a fixed ten-element
matroid. Its contribution can be described by finitely many
weighted basis polynomials. Its deletion and contraction structure is equally
restrictive: a single deletion gives $M(K_{3,3})$ and a single contraction
gives its dual~\cite[Lem.~6.6.9]{Oxley2011}. Thus the exceptional part of the
fixed-weight calculation is confined to copies of $R_{10}$ that remain intact.

\subsection{Attachments in a decomposition}
\label{subsec:r10-attachments}

\begin{proposition}
\label{prop:r10-no-three-sum}
Let $N\cong R_{10}$ be a non-root leaf in a $\{1,2,3\}$-decomposition.  The parent of
$N$ is not a $3$-sum.  Hence, if $N$ lies in the same connected component as
the target, it meets the rest of the decomposition through a $2$-sum.
\end{proposition}

\begin{proof}
A $3$-sum requires the two summands to share a three-element circuit by
Definition~\ref{def:regular-matroid-sums}. Fix a three-element set
$C\subseteq E(R_{10})$ and choose $e\notin C$. If $C$ were a circuit of
$R_{10}$, then it would remain a circuit of $R_{10}\setminus e$. By
Oxley's description of the single-element minors,
$R_{10}\setminus e\cong M(K_{3,3})$~\cite[Lem.~6.6.9]{Oxley2011}, but
$K_{3,3}$ has no three-edge cycle. Hence $R_{10}$ has no three-element
circuit. Thus $N$ can occur at its parent only in a $1$- or $2$-sum. A
$1$-sum separates connected components, so a copy of $R_{10}$ lying in the
target component is joined to the rest by a $2$-sum.
\end{proof}

Thus a copy of $R_{10}$ not containing the target enters the fixed-weight calculation
through the ratio $P_0/P_1$ already defined in
\eqref{eq:seymour-two-sum-P}.  No additional $3$-sum formula is required for
an undeleted copy of $R_{10}$.

\subsection{Weighted basis sums}
\label{subsec:r10-weighted-bases}

Fix $e\in E(R_{10})$.  For weights $w$ on $E(R_{10})\setminus\{e\}$, write
\begin{equation}
    P_0^e(w):=P_0(R_{10},e;w),
    \qquad
    P_1^e(w):=P_1(R_{10},e;w).
    \label{eq:r10-P01}
\end{equation}

\begin{proposition}
\label{prop:r10-basis-polynomials}
For every $e\in E(R_{10})$, the matroid $R_{10}$ has $162$ bases, of which
$81$ contain $e$ and $81$ omit $e$.  Consequently $P_0^e$ and $P_1^e$ are
square-free polynomials with $81$ monomials of coefficient one, of degrees
five and four respectively.  At unit weights,
\begin{equation}
    P_0^e(\mathbf 1)=P_1^e(\mathbf 1)=81.
    \label{eq:r10-unit-P01}
\end{equation}
\end{proposition}

\begin{proof}
Appendix~\ref{app:r10-atlas-geometry} records an exact enumeration of the
$162$ bases and the transitive action of $\operatorname{Aut}(R_{10})$ on the
ten elements.  Since every basis has five elements, there are
$162\cdot5=810$ element--basis incidences.  Transitivity gives $81$
incidences at each element, so exactly $81$ bases contain $e$ and the other
$81$ omit it.  The assertions about $P_0^e$ and $P_1^e$ then follow from
\eqref{eq:seymour-two-sum-P}.
\end{proof}

If $R_{10}$ is a side of a $2$-sum that does not contain the target at $e$, then
Theorem~\ref{thm:seymour-two-sum} replaces it by the common element with
weight
\begin{equation}
    m_e(w)=\frac{P_0^e(w)}{P_1^e(w)}.
    \label{eq:r10-effective-weight}
\end{equation}
If instead $e$ is the target of the copy of $R_{10}$, its target-basis ratio
is $P_1^e/P_0^e$. Thus an intact copy contributes no new kind of quantity to
the Section~\ref{sec:seymour-calculus} calculation: it contributes the same
weighted basis ratio already required by a $2$-sum. The finite size of
$R_{10}$ makes the numerator and denominator explicit polynomials.

\subsection{Symmetry}
\label{subsec:r10-symmetry}

\begin{proposition}
\label{prop:r10-symmetry}
The automorphism group of $R_{10}$ has order $720$ and acts transitively on
$E(R_{10})$.  If $\varphi\in\operatorname{Aut}(R_{10})$ and
$(\varphi_*w)_{\varphi(f)}:=w_f$, then
\begin{equation}
    P_b^{\varphi(e)}(\varphi_*w)=P_b^e(w)
    \qquad (b\in\{0,1\}).
    \label{eq:r10-P01-equivariance}
\end{equation}
\end{proposition}

\begin{proof}
Transitivity is established in the proof of
Oxley~\cite[Lem.~6.6.9]{Oxley2011}. The exact enumeration in
Appendix~\ref{app:r10-atlas-geometry} gives the group order $720$. An
automorphism sends bases to bases and preserves whether a basis contains the
distinguished element. Relabeling the monomials in the two sums thus
gives~\eqref{eq:r10-P01-equivariance}.
\end{proof}

In particular, the calculation at an arbitrary $2$-sum element is obtained
from the calculation at element $0$ by permuting the weights.  The sparse
polynomials for this choice are stored in
\path{anc/r10_atlas/r10_response_atlas.json}.\footnote{The \texttt{anc/}
directory is available from the ``Ancillary files'' link on this paper's arXiv
abstract page.}

\subsection{Deletion and contraction}
\label{subsec:r10-proper-minors}

\begin{proposition}
\label{prop:r10-proper-minors}
For every $e\in E(R_{10})$,
\begin{equation}
    R_{10}\setminus e\cong M(K_{3,3}),
    \qquad
    R_{10}/e\cong M^*(K_{3,3}).
    \label{eq:r10-single-element-minors}
\end{equation}
Every proper minor of $R_{10}$ is graphic or cographic.  Every minor with at
most eight elements is both graphic and cographic.
\end{proposition}

\begin{proof}
The two single-element identities are standard
\cite[Lem.~6.6.9]{Oxley2011}. After the first deletion or contraction, any
further minor is a minor of $M(K_{3,3})$ or of its dual. This proves
that every proper minor is graphic or cographic.

The graph $K_{3,3}$ is minor-minimal nonplanar, so every proper minor of
$M(K_{3,3})$ is represented by a planar graph. A graphic matroid represented
by a planar graph is also cographic, and dually for a proper minor of
$M^*(K_{3,3})$. Thus every minor of $R_{10}$ with at most eight elements is
both graphic and cographic.
\end{proof}

Thus deletion or contraction of any element of a copy of $R_{10}$ removes
the exceptional case. In particular, the zero and infinite limits in the
$2$-sum formula of Section~\ref{subsec:seymour-two-sum} move the corresponding
part of the calculation into the graphic or cographic classes. After a
second deletion or contraction, both descriptions are available.

\subsection{Finite classification}
\label{subsec:r10-finite-classification}

The automorphisms of $R_{10}$ also reduce the number of distinct ways in
which other $2$-sum reductions can meet a fixed copy of $R_{10}$.

\begin{proposition}
\label{prop:r10-subset-orbits}
The action of $\operatorname{Aut}(R_{10})$ on subsets of $E(R_{10})$ has
$19$ orbits.  For fixed $e\in E(R_{10})$, the stabilizer of $e$ has $26$
orbits on subsets of $E(R_{10})\setminus\{e\}$.
\end{proposition}

\begin{proof}
The orbit counts are obtained by exact enumeration of the automorphism group
from the representation in Figure~\ref{fig:r10-k33}.  The representatives,
stabilizer orders, and dual pairs for the $26$ orbits with $e=0$ are listed
in Table~\ref{tab:r10-marked-configurations}.  The ancillary verifier in
\path{anc/r10_atlas/} reconstructs the group and the orbit partitions.
\end{proof}

Suppose $e$ joins the copy of $R_{10}$ to the part of the decomposition that
contains the target.  After the other attached components have been reduced
by Theorem~\ref{thm:seymour-two-sum}, each of their common elements carries
the resulting positive weight.  Substituting those weights, together with
the ordinary element weights, into $P_0^e$ and $P_1^e$ gives the weight
\eqref{eq:r10-effective-weight} supplied at $e$.  Proposition
~\ref{prop:r10-subset-orbits} identifies the distinct subsets of elements on
which such substitutions occur up to automorphism; it does not alter the
polynomials or the weight variables themselves.

\subsection{Fixed-weight regular programs}
\label{subsec:r10-fixed-weight-programs}

\begin{theorem}
\label{thm:r10-fixed-weight-evaluation}
Let $P$ be a regular matroid span program, let $w>0$ be a weight assignment,
and let $z$ be an input.  In the computation of $\mathscr R_z(w)$ from
\eqref{eq:seymour-reduced-ratio}, every copy of $R_{10}$ that is not changed
by deletion or contraction contributes either
\begin{equation}
    \frac{P_0^e(w')}{P_1^e(w')}
    \label{eq:r10-parent-ratio}
\end{equation}
for a $2$-sum element $e$, or its reciprocal when $e$ is the target.  Here
$w'$ consists of the original element weights together with the weights
produced by reductions of the attached $2$-sum components.  Every copy of
$R_{10}$ changed by a proper deletion or contraction is graphic or
cographic.
\end{theorem}

\begin{proof}
Proposition~\ref{prop:r10-no-three-sum} excludes a $3$-sum at an unchanged
copy of $R_{10}$.  Theorem~\ref{thm:seymour-two-sum} gives
\eqref{eq:r10-parent-ratio} for a $2$-sum whose $R_{10}$ side does not contain the target, while the definition
of $Z_{\mathrm{in}}/Z_{\mathrm{out}}$ gives the reciprocal when the
specified element is the target.  Proposition~\ref{prop:r10-proper-minors}
handles every proper deletion or contraction.
\end{proof}

Combining Theorem~\ref{thm:r10-fixed-weight-evaluation} with
Theorem~\ref{thm:seymour-linear-cost} reduces the fixed-weight objective for
a regular program to weighted basis calculations on graphic and cographic
matroids together with the two $81$-term polynomials $P_0^e$ and $P_1^e$ for
an intact copy of $R_{10}$. The full supports of those polynomials and the
exact finite orbit data are recorded in Appendix~\ref{app:r10-atlas} and
\path{anc/r10_atlas/}.

The exceptional contribution cannot proliferate under the input minors used
in the reconstruction and separation energies. If an $R_{10}$ summand remains
intact, it contributes one of the finite polynomial ratios above. Once an
input restriction, an input contraction, or a deletion--contraction limit of
a $2$-sum changes that summand, Proposition~\ref{prop:r10-proper-minors}
places the resulting matroid in the graphic or cographic classes. The
remaining optimization is consequently concentrated in weighted calculations on
graphic and cographic summands, with intact copies of $R_{10}$ contributing a
finite collection of exceptional factors.

\part{Applications}

The examples compare the source-matroid adversary calculation with explicit
regular program matroids. Grover search recovers the familiar $\sqrt n$
value using a rank-one program matroid. For pointed $R_{10}$ closure, the
function is defined by $R_{10}$ but the targeted query matroid is uniform,
and a regular program different from the defining matroid is cheaper than
the natural $R_{10}$ program. Recursive composition propagates these finite
calculations by exact adversary composition and the $2$-sum identities of
Part~III.

\section{Grover Search}
\label{sec:grover-example}
\label{subsec:or-source-routing-example}

Let $D_n:=\{0,1\}^n$ and $f_n:=\operatorname{OR}_n$. For
$i\in[n]$ and $\epsilon\in\{0,1\}$, let
$B_i^\epsilon:=\{x\in D_n:x_i=\epsilon\}$, let $Q_i$ have answer partition
$\llbracket Q_i\rrbracket=\{B_i^0,B_i^1\}$, and put
$\mathcal Q_n:=\{Q_i:i\in[n]\}$. The common refinement of these partitions is
$\bot_{D_n}$, so $\mathcal Q_n$ is complete. Every coordinate query is
essential: if $Q_i$ is omitted, then $0^n$ and the string supported only on
$i$ remain indistinguishable. Proposition~\ref{prop:coordinate-target-matroid}
gives $M_{f_n}\cong U_{n,n+1}$.

Grover's algorithm gives $Q_{\mathcal Q_n}(f_n)=O(\sqrt n)$, and the
general adversary lower bound gives the matching $\Omega(\sqrt n)$ bound
~\cite{Grover1996,HoyerLeeSpalek2007}. In the present framework, the exact
adversary value and an optimal regular program follow independently from the
source and program matroids. Since $0^n$ is the unique negative input, the
source calculation is also the simplest instance of
Theorem~\ref{thm:singleton-fiber-scalarization}.

\begin{proposition}
\label{prop:grover-source-value}
For the query set $\mathcal Q_n$,
$\operatorname{Adv}^{\pm}_{\mathcal Q_n}(f_n)=\sqrt n$.
\end{proposition}

\begin{proof}
Let $0:=0^n$. For a positive input $x$, its agreement flat with $0$ is
$F_{x0}=\{Q_i:x_i=0\}$. Hence the realized agreement flats are the proper
subsets of $\mathcal Q_n$, and Theorem~\ref{thm:flat-adversary-decomposition}
gives
\begin{equation}
    \operatorname{Adv}^{\pm}_{\mathcal Q_n}(f_n)^2
    =
    \max\left\{
        \sum_{F\subsetneq\mathcal Q_n} a_F^2:
        \sum_{F:Q_i\notin F} a_F^2\leq1\text{ for every }i\in[n]
    \right\}.
\end{equation}
Putting $b_F:=a_F^2$ and summing the $n$ constraints gives
$\sum_F(n-|F|)b_F\leq n$. Since every $F$ is proper,
$\sum_F b_F\leq n$. Equality is attained by setting
$b_{\mathcal Q_n\setminus\{Q_i\}}=1$ for every $i$ and all other $b_F$ to
zero.
\end{proof}

The following program has, up to column scaling, the vector matroid of the
standard one-dimensional weighted OR span program of Reichardt and
{\v S}palek~\cite{ReichardtSpalek2008}.

\begin{proposition}
\label{prop:grover-program-value}
There is a regular $f_n$-matroid span program $P_n$ with
$\operatorname{lincost}(P_n)=\sqrt n$.
\end{proposition}

\begin{proof}
Let $E_{P_n}=\{\tau,e_1,\ldots,e_n\}$ and take
$M_{P_n}=U_{1,n+1}$, with $\lambda_{P_n}(e_i)=(Q_i,B_i^1)$. The available
set on $x$ is $\{e_i:x_i=1\}$, which spans $\tau$ exactly when
$f_n(x)=1$.

The program matroid is regular, so
Proposition~\ref{prop:regular-representations-are-weights} reduces the real
representation optimization to positive element weights. For weights
$w_i>0$, Theorem~\ref{thm:regular-matroid-energy} gives reconstruction energy
$1/\sum_{i\in S}w_i$ on a positive support $S$ and separation energy
$\sum_iw_i$ on $0^n$. Thus the largest reconstruction energy occurs on a
singleton support, and
\begin{equation}
    \mathcal E(\phi_w)^2
    =\frac{\sum_iw_i}{\min_iw_i}\geq n.
\end{equation}
Uniform weights attain equality, so
$\operatorname{lincost}(P_n)=\sqrt n$.
\end{proof}

Propositions~\ref{prop:grover-source-value} and~\ref{prop:grover-program-value},
together with Proposition~\ref{prop:regular-optimality-certificate}, give
\begin{equation}
    \operatorname{MSP}(f_n,\mathfrak M_{f_n})
    =\operatorname{Adv}^{\pm}_{\mathcal Q_n}(f_n)
    =\sqrt n.
\end{equation}
For $n\geq2$, the targeted query matroid and the program matroid have
different ranks:
$M_{f_n}\cong U_{n,n+1}$ while $M_{P_n}\cong U_{1,n+1}$. The source matroid
$\mathfrak M_{f_n}$ contains the targeted query matroid together with the
pairwise separation data, whereas $M_{P_n}$ supplies a different dependence
structure that computes the same function at optimal cost.

\section{Pointed \texorpdfstring{$R_{10}$}{R10} Closure}
\label{sec:r10-pointed}
\label{subsec:r10-pointed}

Choose element $0$ of $R_{10}$ from Definition~\ref{def:r10} as the target
$\tau$, and write $E^\circ:=E(R_{10})\setminus\{\tau\}$. By
Proposition~\ref{prop:r10-proper-minors},
$R_{10}\setminus\tau\cong M(K_{3,3})$; use the labeling from
Figure~\ref{fig:r10-k33}. The correspondence with the column indices of
Definition~\ref{def:r10}, used in the verification archive, is
\[
\begin{array}{c|rrrrrrrrr}
\text{column}&1&2&3&4&5&6&7&8&9\\ \hline
\text{edge}&e_{00}&e_{11}&e_{01}&e_{12}&e_{21}&e_{22}&e_{10}&e_{02}&e_{20}.
\end{array}
\]
For $x\in\{0,1\}^{E^\circ}$, put
$A(x):=\{e\in E^\circ:x_e=1\}$ and define
\begin{equation}
    f_R(x):=\mathbf 1\{\tau\in\cl_{R_{10}}(A(x))\}.
\end{equation}
For $e\in E^\circ$, let $Q_e$ reveal whether $e\in A(x)$, and write
$\mathcal Q_R:=\{Q_e:e\in E^\circ\}$.

For $i,j\in\{0,1,2\}$, let
$R_i:=\{e_{i0},e_{i1},e_{i2}\}$ be the three edges incident with row
vertex $i$, and let $C_j:=\{e_{0j},e_{1j},e_{2j}\}$ be the three edges
incident with column vertex $j$.

\begin{lemma}
\label{lem:r10-minimal-positive}
\label{lem:r10-boundary-inputs}
The minimal positive inputs for $f_R$ are the six perfect matchings of
$K_{3,3}$ and the nine sets $R_i\cup C_j$. The maximal negative
inputs are the nine sets $(R_i\cup C_j)\setminus\{e_{ij}\}$ and the six sets
$E^\circ\setminus R_i$ and $E^\circ\setminus C_j$.
\end{lemma}

\begin{proof}
A set $A\subseteq E^\circ$ is minimal with
$\tau\in\cl_{R_{10}}(A)$ exactly when $A\cup\{\tau\}$ is a target circuit.
The target circuits of the representation in Definition~\ref{def:r10} give
the six perfect matchings and the nine sets $R_i\cup C_j$. A negative set is maximal
exactly when adjoining any missing element creates one of these target
circuits. Appendix~\ref{app:r10-program-boundary} records the exact
enumeration of all $512$ inputs.
\end{proof}

\begin{lemma}
\label{lem:r10-targeted-query-matroid}
Every coordinate of $f_R$ is essential, and $M_{f_R}\cong U_{9,10}$.
\end{lemma}

\begin{proof}
Fix $e\in E^\circ$. Since $R_{10}$ is connected, there is a circuit $C$
containing $\tau$ and $e$. The inputs supported on
$C\setminus\{\tau\}$ and $C\setminus\{\tau,e\}$ differ only in coordinate
$e$ and have different outputs. Proposition~\ref{prop:coordinate-target-matroid}
gives $M_{f_R}\cong U_{9,10}$.
\end{proof}

Thus $M_{f_R}\setminus\{\tau_f\}\cong U_{9,9}$ is free even though
$f_R$ is defined by closure in $R_{10}$. The source matroid
$\mathfrak M_{f_R}$ still retains the pairwise separation data through the
coloring $\Sigma_{f_R}$.

\begin{theorem}
\label{thm:r10-source-interval}
For $f_R$ relative to the query set $\mathcal Q_R$,
\begin{equation}
\label{eq:r10-adversary-interval}
    3.908755342845
    <
    \operatorname{Adv}^{\pm}_{\mathcal Q_R}(f_R)
    <
    3.9301.
\end{equation}
\end{theorem}

\begin{proof}
Appendix~\ref{app:r10-adversary-certificates} records the exact certificates.
The lower certificate is a rational adversary matrix on $71$ pair orbits;
all nine filtered operator norms are at most one, and an exact Rayleigh
quotient exceeds $3.908755342845$. The upper certificate is an exact repaired
filtered factorization whose positive and negative loads have product below
$3.9301^2$. The verifier in
\path{anc/r10_adversary_reproducibility_archive/} checks both statements in
exact arithmetic.
\end{proof}

Every query in $\mathcal Q_R$ has two response blocks. Hence
Theorem~\ref{thm:binary-msp-adversary-equality} applies and gives
\begin{equation}
    \operatorname{MSP}(f_R,\mathfrak M_{f_R})
    =\operatorname{Adv}^{\pm}_{\mathcal Q_R}(f_R).
\end{equation}
The exact certificates involve only the fixed $512$-input problem.

\subsection{The natural program}
\label{subsec:r10-program-calculation}

Label each non-target element $e$ of $R_{10}$ by the positive response of
$Q_e$, and denote the resulting program by $P_R$. Its available set on input
$x$ is $A(x)$, so it computes $f_R$.

\begin{proposition}
\label{prop:r10-natural-program-cost}
$\operatorname{lincost}(P_R)=5$.
\end{proposition}

\begin{proof}
Regularity reduces the representation optimization to positive element
weights. At unit weights, Theorem~\ref{thm:regular-matroid-energy} and the
boundary families of Lemma~\ref{lem:r10-boundary-inputs} give
$C_+=C_-=5$. Conversely, every non-target element occurs in five of the nine
target circuits with five non-target elements and in five of the nine target
cocircuits with five non-target elements. Hence Theorem~\ref{thm:regular-matroid-energy} gives
\begin{equation}
    C_+(w)\geq\frac{5}{9}\sum_e w_e^{-1},
    \qquad
    C_-(w)\geq\frac{5}{9}\sum_e w_e.
\end{equation}
Cauchy--Schwarz gives $C_+(w)C_-(w)\geq25$.
Appendix~\ref{app:r10-natural-program-verification} records the exact circuit
and cocircuit check.
\end{proof}

Thus the source interval already separates the closure problem from its
natural program matroid:
\begin{equation}
    \operatorname{MSP}(f_R,\mathfrak M_{f_R})<3.9301<5
    =\operatorname{lincost}(P_R).
\end{equation}

\subsection{A second regular program}

Let $H$ have three internally disjoint length-three paths from $s$ to $m$,
one for each $i\in\{0,1,2\}$, whose edges carry the positive response
labels of the queries indexed by $R_i$. Add three internally disjoint
length-three paths from $m$ to $t$, one for each $j\in\{0,1,2\}$, whose
edges carry the positive response labels of the queries indexed by $C_j$.
Add a target edge from $s$ to $t$. Appendix~\ref{app:r10-regular-program-verification}
gives an explicit regular representation obtained by identifying this target
column with the target column of the pointed $R_{10}$ representation and
placing the remaining coordinates of the two representations in complementary
subspaces. Denote the resulting regular program by $P_{\mathrm{reg}}$.
This parallel connection is regular: add a parallel copy of the target to
$R_{10}$, take a $2$-sum at the original target with $M(H)$, and retain the
copy as the target. The graphic component accepts precisely when some full
row and some full column are available, hence on the nine double-star
minimal positive sets and their supersets. The shared-target construction
accepts the union of the two components' positive inputs, which is exactly
$f_R^{-1}(1)$.

\begin{proposition}
\label{prop:r10-regular-route-b}
For the regular representation of $P_{\mathrm{reg}}$ determined by weight
$5/4$ on each of the nine $R_{10}$ elements and weight $1$ on each of the
eighteen graphic elements,
\begin{equation}
    C_+=\frac{12}{5},
    \qquad
    C_-=\frac{29}{4}.
\end{equation}
Consequently
$\operatorname{lincost}(P_{\mathrm{reg}})\leq\sqrt{87/5}$.
\end{proposition}

\begin{proof}
Give the nine $R_{10}$ elements weight $p>0$ and the eighteen graphic
elements weight $1$. Lemma~\ref{lem:r10-boundary-inputs} gives
\begin{equation}
    C_+(p)=\max\left\{\frac3p,\frac{30}{6p+5}\right\},
    \qquad
    C_-(p)=\max\{5p+1,3p+3\}.
\end{equation}
Their product decreases up to $p=5/4$ and increases afterwards. At
$p=5/4$,
\begin{equation}
    C_+=\frac{12}{5},
    \qquad
    C_-=\frac{29}{4},
    \qquad
    C_+C_-=\frac{87}{5}.
\end{equation}
Appendix~\ref{app:r10-regular-program-verification} records an exact matrix
check of the program and these boundary energies.
\end{proof}

Since $\sqrt{87/5}<5$, the explicit regular program $P_{\mathrm{reg}}$ is
strictly cheaper than the natural pointed-$R_{10}$ program. The source
calculation remains unchanged: the improvement comes entirely from the choice
of program matroid.

\section{Recursive \texorpdfstring{$R_{10}$}{R10} Closure}
\label{sec:r10-recursive}

The pointed closure function of Section~\ref{sec:r10-pointed} can be
composed with itself on disjoint input blocks. This keeps the nine-query
closure rule fixed while the number of queried coordinates grows with the
depth.

\begin{definition}
\label{def:recursive-r10-function}
Let $F_1:=f_R$.  For $h\geq1$, identify an input to $F_{h+1}$ with
$x=(x^{(1)},\ldots,x^{(9)})$, where each $x^{(i)}$ is an input to $F_h$, and
define
\begin{equation}
    F_{h+1}(x)
    :=
    f_R\bigl(F_h(x^{(1)}),\ldots,F_h(x^{(9)})\bigr).
\end{equation}
Write $N_h:=9^h$, and let $\mathcal Q_h$ be the coordinate-query set on
the $N_h$ leaf coordinates, as in Subsection~\ref{subsec:coordinate-query-matroids}.
\end{definition}

\subsection{Recursive closure}

Let $(M_1,\tau_1):=(R_{10},\tau)$.  Given $(M_h,\tau_h)$, take a fresh
pointed copy of $R_{10}$ with non-target elements $e_1,\ldots,e_9$ and nine
disjoint copies $(M_h^{(i)},\tau_h^{(i)})$.  For each $i$, form a $2$-sum by
identifying $e_i$ with $\tau_h^{(i)}$.  Denote the resulting pointed matroid
by $(M_{h+1},\tau_{h+1})$.

\begin{proposition}
\label{prop:recursive-r10-closure}
For every $h\geq1$, $M_h$ is regular, satisfies
$|E(M_h)|=1+9^h$ and $r(M_h)=(1+9^h)/2$, and target closure in $M_h$
computes $F_h$ relative to the query set $\mathcal Q_h$. Moreover, $R_{10}$ is
a minor of $M_h$, so $M_h$ is neither graphic nor cographic.
\end{proposition}

\begin{proof}
Regular matroids are closed under $2$-sums~\cite{Oxley2011}. Connectedness is
also preserved here, so all summing elements are nonloops and noncoloops.
The recurrences
$|E(M_{h+1})|=9|E(M_h)|-8$ and $r(M_{h+1})=9r(M_h)-4$, with
$|E(M_1)|=10$ and $r(M_1)=5$, give the stated formulas.

By the circuit description of a $2$-sum~\cite{Oxley2011}, a circuit through
the target either avoids $e_i$ or joins a circuit through $e_i$ in the top
copy of $R_{10}$ to a circuit through $\tau_h^{(i)}$ in the attached copy of
$M_h$. Hence $e_i$ can be supplied to the top copy exactly when the available
elements in the $i$th copy span $\tau_h^{(i)}$, which by induction is
equivalent to $F_h(x^{(i)})=1$. Thus the target is spanned exactly when
$f_R(F_h(x^{(1)}),\ldots,F_h(x^{(9)}))=1$.

For the minor assertion, in each child choose a circuit $D$ through its
summing element and a non-target leaf $u$. Delete elements outside $D$ and
contract $D\setminus\{\tau_h^{(i)},u\}$. This leaves a two-element circuit;
its $2$-sum with the parent merely relabels the parent attachment by $u$.
Doing this in all nine children recovers the root $R_{10}$. The case $h=1$
is immediate. Graphic and cographic matroids are minor-closed, and $R_{10}$
is neither, proving the last assertion.
\end{proof}

These are properties of the natural program matroids $M_h$. The depth-two
matrix construction is checked independently in
Appendix~\ref{app:r10-recursive-boolean}.

\subsection{Sensitivity and randomized queries}

For $x\in\{0,1\}^{N_h}$ and a coordinate $u$, let $x^{\oplus u}$ denote
$x$ with coordinate $u$ flipped, and set
$\operatorname{Sens}(F_h,x):=\{u:F_h(x)\neq F_h(x^{\oplus u})\}$. For
$b\in\{0,1\}$, set
$s_b(F_h):=\max_{F_h(x)=b}|\operatorname{Sens}(F_h,x)|$. The exact base
enumeration in Appendix~\ref{app:r10-recursive-base} gives
$s_0(f_R)=s_1(f_R)=5$.

\begin{proposition}
\label{prop:recursive-r10-sensitivity}
For every $h\geq1$, $s_0(F_h)=s_1(F_h)=5^h$.
\end{proposition}

\begin{proof}
Appendix~\ref{app:r10-recursive-base} gives inputs $x^{(0)}$ and $x^{(1)}$
of $f_R$ with outputs $0$ and $1$, respectively, and with five sensitive
coordinates each. Define $X_1^{(b)}:=x^{(b)}$. Recursively, the $i$th input
block of $X_h^{(b)}$ is $X_{h-1}^{(x_i^{(b)})}$. If
$i\in\operatorname{Sens}(f_R,x^{(b)})$, then every coordinate in
$\operatorname{Sens}(F_{h-1},X_{h-1}^{(x_i^{(b)})})$ is sensitive for
$F_h$ at $X_h^{(b)}$. Induction gives $s_b(F_h)\geq5^h$.

Conversely, if a coordinate in the $i$th input block is sensitive for
$F_{h+1}$ at $x$, then $i$ is sensitive for $f_R$ at
$(F_h(x^{(1)}),\ldots,F_h(x^{(9)}))$, and that coordinate is sensitive for
$F_h$ at $x^{(i)}$. Since $s_0(f_R)=s_1(f_R)=5$, induction gives
$s_b(F_h)\leq5^h$.
\end{proof}

For a query set $\mathcal Q$ and $0\leq\epsilon<1/2$, write
$R^\epsilon_{\mathcal Q}(f)$ for the minimum, over randomized decision trees
that compute $f$ with error at most $\epsilon$ on every input, of the maximum
expected number of queries over all inputs. We use the standard worst-case
expected-query convention~\cite{BuhrmanDeWolf2002}. Thus
$R^{1/3}_{\mathcal Q_h}(F_h)$ denotes bounded-error randomized query
complexity relative to the query set $\mathcal Q_h$.

\begin{theorem}
\label{thm:recursive-r10-randomized-lower}
For every $h\geq1$ and $0\leq\epsilon<1/2$,
\begin{equation}
    R^\epsilon_{\mathcal Q_h}(F_h)
    \geq
    (1-2\epsilon)5^h.
\end{equation}
In particular,
\begin{equation}
    R^{1/3}_{\mathcal Q_h}(F_h)
    \geq
    \frac13 5^h
    =
    \frac13 N_h^{\log_9 5}.
\end{equation}
\end{theorem}

\begin{proof}
Fix an input $x$ and a sensitive coordinate $u$.  Couple the runs on $x$ and
$x^{\oplus u}$ using the same internal randomness.  Until $u$ is queried the
two transcripts are identical.  If $u$ is never queried, the two outputs are
equal although the correct function values differ, so at least one
of the two coupled runs errs.  Hence $\Pr_x[\text{$u$ is queried}]\geq1-2\epsilon$. Summing over the
sensitive coordinates of $x$ gives an expected number of
queries at least $(1-2\epsilon)|\operatorname{Sens}(F_h,x)|$.  Apply
Proposition~\ref{prop:recursive-r10-sensitivity} to either of the inputs
$X_h^{(0)}$ and $X_h^{(1)}$.
\end{proof}

Sensitivity $5^h$ is attained on inputs in both $F_h^{-1}(0)$ and $F_h^{-1}(1)$.
In particular, the recursive family does not have the one-sided sensitive
structure of $\operatorname{OR}_n$.

\subsection{The natural recursive program}

Label the $9^h$ non-target elements of $M_h$ by the positive responses of the
corresponding coordinate queries in $\mathcal Q_h$, and denote the resulting
regular matroid span program by $P_h^{\mathrm{nat}}$.  By
Proposition~\ref{prop:recursive-r10-closure}, $P_h^{\mathrm{nat}}$ computes
$F_h$.

\begin{proposition}
\label{prop:recursive-r10-natural-cost}
For every $h\geq1$, $\operatorname{lincost}(P_h^{\mathrm{nat}})=5^h$.
\end{proposition}

\begin{proof}
At depth one, pointed $R_{10}$ has target circuits and target cocircuits with
five non-target elements, and no target circuit or cocircuit contains more
than five non-target elements.  Under a $2$-sum, a target circuit is obtained
from a target circuit of the parent by replacing each selected common element
by a target circuit of the corresponding child.  The dual statement holds
for target cocircuits.  Induction shows that every target circuit
and target cocircuit of $M_h$ contains at most $5^h$ non-target leaves, and
that both bounds are attained.  With unit weights,
Theorem~\ref{thm:regular-matroid-energy} consequently gives a representation
of cost at most $5^h$.

For the reverse inequality, retain at depth one the nine target circuits with
five non-target elements and the nine target cocircuits with five non-target
elements.  Recursively substituting only members of these families gives
families $\mathcal C_h$ and $\mathcal D_h$ in which every member contains
exactly $5^h$ leaves.  Each leaf occurs in the same fraction
$5^h/9^h$ of the members of either family.

Give the leaves arbitrary positive weights $w_e$, and let $C_+(w)$ and $C_-(w)$ denote the two fixed-weight maxima in Corollary~\ref{cor:regular-linear-cost}.  On the minimal positive input determined by $C\in\mathcal C_h$, total
unimodularity gives reconstruction energy
$\sum_{e\in C\setminus\{\tau_h\}}w_e^{-1}$. On the hyperplane complementary
to a member of $\mathcal D_h$, the separation energy is the sum of the
weights on its non-target elements. Averaging over the two families gives
\begin{equation}
    C_+(w)\geq\frac{5^h}{9^h}\sum_e w_e^{-1},
    \qquad
    C_-(w)\geq\frac{5^h}{9^h}\sum_e w_e.
\end{equation}
Cauchy--Schwarz yields $C_+(w)C_-(w)\geq25^h$.  Corollary
~\ref{cor:regular-linear-cost} and the unit-weight upper bound give the
claim.
\end{proof}

Thus the natural recursive program has the same exponential base $5$ as the
randomized lower bound above.  The improvement already present at depth one
requires a different program matroid.

\subsection{A second recursive regular program}

Let $P_1^{\mathrm{reg}}:=P_{\mathrm{reg}}$ from
Section~\ref{sec:r10-pointed}.  Each of its nine query labels occurs once among the $R_{10}$ elements and
twice among the graphic elements.  To construct $P_{h+1}^{\mathrm{reg}}$,
attach a distinct copy of $P_h^{\mathrm{reg}}$ by a $2$-sum at each of the
$27$ non-target elements of a fresh copy of $P_{\mathrm{reg}}$.  The three copies attached at elements carrying the same query label receive
identical query labels on corresponding coordinates; distinct query labels of
the fresh copy use disjoint blocks of input coordinates. The resulting
program has $9^{h+1}$ distinct queries even though its program matroid has
$1+27^{h+1}$ elements.

Use the weight assignment of Proposition~\ref{prop:r10-regular-route-b} at
every level.  Thus an outer $R_{10}$ element has weight $p=5/4$ and an outer
graphic element has weight $1$.  When a child program is attached at an element of weight $\omega$, multiply
all non-target weights in that child copy by $\omega$.  Let $\phi_h^{\mathrm{reg}}$ denote the resulting weighted representation, and
set $C_+(h):=\max_{F_h(x)=1}c_+(\phi_h^{\mathrm{reg}},x)$ and
$C_-(h):=\max_{F_h(x)=0}c_-(\phi_h^{\mathrm{reg}},x)$.

\begin{theorem}
\label{thm:recursive-r10-regular-cost}
For every $h\geq1$, $P_h^{\mathrm{reg}}$ is a regular matroid span program
computing $F_h$, and
\begin{equation}
    C_+(h)=\left(\frac{12}{5}\right)^h,
    \qquad
    C_-(h)=\left(\frac{29}{4}\right)^h.
\end{equation}
Consequently,
$\operatorname{lincost}(P_h^{\mathrm{reg}})\leq(\sqrt{87/5})^h$.
\end{theorem}

\begin{proof}
Regularity follows from closure of regular matroids under $2$-sums.  The
three elements carrying one query label are driven by the same child input,
so each is present exactly when that child program spans its target.
The outer availability pattern is thus
$(F_h(x^{(1)}),\ldots,F_h(x^{(9)}))$, and the outer copy of
$P_{\mathrm{reg}}$ computes $f_R$.  Hence $P_{h+1}^{\mathrm{reg}}$ computes
$F_{h+1}$.

Put $A:=12/5$ and $B:=29/4$. Suppose first that $F_{h+1}(x)=1$. If
$F_h(x^{(i)})=1$, then \eqref{eq:two-sum-positive-input-value} gives
$m_i^+=1/c_+(\phi_h^{\mathrm{reg}},x^{(i)})\geq A^{-h}$; if
$F_h(x^{(i)})=0$, it gives $m_i^+=0$. After the prescribed scaling, the three
elements carrying query label $i$ have finite weights $p m_i^+$, $m_i^+$,
and $m_i^+$. Increasing available weights cannot increase reconstruction
energy, so homogeneity and Proposition~\ref{prop:r10-regular-route-b} give
$c_+(\phi_{h+1}^{\mathrm{reg}},x)\leq A^{h+1}$.

If $F_{h+1}(x)=0$, then \eqref{eq:two-sum-negative-input-value} contracts the
$i$th attachment when $F_h(x^{(i)})=1$, while for $F_h(x^{(i)})=0$ it gives
$m_i^-=c_-(\phi_h^{\mathrm{reg}},x^{(i)})\leq B^h$. Separation energy is
nondecreasing in the finite noncontracted weights, so the same homogeneity
argument gives $c_-(\phi_{h+1}^{\mathrm{reg}},x)\leq B^{h+1}$.

The inputs $X_{h+1}^{(1)}$ and $X_{h+1}^{(0)}$ from the proof of
Proposition~\ref{prop:recursive-r10-sensitivity} attain both inequalities.
For $X_{h+1}^{(1)}$, every finite value $m_i^+$ equals $A^{-h}$; for
$X_{h+1}^{(0)}$, every finite value $m_i^-$ equals $B^h$.
Proposition~\ref{prop:r10-regular-route-b} and homogeneity together give
equality. The base case is the same proposition. Hence
$\mathcal E(\phi_h^{\mathrm{reg}})=\sqrt{C_+(h)C_-(h)}=(\sqrt{87/5})^h$,
which is an admissible representation of $P_h^{\mathrm{reg}}$.
\end{proof}

The exact depth-two values $C_+(2)=144/25$ and $C_-(2)=841/16$ are
independently checked in Appendix~\ref{app:r10-recursive-program}. The
standard conversion of represented span programs to bounded-error quantum
query algorithms~\cite{Reichardt2011,Reichardt2014Equivalent} gives
$Q_{\mathcal Q_h}(F_h)=O((\sqrt{87/5})^h)=O(N_h^{0.650017807848\ldots})$.
Together with Theorem~\ref{thm:recursive-r10-randomized-lower}, this is a
polynomial separation between bounded-error quantum and randomized query
complexity for the family $F_h$.

\subsection{Quantum query complexity}

Let $a:=\operatorname{Adv}^{\pm}_{\mathcal Q_R}(f_R)$. Reichardt's exact
composition theorem for the general adversary bound
\cite[Theorem~7.2]{Reichardt2009} applies to the disjoint child blocks in
Definition~\ref{def:recursive-r10-function}; its lower-bound direction is the
composition inequality of H{\o}yer, Lee, and {\v S}palek
\cite{HoyerLeeSpalek2007}. Hence
\begin{equation}
    \operatorname{Adv}^{\pm}_{\mathcal Q_h}(F_h)=a^h.
\end{equation}
Tightness of the general adversary bound
\cite{Reichardt2011,Reichardt2014Equivalent} then gives
$Q_{\mathcal Q_h}(F_h)=\Theta(a^h)$. By
Theorem~\ref{thm:r10-source-interval}, $3.908755342845<a<3.9301$. Thus, for
$\alpha:=\log_9 a$,
$Q_{\mathcal Q_h}(F_h)=\Theta(N_h^\alpha)$ with
$0.620427702631\ldots<\alpha<0.622906226777\ldots$.

\begin{theorem}
\label{thm:recursive-r10-comparison}
For the recursive pointed-$R_{10}$ closure family of
Definition~\ref{def:recursive-r10-function}, with
$a=\operatorname{Adv}^{\pm}_{\mathcal Q_R}(f_R)$,
\begin{equation}
\begin{aligned}
R^{1/3}_{\mathcal Q_h}(F_h)&\geq\tfrac13 5^h,\\
\operatorname{lincost}(P_h^{\mathrm{nat}})&=5^h,\\
\operatorname{lincost}(P_h^{\mathrm{reg}})&\leq(\sqrt{87/5})^h,\\
Q_{\mathcal Q_h}(F_h)&=\Theta(a^h),
\qquad 3.908755342845<a<3.9301.
\end{aligned}
\end{equation}
\end{theorem}

\begin{proof}
Combine Theorem~\ref{thm:recursive-r10-randomized-lower},
Proposition~\ref{prop:recursive-r10-natural-cost},
Theorem~\ref{thm:recursive-r10-regular-cost}, the exact adversary composition
above, and Theorem~\ref{thm:r10-source-interval}.
\end{proof}

The semidefinite optimization used for this family is confined to the fixed
nine-bit problem in Theorem~\ref{thm:r10-source-interval}.  Its exact lower
and upper certificates are recorded in
Appendix~\ref{app:r10-adversary-certificates}; the source-matroid agreement
flats organize the finite adversary upper-certificate calculation. The
dependence on $h$ is obtained from exact adversary composition and the
$2$-sum identities of
Sections~\ref{sec:seymour-calculus}--\ref{sec:r10-summand}.  No semidefinite program whose dimension grows with
$N_h$ is required.

The natural program has optimized cost $5^h$, while the displayed
representation of the alternative program costs $(\sqrt{87/5})^h$.
The finite improvement
$5\mapsto\sqrt{87/5}$ from Section~\ref{sec:r10-pointed} propagates
to a polynomial asymptotic improvement, while the adversary composition
identifies a still smaller quantum-query exponent.

Appendix~\ref{app:r10-classical-stress} checks two natural refinements of the
classical lower bound. Fractional block sensitivity is exactly $5^h$ on both
output fibers, while the certified sabotage-complexity base is
$240947/53451<5$. Thus neither measure improves the sensitivity-based
$5^h$ lower-bound scale for this family; strengthening the randomized lower
bound would require a different technique.

\section{Discussion and Outlook}
\label{sec:discussion}

For a fixed regular program matroid, the remaining continuous freedom is now
comparatively explicit. Section~\ref{sec:regular-matroid-optimization}
reduces real representations to positive element weights, while
Sections~\ref{sec:seymour-calculus} and~\ref{sec:r10-summand} decompose the
resulting weighted calculation into graphic and cographic pieces together
with finite $R_{10}$ data. The recursive example of
Section~\ref{sec:r10-recursive} shows that these local calculations can be
propagated through arbitrarily large programs without repeating the
underlying semidefinite optimization.

Related linear-algebraic models expose some of the same structural choices,
although with different objectives. Secret-sharing theory studies monotone
access structures through polymatroid extensions; linear secret-sharing
schemes and monotone span programs allow repeated participant labels; and
linear realizations of a fixed access structure can be optimized for
information rate~\cite{FarrasMartiFarrePadro2012,CramerEtAl2008,
Beimel2011SecretSharingSurvey,TangGaoZhang2013}. In quantum span-program
algorithm design, optimizing numerical weights inside a chosen construction
is likewise standard~\cite{ReichardtSpalek2008,CornelissenMandePatro2025}.
These comparisons motivate keeping the abstract response-labeled program
separate from the optimization over its real representations. Our objective
is the positive/negative span-program cost.

The larger optimization problem remains discrete: one must still choose the
response-labeled program matroid. A natural next step is therefore a global
optimizer over regular program matroids, combining the graphic, cographic,
and $R_{10}$ calculations with the exact $1$-, $2$-, and $3$-sum formulas.
Such an optimizer would have to search over labeled regular matroids and
their decomposition structure as well as over element weights. Seymour
decomposition makes the evaluation of a proposed regular program
compositional; controlling the search over program structures is the
remaining problem.

A related question is when restriction to regular program matroids loses
nothing. The source-matroid adversary calculation gives a program-independent
lower bound, while any regular construction gives an upper bound through
$\operatorname{lincost}$. Matching the two certifies global optimality, not
merely optimality within the regular class. It would be useful to identify
structural properties of a query problem, its source matroid, or its optimal
adversary factorization that imply the existence of an optimal regular
program matroid. The pointed-$R_{10}$ example also leaves open the converse
possibility: a regular program may substantially improve on the natural
matroidal construction without attaining the unrestricted optimum.

Undirected $s$--$t$ connectivity is a particularly natural test case.
The standard incidence span program is graphic
~\cite{KarchmerWigderson1993,BelovsReichardt2012}, and its reconstruction
and separation energies admit the familiar resistance and capacitance
interpretations
~\cite{JefferyKimmel2017,JarretJefferyKimmelPiedrafita2018}. This makes
$s$--$t$ connectivity attractive from two directions. Structurally, one can
ask whether the incidence program, or another graphic program, is optimal
among regular program matroids and whether source-matroid certificates can
prove that optimality directly. Algorithmically, graph families and promise
versions of connectivity provide a setting in which a global regular
optimizer might select a different graphic, cographic, or decomposed
program and thereby produce a new query algorithm.

These questions separate two possible outcomes for $s$--$t$ connectivity.
If the incidence construction is globally optimal within the relevant
class, the framework should expose a structural reason for that fact. If it
is not, the competing regular program matroid gives a concrete alternative
construction to analyze.

\section*{Acknowledgments}

We used AI tools (OpenAI's ChatGPT) to explore proof strategies, check written proofs, develop and test reproducibility code, survey and compare relevant literature, and assist in revising and auditing the manuscript. All mathematical claims, proofs, citations, and final text were reviewed by the authors.

\phantomsection
\addcontentsline{toc}{section}{References}
\begingroup
\small
\setlength{\parskip}{0pt}
\hbadness=10000
\bibliographystyle{plain}
\bibliography{references}
\endgroup

\part{Appendices}
\appendix

\section{Exact Certificates for the Pointed \texorpdfstring{$R_{10}$}{R10} Adversary Bound}
\label{app:r10-adversary-certificates}

This appendix records the conventions and exact checks supporting
Theorem~\ref{thm:r10-source-interval}. The complete coefficient arrays are
stored in \path{anc/r10_adversary_reproducibility_archive/}, where
they can be replayed without an SDP solver or floating-point acceptance
threshold.

\subsection{Source convention and lower certificate}
\label{app:r10-exact-source-certificate}

Index the nine edges of $K_{3,3}$ in row-major order, using the column-to-edge
correspondence in Section~\ref{sec:r10-pointed}, and encode a subset by
its nine-bit mask. Column $0$ of Figure~\ref{fig:r10-k33} is the target; the
remaining columns are identified with the nine edge queries. The verifier
checks by exact rank on all $512$ masks that this closure predicate agrees
with Lemma~\ref{lem:r10-boundary-inputs} and verifies the $72$ source
symmetries generated by row permutations, column permutations, and
transposition.

The lower certificate is a $G$-invariant rational adversary witness on the
restricted $87\times51$ opposite-output input space. It is specified by one
integer coefficient on each of $71$ pair orbits with common denominator
$10^6$. For a representative query, the verifier forms the integer filtered
matrix $F$ and checks
\begin{equation}
    10^{12}I-F^{\mathsf T}F\succ0
\end{equation}
by an exact rational $LDL^{\mathsf T}$ factorization. Source symmetry gives
the remaining eight query inequalities. Exact integer test vectors then
certify a Rayleigh quotient larger than
\begin{equation}
    \frac{781751068569}{200000000000}
    =3.908755342845.
\end{equation}
The machine-readable orbit table and test vectors are stored in the
\texttt{lower/} subdirectory of
\path{anc/r10_adversary_reproducibility_archive/}.

\subsection{Upper routing certificate}
\label{app:r10-exact-routing-certificate}

The upper certificate starts from a dimension-$32$ dyadic filtered
factorization. For each opposite-output pair, its exact residual is assigned
to a fresh orthogonal coordinate on the least query separating that pair.
This pair-specific repair makes every factorization identity exact while
adding the absolute residual to the squared load of each endpoint.

Exact integer summation gives repaired loads
\begin{equation}
    P=\frac{1106245605369758}{2^{48}},
    \qquad
    N=\frac{1106181546001981}{2^{48}}.
\end{equation}
They satisfy
\begin{equation}
    PN<\left(\frac{39301}{10000}\right)^2,
\end{equation}
which gives the upper endpoint $3.9301$ in
Theorem~\ref{thm:r10-source-interval}.

\subsection{Ancillary verification}
\label{app:r10-ancillary-verification}

The directory
\path{anc/r10_adversary_reproducibility_archive/} contains the exact
source convention, the $71$-orbit lower certificate, the dimension-$32$
dyadic base factorization, and a verifier. From that directory, running
\begin{center}
\texttt{python3 verify\_all.py}
\end{center}
checks all nine lower filtered-norm constraints, the exact lower Rayleigh
inequality, all upper residual repairs, and the strict $3.9301$ load-product
inequality. The archive manifest records byte-level hashes for the replay
data.

\section{Exact Verification of the Explicit Pointed \texorpdfstring{$R_{10}$}{R10} Programs}
\label{app:r10-program-verification}

The analytic program-cost arguments are given in
Section~\ref{subsec:r10-pointed}. The verifier in
\path{anc/r10_program_reproducibility/} independently reconstructs the
displayed $R_{10}$ matrix and the regular parallel program using exact
rational arithmetic.

\subsection{Boundary inputs}
\label{app:r10-program-boundary}

Using the row-major edge order, the verifier evaluates the closure predicate
on all $512$ inputs. It confirms that the minimal positives are the six
perfect matchings and the nine sets $R_i\cup C_j$, and that the maximal negatives are
the nine four-edge sets and six six-edge sets in
Lemma~\ref{lem:r10-boundary-inputs}. These boundary families suffice for the
worst-case reconstruction and separation energies by monotonicity under
availability.

\subsection{Natural and regular programs}
\label{app:r10-natural-program-verification}
\label{app:r10-regular-program-verification}

For the natural pointed-$R_{10}$ program, exact circuit and cocircuit
enumeration gives boundary energies $3,5,5,3$ on the four families
$P_3,P_5,N_4,N_6$, respectively, at unit weights. It also verifies the
incidence counts used in the Cauchy--Schwarz lower bound proving
$\operatorname{lincost}(P_R)=5$.

For $P_{\mathrm{reg}}$, the verifier constructs the exact representation
described in Section~\ref{sec:r10-pointed} by identifying the two target
columns and placing the remaining coordinates in complementary subspaces. It
checks the resulting closure rule on all $512$ inputs. At $p=5/4$, the four boundary
energies are
\begin{equation}
\begin{array}{c|cccc}
\text{family} & P_3 & P_5 & N_4 & N_6\\
\hline
\text{energy} & 12/5 & 12/5 & 29/4 & 27/4.
\end{array}
\end{equation}
Thus $C_+=12/5$, $C_-=29/4$, and $C_+C_-=87/5$.

\subsection{Ancillary verification}
\label{app:r10-program-ancillary}

The directory \path{anc/r10_program_reproducibility/} contains the
standard-library Python replay script and manifest. The script uses exact
fractions throughout and contains no floating-point acceptance threshold.

\section{Finite Data for the \texorpdfstring{$R_{10}$}{R10} Decomposition}
\label{app:r10-atlas}

This appendix records the finite data used to treat an $R_{10}$ summand in a
Seymour decomposition. The structural facts about $R_{10}$ itself are
standard; the orbit enumeration and sparse weighted-basis data below were
obtained by exact computation from the representation in
Figure~\ref{fig:r10-k33}. The computation is independent of the pointed
closure application and of its adversary certificates.

\subsection{Geometry and marked configurations}
\label{app:r10-atlas-geometry}

The displayed representation has rank five and $162$ bases. It has fifteen
four-element circuits, fifteen six-element circuits, no triangles, and no
triads. Its automorphism group has order $720$ and acts transitively on the
ten elements. The stabilizer of a marked element has order $72$. For every
element $e$,
\begin{equation}
    R_{10}\setminus e\cong M(K_{3,3}),
    \qquad
    R_{10}/e\cong M^*(K_{3,3}).
    \label{eq:r10-atlas-single-minors}
\end{equation}
Every minor on at most eight elements is both graphic and cographic.

Since $R_{10}$ has no triangle or triad, an undeformed $R_{10}$ summand has
no native $3$-sum attachment. Its nontrivial attachments in a Seymour
decomposition are ordinary $2$-sums at distinct elements. After
fixing element $0$ as an outgoing element or target, subsets of the other
nine elements fall into the following $26$ orbits under its stabilizer.

\begin{center}
\small
\begin{longtable}{c l r r c}
\caption{Marked $R_{10}$ port configurations with marked element $0$.}\label{tab:r10-marked-configurations}\\
\hline
Type & Incoming elements & Orbit size & Stabilizer & Dual \\
\hline
\endfirsthead
\hline
Type & Incoming elements & Orbit size & Stabilizer & Dual \\
\hline
\endhead
P00 & $\varnothing$ & 1 & 72 & P00 \\
P01 & $\{1\}$ & 9 & 8 & P01 \\
P02 & $\{1,2\}$ & 18 & 4 & P03 \\
P03 & $\{1,3\}$ & 18 & 4 & P02 \\
P04 & $\{1,2,3\}$ & 36 & 2 & P05 \\
P05 & $\{1,2,4\}$ & 36 & 2 & P04 \\
P06 & $\{2,3,5\}$ & 6 & 12 & P07 \\
P07 & $\{1,4,5\}$ & 6 & 12 & P06 \\
P08 & $\{1,2,3,4\}$ & 36 & 2 & P08 \\
P09 & $\{1,2,3,5\}$ & 36 & 2 & P10 \\
P10 & $\{1,2,4,5\}$ & 36 & 2 & P09 \\
P11 & $\{1,3,4,6\}$ & 9 & 8 & P12 \\
P12 & $\{2,4,5,6\}$ & 9 & 8 & P11 \\
P13 & $\{1,2,3,4,5\}$ & 36 & 2 & P13 \\
P14 & $\{1,2,3,4,6\}$ & 36 & 2 & P16 \\
P15 & $\{1,2,4,5,6\}$ & 9 & 8 & P17 \\
P16 & $\{2,3,4,5,6\}$ & 36 & 2 & P14 \\
P17 & $\{2,3,4,5,7\}$ & 9 & 8 & P15 \\
P18 & $\{1,2,3,4,5,6\}$ & 36 & 2 & P19 \\
P19 & $\{1,2,3,4,5,7\}$ & 36 & 2 & P18 \\
P20 & $\{1,3,4,5,6,7\}$ & 6 & 12 & P21 \\
P21 & $\{2,3,4,5,6,8\}$ & 6 & 12 & P20 \\
P22 & $\{1,2,3,4,5,6,7\}$ & 18 & 4 & P23 \\
P23 & $\{1,2,3,4,5,6,8\}$ & 18 & 4 & P22 \\
P24 & $\{1,2,3,4,5,6,7,8\}$ & 9 & 8 & P24 \\
P25 & $\{1,2,3,4,5,6,7,8,9\}$ & 1 & 72 & P25 \\
\hline
\end{longtable}
\end{center}

The listed set contains only incoming elements. For an outgoing
configuration, element $0$ is added as the outgoing $2$-sum element; for a
root configuration, element $0$ is retained as the target.

\subsection{Weighted basis data}
\label{app:r10-atlas-response}

Fix the marked element $0$. Write
\begin{equation}
    P_0(x)
    :=
    \sum_{\substack{B\in\mathcal B(R_{10})\\0\notin B}}x^B,
    \qquad
    P_1(x)
    :=
    \sum_{\substack{B\in\mathcal B(R_{10})\\0\in B}}x^{B\setminus\{0\}}.
    \label{eq:r10-atlas-polynomials}
\end{equation}
Each polynomial has $81$ square-free monomials with coefficient one; their
degrees are five and four. For any marked configuration in
Table~\ref{tab:r10-marked-configurations}, substitute the incoming $2$-sum weights
for the corresponding variables and the ordinary element weights for the
remaining variables. The quantity supplied to a parent is
\begin{equation}
    m=\frac{P_0}{P_1},
    \label{eq:r10-atlas-ratio}
\end{equation}
and the target-basis ratio at a root is its reciprocal.

The sparse $81$-term supports need not be printed in the manuscript. They
are stored in
\path{anc/r10_atlas/r10_response_atlas.json}. An independent
determinant check
is obtained from the matrix $A=[a_0\;a_1\;\cdots\;a_9]$ of
Figure~\ref{fig:r10-k33}. With
\begin{equation}
    L(x):=\sum_{f=1}^9 x_f a_fa_f^{\mathsf T},
\end{equation}
Cauchy--Binet and total unimodularity give $P_0=\det L$ and $P_1$ as the
$(0,0)$ principal cofactor in the coordinates of Figure~\ref{fig:r10-k33}. At unit weights,
$P_0=P_1=81$.

\subsection{Deletion and contraction}
\label{app:r10-atlas-minors}

The ratio~\eqref{eq:r10-atlas-ratio} extends to zero and infinite incoming
weights by ordinary matroid minors. A zero incoming weight deletes the
corresponding element and an infinite incoming weight contracts it, after
cancelling common homogeneous factors. If the marked element becomes a loop,
then $m=\infty$; if it becomes a coloop, then $m=0$. When this value is sent
to a parent through a $2$-sum, $m=0$ deletes the parent's common element and
$m=\infty$ contracts it.

The exact automorphism groups give $2{,}488$ unrooted and $4{,}480$ marked
orbits of the four element states ordinary, port, deleted, and contracted.
These are embedding-orbit counts with the original element roles retained,
rather than counts of abstract unlabelled minor isomorphism classes. The
ancillary verifier recomputes both counts exactly by Burnside's lemma.

\subsection{Ancillary data}
\label{app:r10-atlas-ancillary}

The self-contained directory \path{anc/r10_atlas/} contains the
geometry data, response data, generator, exact verifier, and checksum
manifest. Its principal files are
\begin{center}
\small
\texttt{r10\_port\_geometry.json}\\
\texttt{r10\_response\_atlas.json}\\
\texttt{build\_atlas.py}\\
\texttt{verify\_atlas.py}\\
\texttt{verification\_results.json}\\
\texttt{INTRINSIC\_ATLAS\_FREEZE.json}\\
\texttt{R10\_ATLAS\_MANIFEST.sha256}.
\end{center}
The verifier reconstructs ranks, bases, circuits, cocircuits, automorphisms,
duality, the marked orbit partition, the deletion/contraction orbit counts,
and the sparse basis-polynomial supports from the displayed $R_{10}$ matrix.
The marked table above is also recorded in
\texttt{r10\_port\_geometry.json}.

\section{Exact Verification for Recursive Pointed \texorpdfstring{$R_{10}$}{R10} Closure}
\label{app:r10-recursive-verification}

This appendix records finite checks for the recursive pointed-$R_{10}$
construction. The recursive statements themselves are proved algebraically;
the scripts check the base gate, the first nontrivial $2$-sum level, and the
exact cost recurrence used in the explicit regular program.

\subsection{Base closure data}
\label{app:r10-recursive-base}

Taking element $0$ as the target, exact rank evaluation on all $512$
availability sets gives $256$ positive and $256$ negative inputs. There are
six minimal positive sets of size three and nine of size five, and six
maximal negative sets of size six and nine of size four. The inputs
\begin{equation}
    x^{(0)}=\mathbf 1_{\{4,5,8,9\}},
    \qquad
    x^{(1)}=\mathbf 1_{\{4,5,6,8,9\}}
\end{equation}
have five sensitive coordinates on the negative and positive sides,
respectively. Write $C_b(f)$ for the maximum, over inputs of output $b$,
of the smallest number of fixed coordinates that force that output.
These certificate complexities are distinct from the energy maxima $C_+$
and $C_-$. Direct enumeration also gives
\begin{equation}
    s_0(f_R)=s_1(f_R)=C_0(f_R)=C_1(f_R)=5.
    \label{eq:r10-recursive-base-sensitivity}
\end{equation}

The same exact reconstruction rechecks the four boundary energies of the
natural and alternative regular programs recorded in
Appendix~\ref{app:r10-program-verification}.

\subsection{Additional classical bounds}
\label{app:r10-classical-stress}

The exact checks in \path{anc/audit_classical_bounds/} give fractional block
sensitivity $3$ on $247$ inputs in each output fiber and $5$ on the remaining
nine inputs in each fiber. The package supplies a feasible fractional packing
and an equal-value fractional cover for every input. Consequently
$\operatorname{fbs}_0(f_R)=\operatorname{fbs}_1(f_R)=5$.
Recursive certificates give $C_b(F_h)\leq5^h$, while
Proposition~\ref{prop:recursive-r10-sensitivity} gives $s_b(F_h)=5^h$.
The inequalities $s_b\leq\operatorname{fbs}_b\leq C_b$ therefore imply
$\operatorname{fbs}_b(F_h)=5^h$ directly, without assuming a general
multiplicativity theorem for fractional block sensitivity.

Under the zero-error, worst-case expected-query convention for randomized
sabotage complexity of Ben-David and Kothari~\cite{BenDavidKothari2018Sabotage},
the same package certifies
\[
 \operatorname{RS}(f_R)=\frac{240947}{53451}=4.507810892\ldots.
\]
The upper certificate is a rational randomized decision policy; the lower
certificate is a distribution on single-sabotage inputs evaluated by an
independent exact decision-tree recursion. For total Boolean functions,
restriction to single-sabotage inputs preserves this complexity
~\cite[Thm.~3.6]{BenDavidKothari2018Sabotage}. The composition and comparison
theorems give
$R^\epsilon(F_h)\geq(1-2\epsilon)\operatorname{RS}(f_R)^h$
~\cite[Thms.~3.4 and 4.5]{BenDavidKothari2018Sabotage}, which is weaker than
the sensitivity bound. Thus neither of these two measures improves its base
$5$.

\subsection{Depth-two closure check}
\label{app:r10-recursive-boolean}

The verifier constructs the binary representation obtained by attaching nine
copies of pointed $R_{10}$ to a fresh root copy by literal $2$-sums. The
result has $82$ elements and rank $41$. It checks all $512$ possible outer
child-output patterns against Boolean self-composition of the pointed
closure function and then checks an additional $500$ deterministic
pseudo-random $81$-bit leaf inputs directly by matrix closure. The recursive
hard inputs have sensitivity $25$ at depth two.

\subsection{Recursive regular-program costs}
\label{app:r10-recursive-program}

For the alternative regular program with weight $5/4$ on its $R_{10}$
elements and weight $1$ on its graphic elements,
the verifier independently reconstructs the $R_{10}$ and graphic components
and checks
\begin{equation}
    C_+(1)=\frac{12}{5},
    \qquad
    C_-(1)=\frac{29}{4},
    \qquad
    C_+(1)C_-(1)=\frac{87}{5}.
\end{equation}
At depth two it checks all $512$ outer Boolean patterns in the extremal
$2$-sum envelope and obtains
\begin{equation}
    C_+(2)=\frac{144}{25},
    \qquad
    C_-(2)=\frac{841}{16},
    \label{eq:r10-recursive-depth-two-costs}
\end{equation}
with equality on the recursively defined hard positive and negative inputs.
These checks support the recurrence proof used in the scalable application;
they are not an extrapolation to larger depth.

\subsection{Ancillary verification}
\label{app:r10-recursive-ancillary}

The files for this appendix are under
\path{anc/r10_recursive_verification/}. The directory contains four
standard-library verification scripts, their shared module, and the report
and theorem ledger from the recursive calculation:
\begin{center}
\small
\texttt{verify\_base\_r10.py}\\
\texttt{verify\_recursive\_boolean.py}\\
\texttt{verify\_response\_recurrence.py}\\
\texttt{verify\_all.py}.
\end{center}
Running
\begin{center}
\texttt{python3 verify\_all.py}
\end{center}
replays all checks in this appendix using exact integer or rational
arithmetic. All four computational appendices can be replayed from the top
level of the ancillary directory with \texttt{python3 verify\_all.py}.

\section{Matroidal Form of the Span-Program/Adversary Correspondence}
\label{app:universal-normal-form}
\label{sec:matroid-span-program-optimality}
\label{sec:matroid-span-program-normal-forms}

Canonical span programs and the construction of optimal span programs from
adversary solutions are standard
~\cite{KarchmerWigderson1993,Reichardt2009,Reichardt2014Equivalent}. This
appendix records a different structural organization of that correspondence.
A represented program gives a correlation transport. In the reverse
direction, an optimal Gram factorization is expanded into free response
generators; the positive-input relations are then adjoined and contracted.
The resulting represented program has a response-labeled program matroid.
The equality of optimum values is known in the ordinary Boolean setting; the
additional organization here is the response-labeled quotient construction
and its compatibility with the distinction between source and program
matroids.

Throughout this appendix, $f:D\to\{0,1\}$ has an $f$-separation matroid
$\mathfrak M_f=(M_f,\Sigma_f)$, and $\mathcal Q$ denotes its finite query
set.

\subsection{From a represented program to correlation transport}
\label{subsec:program-to-correlation-transport}

Fix
$P=(M_P,\tau_P,\lambda_P)\in\mathscr R_{\mathbb R}(f,\mathfrak M_f)$
and a real representation $\phi_P:E_P\to V_P$.

For each query $q$, let $E_{P,q}$ be the set of non-target elements whose
response label has first coordinate $q$. The sets $E_{P,q}$ partition
$E_P\setminus\{\tau_P\}$.

For every positive input $x$, choose a reconstruction
$a_x\in\mathcal R_x(\phi_P)$ attaining
$c_+(\phi_P,x)$. For every negative input $y$, choose
$\ell_y\in\mathcal S_y(\phi_P)$ attaining
$c_-(\phi_P,y)$. These choices exist by
Proposition~\ref{prop:representation-energy-attainment}.

For $f(x)=1$ and $q\in\mathcal Q$, define
$r_{xq}\in\mathbb R^{E_{P,q}}$ by
\begin{equation}
    r_{xq}(e)
    :=
    \begin{cases}
        a_x(e),&e\in A_P(x),\\
        0,&e\notin A_P(x).
    \end{cases}
\end{equation}
For $f(y)=0$, define $s_{yq}\in\mathbb R^{E_{P,q}}$ by
$s_{yq}(e):=\ell_y(\phi_P(e))$.

The key pairing identity shows that reconstruction on a positive input and
separation on a negative input interact only through queries that distinguish
the pair.

\begin{lemma}
\label{lem:reconstruction-separation-pairing}
If $f(x)=1$ and $f(y)=0$, then
\begin{equation}
    1
    =
    \sum_{\substack{q\in\mathcal Q\\
                    (x,y)\in\operatorname{Sep}(q)}}
        r_{xq}\mathbin{\cdot}s_{yq},
\end{equation}
where
$r_{xq}\mathbin{\cdot}s_{yq}:=\sum_{e\in E_{P,q}}r_{xq}(e)s_{yq}(e)$.
\end{lemma}

\begin{proof}
Since $a_x$ reconstructs the target,
\begin{equation}
    \phi_P(\tau_P)
    =
    \sum_{e\in A_P(x)}
        a_x(e)\phi_P(e).
\end{equation}
Applying $\ell_y$ and using
$\ell_y(\phi_P(\tau_P))=1$ gives
\begin{equation}
    1
    =
    \sum_{e\in A_P(x)}
        a_x(e)\ell_y(\phi_P(e))
    =
    \sum_{q\in\mathcal Q}
        r_{xq}\mathbin{\cdot}s_{yq}.
\end{equation}

Suppose $q$ does not separate $x$ and $y$. Then
$B_q(x)=B_q(y)$. Every element of $E_{P,q}$ which is available on $x$
is also available on $y$. Since $\ell_y$ vanishes on every
available vector at $y$, the corresponding term
$r_{xq}\mathbin{\cdot}s_{yq}$ is zero.

Only queries which separate $x$ and $y$ remain in the sum.
\end{proof}

This identity produces the positive semidefinite kernels required in the
correlation transport problem. The resulting adversary lower comparison is
the standard span-program-witness direction~\cite{Reichardt2009,Reichardt2014Equivalent},
written here in response-label and correlation-transport notation.

\begin{proposition}
\label{prop:represented-program-gives-transport}
For every real representation $\phi_P$,
\begin{equation}
    \operatorname{CT}_{\mathcal Q}(f)
    \leq
    \max
    \{
        C_+(\phi_P),
        C_-(\phi_P)
    \}.
\end{equation}
Consequently,
\begin{equation}
    \operatorname{Adv}^{\pm}_{\mathcal Q}(f)
    \leq
    \mathcal E(\phi_P).
\end{equation}
\end{proposition}

\begin{proof}
For each input $x$ and query $q$, define
$u_{xq}\in\mathbb R^{E_{P,q}}$ by
$u_{xq}:=r_{xq}$ when $f(x)=1$ and
$u_{xq}:=s_{xq}$ when $f(x)=0$. Set
$X_q(x,y):=u_{xq}\mathbin{\cdot}u_{yq}$.

Each $X_q$ is a positive semidefinite kernel. Indeed, for every
$\alpha:D\to\mathbb R$,
\begin{equation}
    \sum_{x,y}
        \alpha(x)\alpha(y)X_q(x,y)
    =
    \left\|
        \sum_x\alpha(x)u_{xq}
    \right\|_2^2
    \geq0.
\end{equation}

For a pair with different target values,
Lemma~\ref{lem:reconstruction-separation-pairing} gives the
correlation-transport equation. Hence $(X_q)_{q\in\mathcal Q}$ is a
correlation transport.

If $f(x)=1$, its load at $x$ is
\begin{equation}
    \sum_qX_q(x,x)
    =
    \sum_q\|r_{xq}\|_2^2
    =
    \|a_x\|_2^2
    =
    c_+(\phi_P,x).
\end{equation}
If $f(x)=0$, its load is similarly $c_-(\phi_P,x)$. Therefore the
transport cost is at most
\begin{equation}
    \max\{C_+(\phi_P),C_-(\phi_P)\}.
\end{equation}

Now rescale the target and non-target vectors as in
Proposition~\ref{prop:representation-scaling}. Every such rescaling
represents the same matroid program, so the preceding transport bound holds
for each relative scale. Proposition~\ref{prop:representation-balance}
then gives
\begin{equation}
    \operatorname{CT}_{\mathcal Q}(f)
    \leq
    \mathcal E(\phi_P).
\end{equation}
Theorem~\ref{thm:correlation-transport-adversary} identifies the left side
with $\operatorname{Adv}^{\pm}_{\mathcal Q}(f)$.
\end{proof}

Taking the two infima from
Definitions~\ref{def:linear-cost} and
\ref{def:matroid-span-program-cost} gives the first comparison.

\begin{theorem}
\label{thm:adversary-lower-bounds-msp}
Let $f:D\to\{0,1\}$, let $\mathcal Q$ be a finite query set with finite
response partitions, and suppose the query problem admits the
$f$-separation matroid $\mathfrak M_f$. Then
\begin{equation}
    \operatorname{Adv}^{\pm}_{\mathcal Q}(f)
    \leq
    \operatorname{MSP}(f,\mathfrak M_f).
\end{equation}
\end{theorem}

\subsection{Computational factorizations}
\label{subsec:computational-factorizations}

The vector factorization used in this subsection is the standard dual-adversary/filtered-$\gamma_2$ factorization~\cite{Reichardt2009,LeeMittalReichardtSpalekSzegedy2011}. The quotient construction begins after this factorization has been fixed.

The reverse construction starts from a Gram factorization of the positive
kernels in a correlation transport.

\begin{definition}
\label{def:computational-factorization}
A \emph{computational factorization} of $(f,\mathfrak M_f)$ consists, for
each $q\in\mathcal Q$, of a finite-dimensional real coordinate space
$\mathbb R^{d_q}$ and vectors $u_{xq}\in\mathbb R^{d_q}$ for every
$x\in D$, such that
\begin{equation}
    \sum_{\substack{q\in\mathcal Q\\
                    (x,y)\in\Sigma_f(q)}}
        u_{xq}\mathbin{\cdot}u_{yq}
    =
    1
\end{equation}
whenever $f(x)\neq f(y)$. Its load at $x$ is
$\sum_q\|u_{xq}\|_2^2$, and its cost is the maximum load over $x\in D$.
\end{definition}

The dimensions $d_q$ may be zero. Since $D$ is finite, finite-dimensional
coordinate spaces always suffice.

\begin{proposition}
\label{prop:factorization-transport-equivalence}
A computational factorization and a correlation transport are equivalent
for feasibility and cost, up to the choice of Gram coordinates. Consequently,
if $\mathcal Q$ separates $f$, an optimal computational factorization exists
and its cost is $\operatorname{Adv}^{\pm}_{\mathcal Q}(f)$.
\end{proposition}

\begin{proof}
Given a computational factorization, define
$X_q(x,y):=u_{xq}\mathbin{\cdot}u_{yq}$. Each $X_q$ is a positive semidefinite kernel,
and the factorization equation is the transport equation because
$\Sigma_f(q)=\operatorname{Sep}(q)$. Moreover,
$X_q(x,x)=\|u_{xq}\|_2^2$, so the costs agree.

Conversely, each positive semidefinite kernel $X_q$ on the finite set $D$ is a Gram
kernel of finitely many vectors $u_{xq}$. Choosing such a Gram
factorization for every $q$ gives a computational factorization of the same
cost. This choice is not unique: orthogonal coordinate changes and added zero
coordinates preserve the Gram kernels and all costs. Proposition~\ref{prop:correlation-transport-attainment} gives an
optimal transport when $\mathcal Q$ separates $f$, and
Theorem~\ref{thm:correlation-transport-adversary} identifies its cost with
$\operatorname{Adv}^{\pm}_{\mathcal Q}(f)$.
\end{proof}

\subsection{The free response expansion}
\label{subsec:free-response-expansion}

The binary linear construction below is already implicit in
Reichardt's canonical-program machinery. In the extended version of
\cite{Reichardt2009}, Lemma~6.5 constructs a canonical program from adversary
vectors, Definition~4.2 complements it by adjoining the target coordinate and
free relation vectors, and Proposition~4.10 removes the free vectors. After a
uniform sign convention, these operations give the same free response
generators, positive-input relations, and linear quotient used below. Our
additional organization retains the response labels and the resulting vector
matroid, and treats that abstract program type as a stratum of the realization
optimization in Definition~\ref{def:linear-cost}.

Fix a computational factorization
$u=(u_{xq})$ with dimensions $(d_q)_{q\in\mathcal Q}$.

Let $E_u^\circ:=\{(q,B,k):q\in\mathcal Q,\ B\in\mathcal B_q,\ 1\leq k\leq d_q\}$.
Introduce one additional symbol $\widetilde\tau$, and set
$\widetilde E_u:=E_u^\circ\sqcup\{\widetilde\tau\}$.

Let $F_f(u)$ be the real vector space with basis consisting of a target
generator $\mathbf t$ and generators
$\mathbf e_{q,B,k}$ indexed by $E_u^\circ$.

Each non-target ground-set element $(q,B,k)$ receives the response label
$(q,B)$. Thus the elements available on input $x$ are
\begin{equation}
    \widetilde A_u(x)
    =
    \{
        (q,B_q(x),k):
        q\in\mathcal Q,\
        1\leq k\leq d_q
    \}.
\end{equation}

There is also a set map from the expanded ground set back to the
ground set $\mathcal Q\sqcup\{\tau_f\}$ of the targeted query matroid:
$\widetilde\tau$ maps to $\tau_f$, while $(q,B,k)$ maps to $q$. We call
this the \emph{query projection}. It records the source query of each
expanded element; no matroid-map property is included in this terminology.

The free response expansion contains no imposed linear dependencies.
Its ground-set organization, labels, and input-dependent availability are
determined by the response partitions recovered from $\Sigma_f$.

\subsection{The relation space}
\label{subsec:normal-form-relation-space}

For each positive input $x$, define
\begin{equation}
    r_x
    :=
    \mathbf t
    -
    \sum_{q\in\mathcal Q}
    \sum_{k=1}^{d_q}
        u_{xq}(k)\,
        \mathbf e_{q,B_q(x),k}.
\end{equation}

\begin{definition}
\label{def:normal-form-relation-space}
The \emph{relation space} of the computational factorization $u$ is
\begin{equation}
    L_f(u)
    :=
    \operatorname{span}
    \{
        r_x:
        f(x)=1
    \}
    \subseteq F_f(u).
\end{equation}
\end{definition}

Each generator appearing in $r_x$ carries the response observed on $x$.
Thus the equation $r_x=0$ says that the target can be reconstructed from
the response generators available on that positive input.

\subsection{The quotient program}
\label{subsec:quotient-normal-form}

Define the quotient space
$V_f(u):=F_f(u)/L_f(u)$ and let
$\pi_u:F_f(u)\to V_f(u)$ be the quotient map.

Define
$\widetilde\phi_u:\widetilde E_u\to V_f(u)$ by
$\widetilde\phi_u(\widetilde\tau):=\pi_u(\mathbf t)$ and
\begin{equation}
    \widetilde\phi_u(q,B,k)
    :=
    \pi_u(\mathbf e_{q,B,k}).
\end{equation}
Let $\widetilde M_f(u)$ be the vector matroid of
$\widetilde\phi_u$, and let $\widetilde\lambda_u(q,B,k):=(q,B)$.

\begin{definition}
\label{def:quotient-normal-form}
The \emph{quotient normal form associated with $u$} is the matroid span
program
\begin{equation}
    \widetilde P_f(u)
    :=
    \bigl(
        \widetilde M_f(u),
        \widetilde\tau,
        \widetilde\lambda_u
    \bigr).
\end{equation}
The map $\widetilde\phi_u:\widetilde E_u\to V_f(u)$ defined above is
its \emph{quotient representation}.
\end{definition}

\begin{remark}[Orthogonal gauge covariance]
\label{rem:quotient-gauge-covariance}
Give each free response space the Euclidean inner product for which
$\mathbf t$ and the generators $\mathbf e_{q,B,k}$ form an orthonormal basis,
and give each quotient $V_f(u)$ the induced quotient inner product. For each
query $q$, let $O_q\in O(d_q)$ and replace
$u_{xq}$ by $u'_{xq}:=O_qu_{xq}$. This preserves every factorization
equation and every load. It also induces an isometry of free response
spaces which fixes $\mathbf t$ and applies $O_q$ separately inside each
response block for $q$. Under this isometry,
\begin{equation}
    L_f(u)\longmapsto L_f(u'),
    \qquad
    V_f(u)\cong V_f(u').
\end{equation}
Thus the target and the available response \emph{subspaces} are preserved,
up to the induced quotient isometry.

The individual coordinate generators are not preserved unless the $O_q$
are signed permutation matrices. Consequently the response-subspace program
is invariant under an orthogonal gauge change, while the vector matroid on
the chosen coordinate generators need not be.
\end{remark}

There is also a standard matroid description of the quotient.
Proposition~\ref{prop:represented-minors} identifies contraction of represented
elements with quotienting the ambient vector space by their span. For each
positive input $x$, introduce a relation element $\rho_x$. Let
$\widehat M_f(u)$ be the vector matroid represented in $F_f(u)$ by
$\mathbf t$ for $\widetilde\tau$, by $\mathbf e_{q,B,k}$ for
$(q,B,k)\in E_u^\circ$, and by $r_x$ for $\rho_x$. Write
$R_+(u):=\{\rho_x:f(x)=1\}$.

Before contraction, each positive-input relation element has a fundamental
circuit that records exactly which response generators participate in that
relation.

\begin{proposition}
\label{prop:relation-fundamental-circuits}
The set
\begin{equation}
    \widetilde E_u
    =
    E_u^\circ\sqcup\{\widetilde\tau\}
\end{equation}
is a basis of $\widehat M_f(u)$. For every positive input $x$, the
fundamental circuit of $\rho_x$ with respect to this basis is
\begin{equation}
    \{\rho_x,\widetilde\tau\}
    \cup
    \{
        (q,B_q(x),k):u_{xq}(k)\neq0
    \}.
\end{equation}
\end{proposition}

\begin{proof}
The vectors representing the elements of $\widetilde E_u$ are the
basis vectors $\mathbf t$ and $\mathbf e_{q,B,k}$ of $F_f(u)$, so they form
a basis of the represented matroid. The vector representing $\rho_x$ is
\begin{equation}
    r_x
    =
    \mathbf t-
    \sum_{q,k}u_{xq}(k)\mathbf e_{q,B_q(x),k}.
\end{equation}
Relative to the basis $\widetilde E_u$, its nonzero coordinates are exactly
the displayed elements, together with the target coordinate. The unique
fundamental circuit is thus the support of this dependence together
with $\rho_x$.
\end{proof}

Thus the factorization coefficients refine a purely matroidal statement:
each positive input supplies a target circuit before the relation elements
are contracted. The coefficient values determine how the different relation
elements interact, while the circuit supports record which response
generators participate.

\begin{proposition}
\label{prop:quotient-normal-form-contraction}
The quotient normal-form matroid is the contraction of the relation elements:
\begin{equation}
    \widetilde M_f(u)
    \cong
    \widehat M_f(u)/R_+(u),
\end{equation}
where the isomorphism is the identity on
$\widetilde E_u=E_u^\circ\sqcup\{\widetilde\tau\}$.
\end{proposition}

\begin{proof}
Apply Proposition~\ref{prop:represented-minors} to the represented set
$R_+(u)$. Its representing vectors are the relation vectors $r_x$,
whose span is $L_f(u)$. Hence contraction is represented in
$F_f(u)/L_f(u)=V_f(u)$ by the remaining maps $\widetilde\phi_u$, whose
vector matroid is $\widetilde M_f(u)$.
\end{proof}

The dependence structure of $\widetilde M_f(u)$ is produced in
two stages. The response expansion specifies the generators and their
labels. The relation elements record the positive-input equations, and
contracting those elements produces the dependence relations of the normal
form.

\subsection{Correctness and cost}
\label{subsec:quotient-normal-form-correctness}

The quotient construction computes $f$ for every computational
factorization, not only an optimal one.

\begin{theorem}
\label{thm:quotient-normal-form-correct}
For every computational factorization $u$,
$\widetilde P_f(u)$ is an $f$-matroid span program over
$\mathfrak M_f$. Equivalently,
\begin{equation}
    f(x)=1
    \quad\Longleftrightarrow\quad
    \widetilde\tau
    \in
    \cl_{\widetilde M_f(u)}
    \bigl(\widetilde A_u(x)\bigr).
\end{equation}
\end{theorem}

\begin{proof}
Suppose first that $f(x)=1$. Since $r_x\in L_f(u)$, its image in the
quotient is zero. Hence
\begin{equation}
    \widetilde\phi_u(\widetilde\tau)
    =
    \sum_{q,k}
        u_{xq}(k)
        \widetilde\phi_u(q,B_q(x),k).
\end{equation}
Every element on the right belongs to $\widetilde A_u(x)$. The target
hence lies in the span of the available represented elements, which is
equivalent to the stated closure relation.

Now suppose $f(y)=0$. Define a linear functional
$\widehat\ell_y:F_f(u)\to\mathbb R$ by
$\widehat\ell_y(\mathbf t):=1$ and
\begin{equation}
    \widehat\ell_y(\mathbf e_{q,B,k})
    :=
    \begin{cases}
        0,&B=B_q(y),\\
        u_{yq}(k),&B\neq B_q(y).
    \end{cases}
\end{equation}

For a positive input $x$,
\begin{equation}
\begin{aligned}
    \widehat\ell_y(r_x)
    &=
    1-
    \sum_{\substack{q\in\mathcal Q\\
                    B_q(x)\neq B_q(y)}}
        u_{xq}\mathbin{\cdot}u_{yq}\\
    &=
    1-
    \sum_{\substack{q\in\mathcal Q\\
                    (x,y)\in\Sigma_f(q)}}
        u_{xq}\mathbin{\cdot}u_{yq}\\
    &=
    0.
\end{aligned}
\end{equation}
Thus $L_f(u)\subseteq\ker\widehat\ell_y$, so
$\widehat\ell_y$ descends to a linear functional
$\ell_y:V_f(u)\to\mathbb R$.

The descended functional takes value $1$ on
$\widetilde\phi_u(\widetilde\tau)$. For an available generator
$(q,B_q(y),k)\in\widetilde A_u(y)$, we have
\begin{equation}
    \ell_y(\widetilde\phi_u(q,B_q(y),k))=0.
\end{equation}
Hence $\ell_y$ vanishes on the span of every represented element available
on $y$ while remaining nonzero on the target. The target is consequently
outside that span.
\end{proof}

The positive relations and the negative separating functionals also give
complementary rank certificates. Write $m_u:=\dim F_f(u)$. In the free
basis of $F_f(u)$, let $R_u$ be the matrix whose rows are the relation
vectors $r_x$, $x\in D_+$, and let $S_u$ be the matrix whose rows are the
coordinate vectors of the functionals $\widehat\ell_y$, $y\in D_-$, used
in the proof above.

\begin{proposition}
\label{prop:quotient-complementary-rank}
For every computational factorization $u$,
\begin{equation}
    S_uR_u^{\mathsf T}=0.
\end{equation}
Hence
\begin{equation}
    \operatorname{rank}R_u+\operatorname{rank}S_u\leq m_u.
\end{equation}
Moreover,
\begin{equation}
    \operatorname{rank}R_u=\dim L_f(u),
    \qquad
    \dim V_f(u)=m_u-\operatorname{rank}R_u.
\end{equation}
If lower-rank certificates give
\begin{equation}
    \operatorname{rank}R_u\geq r,
    \qquad
    \operatorname{rank}S_u\geq m_u-r,
\end{equation}
then equality holds throughout:
\begin{equation}
    \operatorname{rank}R_u=r,
    \qquad
    \operatorname{rank}S_u=m_u-r,
    \qquad
    \dim V_f(u)=m_u-r.
\end{equation}
In this case the rows of $S_u$ span $L_f(u)^\perp$, identified with
$V_f(u)^*$.
\end{proposition}

\begin{proof}
For $x\in D_+$ and $y\in D_-$,
\begin{equation}
    (S_uR_u^{\mathsf T})_{y,x}
    =
    \widehat\ell_y(r_x)
    =
    1-
    \sum_{\substack{q\in\mathcal Q\\
                    B_q(x)\neq B_q(y)}}
        u_{xq}\mathbin{\cdot}u_{yq}
    =
    0
\end{equation}
by the computational factorization equation. Thus the row space of $S_u$
is contained in the orthogonal complement of the row space of $R_u$, which
gives the rank inequality. The row space of $R_u$ is $L_f(u)$ in the free
basis, so the dimension identities follow from the definition of the
quotient. If the two lower bounds sum to $m_u$, the rank inequality forces
both to be equalities. The last statement follows because then the row space
of $S_u$ has the same dimension as $L_f(u)^\perp$.
\end{proof}

There is a matroidal form of the same correctness statement. Let $N_u:=\widetilde M_f(u)\setminus\widetilde\tau$ and let $\mathcal C_u$ be the modular cut of $N_u$ determined by the target
extension $\widetilde M_f(u)$. Then
\begin{equation}
    \cl_{N_u}(\widetilde A_u(x))\in\mathcal C_u
    \quad\Longleftrightarrow\quad
    f(x)=1.
\end{equation}
Thus the factorization has produced a real-representable target extension
whose modular cut contains all positive availability flats and excludes all
negative ones. Before contraction,
Proposition~\ref{prop:relation-fundamental-circuits} gives the fundamental
circuit of each positive relation element. After contracting all relation
elements, the available support $S_x$ of that relation spans the target, so
$S_x\cup\{\widetilde\tau\}$ contains a circuit through $\widetilde\tau$.
On a negative input, the descended functional in
the proof above has a zero flat containing every available element and
excluding the target. Extending that flat to a hyperplane which still
excludes the target gives a target cocircuit disjoint from the available
set. The linear factorization supplies the coefficients, while the resulting
correctness certificates are ordinary circuits, cocircuits, and a target
modular cut.

\begin{proposition}
\label{prop:quotient-normal-form-cost}
Let $u$ be a computational factorization. If $f(x)=1$, then
\begin{equation}
    c_+(\widetilde\phi_u,x)
    \leq
    \sum_{q\in\mathcal Q}\|u_{xq}\|_2^2.
\end{equation}
If $f(y)=0$, then
\begin{equation}
    c_-(\widetilde\phi_u,y)
    \leq
    \sum_{q\in\mathcal Q}
        \bigl(|\mathcal B_q|-1\bigr)
        \|u_{yq}\|_2^2.
\end{equation}
\end{proposition}

\begin{proof}
For a positive input, use the coefficients appearing in the relation
$r_x$. Their squared coordinate size is
$\sum_q\|u_{xq}\|_2^2$.

For a negative input, use the functional $\ell_y$ constructed in
Theorem~\ref{thm:quotient-normal-form-correct}. Its evaluation on
$\widetilde\phi_u(q,B,k)$ is zero for $B=B_q(y)$ and equals
$u_{yq}(k)$ for every other response block. Its squared evaluation size is
hence
\begin{equation}
    \sum_q
        \bigl(|\mathcal B_q|-1\bigr)
        \|u_{yq}\|_2^2.
\end{equation}
Minimizing over reconstructions and separating functionals gives the stated
inequalities.
\end{proof}

\subsection{An optimal quotient program}
\label{subsec:optimal-matroid-normal-form}

Assume now that $f$ is nonconstant, that $\mathcal Q$ separates $f$, and
that every query has two response blocks. Choose an optimal computational
factorization $u^\star$ and write
$\widetilde P_f:=\widetilde P_f(u^\star)$,
$\widetilde M_f:=\widetilde M_f(u^\star)$, and
$\widetilde\phi_f:=\widetilde\phi_{u^\star}$.

Reichardt established the equality between the optimal standard span-program
quadratic complexity measure and the general adversary bound
~\cite{Reichardt2009,Reichardt2014Equivalent},
and the binary linear quotient underlying this construction is the
canonical/complement/strictification construction described above. The theorem
below records the labeled program produced by that quotient and its
compatibility with the fixed-type optimization of Definition~\ref{def:linear-cost};
it does not assert a new linear adversary-to-span-program conversion.

\begin{theorem}
\label{thm:optimal-matroid-normal-form}
\label{thm:binary-msp-adversary-equality}
Let $f:D\to\{0,1\}$ be nonconstant. Let $\mathcal Q$ be a finite query set
that separates $f$, suppose every query has two response blocks, and suppose
the query problem admits the $f$-separation matroid $\mathfrak M_f$. Let
$u^\star$ be an optimal computational factorization, and let
$\widetilde P_f(u^\star)$ with representation
$\widetilde\phi_{u^\star}$ be the associated quotient program. Then
$\widetilde P_f(u^\star)$ is a real-representable response-labeled program
computing $f$ through target closure, and
\begin{equation}
    \mathcal E(\widetilde\phi_{u^\star})
    =
    \operatorname{lincost}(\widetilde P_f(u^\star))
    =
    \operatorname{MSP}(f,\mathfrak M_f)
    =
    \operatorname{Adv}^{\pm}_{\mathcal Q}(f).
\end{equation}
\end{theorem}

\begin{proof}
Correctness is Theorem~\ref{thm:quotient-normal-form-correct}. Let
$w$ be the cost of the optimal computational factorization. By
Proposition~\ref{prop:factorization-transport-equivalence},
$w=\operatorname{Adv}^{\pm}_{\mathcal Q}(f)$.

Because every query is binary,
Proposition~\ref{prop:quotient-normal-form-cost} gives
$C_+(\widetilde\phi_f)\leq w$ and
$C_-(\widetilde\phi_f)\leq w$. Hence
$\mathcal E(\widetilde\phi_f)\leq w$.
Proposition~\ref{prop:represented-program-gives-transport} gives the reverse
inequality
$\operatorname{Adv}^{\pm}_{\mathcal Q}(f)
\leq\mathcal E(\widetilde\phi_f)$.
Thus $\mathcal E(\widetilde\phi_f)=w$.

Proposition~\ref{prop:represented-program-gives-transport} applies to every
real representation of the fixed program $\widetilde P_f$, so
$\operatorname{lincost}(\widetilde P_f)\geq w$. The displayed representation
has balanced cost $w$, and thus
$\operatorname{lincost}(\widetilde P_f)=w$. Finally,
\begin{equation}
    w
    \leq
    \operatorname{MSP}(f,\mathfrak M_f)
    \leq
    \operatorname{lincost}(\widetilde P_f)
    =
    w,
\end{equation}
so the global program optimum has the same value.
\end{proof}

The theorem also describes the optimizer structurally. Its response-labeled
generators are free before the positive-input relations are adjoined; those
relation elements are then contracted, leaving the program matroid
$\widetilde M_f$. The source matroid $\mathfrak M_f=(M_f,\Sigma_f)$ and
the resulting program matroid have different roles. The targeted query
matroid $M_f$ records fixed-query closure, while $\Sigma_f$ records the
separation colors used to recover response partitions. The computational
factorization determines the new dependence relations imposed on the response
generators.

For arbitrary finite response partitions, the quotient construction remains
correct but the same cost comparison gives a controlled loss. For
$\ell$-ary coordinate queries, Beigi and Taghavi obtain the same
$\sqrt{\ell-1}$ loss when converting a dual-adversary solution to the
orthogonal-input non-binary span-program form
~\cite{BeigiTaghavi2019}. Theorem~\ref{thm:multiresponse-msp-comparison}
extends that response-cardinality dependence to arbitrary finite
partition-valued queries and the present quotient program-matroid normal
form.

\begin{proof}[Detailed proof of Theorem~\ref{thm:multiresponse-msp-comparison}]
Choose an optimal computational factorization $u$ of cost
$w=\operatorname{Adv}^{\pm}_{\mathcal Q}(f)$. By
Proposition~\ref{prop:quotient-normal-form-cost},
$C_+(\widetilde\phi_u)\leq w$ and
$C_-(\widetilde\phi_u)\leq\kappa(\mathcal Q)w$, so
$\mathcal E(\widetilde\phi_u)\leq\sqrt{\kappa(\mathcal Q)}\,w$. The upper
bound follows from the definition of $\operatorname{MSP}$, and the lower
bound is Theorem~\ref{thm:adversary-lower-bounds-msp}.
\end{proof}

The factor enters only because the negative functional sees the same
$q$-component once for every response block other than the observed one. For
binary queries there is only one such block, and the loss disappears.

This is the Gram form of the general adversary optimization used in
Reichardt's span-program optimality theorem
~\cite{Reichardt2009,Reichardt2014Equivalent}. Its role here is structural:
for binary queries, the quotient construction gives a matroidal realization
of the known optimum. For larger response alphabets,
Theorem~\ref{thm:multiresponse-msp-comparison} generalizes the known
non-binary response-cardinality loss~\cite{BeigiTaghavi2019}. Standard
span-program algorithms convert represented programs into bounded-error
quantum query algorithms with comparable query complexity
~\cite{Reichardt2011,Reichardt2014Equivalent}.

\section{Real Amplitudes and Reversible Recombination}
\label{app:square-root-carrier}

The main query model begins directly with normalized real amplitudes. This
appendix records the elementary reversible-mixing argument that motivates the
quadratic carrier while leaving the query-count definition independent of
that derivation.

\subsection{Reversible classical gates}

For a finite configuration arena $\Omega$, write
\begin{equation}
    \Delta(\Omega):=\{p:\Omega\to[0,1]:\sum_{c\in\Omega}p(c)=1\}.
\end{equation}

\begin{definition}
\label{def:reversible-classical-gate}
A \emph{reversible classical gate} on $\Omega$ is an affine bijection
$M:\Delta(\Omega)\to\Delta(\Omega)$.
\end{definition}

For $c\in\Omega$, let $e_c\in\Delta(\Omega)$ denote the state concentrated
at $c$.

\begin{theorem}
\label{thm:permutation-wall}
Every reversible classical gate permutes the configurations in $\Omega$.
\end{theorem}

\begin{proof}
The states $e_c$ are exactly the extreme points of the simplex
$\Delta(\Omega)$. An affine bijection preserves extreme points, so it
permutes these states. Affineness then determines the same permutation on
every convex combination.
\end{proof}

Thus reversible classical mixing can only relabel configurations. To allow
spreading and recombination, consider signed carrier values with linear
evolution and recover probabilities from their magnitudes.

\subsection{The square-root carrier}

For $c\in\Omega$, now let $e_c:\Omega\to\mathbb R$ denote the atomic weight
assignment with value $1$ at $c$ and $0$ elsewhere.

\begin{definition}
\label{def:power-law-carrier}
Fix $r\geq1$. An \emph{$r$-carrier state} on $\Omega$ is a function
$a:\Omega\to\mathbb R$ satisfying
$\sum_c|a(c)|^r=1$. It reports probability $p_a(c):=|a(c)|^r$.
A \emph{reversible $r$-carrier gate} is an invertible linear map on
$\mathbb R^\Omega$ preserving $\sum_c|a(c)|^r$.
\end{definition}

\begin{definition}
\label{def:balanced-recombination}
Let $c,d\in\Omega$ be distinct. A \emph{balanced recombination} of $c$ and
$d$ is a linear involution $R$ which fixes $e_u$ for
$u\notin\{c,d\}$ and, for some nonzero $\alpha\in\mathbb R$, satisfies
\begin{equation}
    Re_c=\alpha(e_c+e_d),
    \qquad
    Re_d=\alpha(e_c-e_d).
\end{equation}
\end{definition}

\begin{theorem}
\label{thm:square-root-forced}
If a balanced recombination is a reversible $r$-carrier gate, then $r=2$.
\end{theorem}

\begin{proof}
Applying the recombination twice to $e_c$ gives
$R^2e_c=2\alpha^2e_c$. Since $R$ is an involution,
$|\alpha|=1/\sqrt2$. Probability preservation after one recombination gives
$2|\alpha|^r=1$. Hence $2^{1-r/2}=1$ and $r=2$.
\end{proof}

Within this power-law family, reversible balanced mixing thus selects
the quadratic carrier used in Definition~\ref{def:configuration-amplitude}.
Complex quantum circuits may be realified by adjoining a two-dimensional
real/imaginary register~\cite{BernsteinVazirani1997,Aharonov2003}; Appendix
~\ref{app:quantum-model-equivalence} uses this fact without increasing the
number of queries.

\section{Quantum Query Model Checks}
\label{app:quantum-query-checks}

The main development uses the partition-based quantum query model of Section~\ref{sec:quantum-query-circuits}. We compare it here with the standard response-oracle model of quantum query complexity~\cite{BealsBuhrmanCleveMoscaDeWolf2001}. Because an answer partition alone does not determine a quantum oracle action~\cite{Svozil2026AnswerPartitions}, the comparison is with the explicit one-hot response realization defined below. The equivalence is exact in query count for that realization.

\subsection{Equivalence with the standard quantum query model}
\label{app:quantum-model-equivalence}

A query partition has no preferred numerical labels on its blocks. To compare Definition~\ref{def:quantum-query-circuit} with a standard reversible response oracle, we choose one explicit encoding of the blocks.

Assume $\mathcal Q\ne\varnothing$; the empty family has only zero-query
algorithms and is handled directly. For each $q\in\mathcal Q$, let
$k_q:=|\llbracket q\rrbracket|$, and put
$K:=\max_{q\in\mathcal Q}k_q$. Choose an ordering $B_1^q,\ldots,B_{k_q}^q$ of the blocks of each query partition. Encode the answer $B_j^q$ by the
$j$th unit string $e_j\in\{0,1\}^K$.

The reference response oracle has a query register, a $K$-bit response register, and arbitrary additional work registers. Whenever $x\in B_j^q$, it sends $(q,z,w)$ to $(q,z\oplus e_j,w)$. This is the usual reversible response-oracle pattern, with the answer encoded by the chosen one-hot string~\cite{BealsBuhrmanCleveMoscaDeWolf2001}. The encoding itself is only a convenience for this comparison. As usual in query complexity, input-independent unitary gates and final classical postprocessing are free; only oracle calls are counted~\cite{BealsBuhrmanCleveMoscaDeWolf2001}.

\begin{theorem}
\label{thm:standard-quantum-equivalence}
For every finite domain $D$, finite query gate set $\mathcal Q$, target $f:D\to Y$, and error parameter $0\leq\epsilon<1$, the minimum number of queries in Definition~\ref{def:quantum-query-circuit} equals the minimum number of calls to the one-hot response oracle defined in this subsection required to $\epsilon$-compute $f$.
\end{theorem}

\begin{proof}
We first show that one phase query can be implemented using one standard
response-oracle query.

Fix a configuration $c\in\Omega$. Its phase address consists of a query $q_c$ and a phase assignment $\varphi_c:\llbracket q_c\rrbracket\to\{-1,+1\}$. Define a string $s_c\in\{0,1\}^K$ by
\begin{equation}
    (s_c)_j
    :=
    \begin{cases}
        0,&\varphi_c(B_j^{q_c})=1,\\
        1,&\varphi_c(B_j^{q_c})=-1.
    \end{cases}
\end{equation}
Bits beyond $k_{q_c}$ may be set to zero.

We use the standard phase-kickback observation that translation of the response register is diagonal in the Walsh basis. Express the response register in the Walsh basis indexed by $s\in\{0,1\}^K$. Translation by $e_j$ acts on the basis state indexed by
$s$ with sign $(-1)^{s_j}$. Hence, on the sector indexed by $s_c$, the
response oracle contributes
\begin{equation}
    (-1)^{(s_c)_j}
    =
    \varphi_c(B_j^{q_c})
\end{equation}
whenever $x\in B_j^{q_c}$.

Thus the standard response oracle acts on configuration $c$ by exactly the
sign prescribed by the phase query $O_x$.

If several configurations have the same query and phase assignment, assign
them distinct work labels. We may hence identify every configuration
$c\in\Omega$ with a standard basis label $(q_c,s_c,w_c)$ so that one
response-oracle call realizes $O_x$ simultaneously on all configurations.
The change between the ordinary response basis and the Walsh basis is input
independent and can be absorbed into the surrounding free gates.

It remains to compare the input-independent dynamics. By Definition~\ref{def:recombination}, a recombination $R:\mathbb R^\Omega\to\mathbb R^\Omega$ preserves the Euclidean norm. Since $R$ is linear on the finite-dimensional space $\mathbb R^\Omega$, its matrix is orthogonal. The same matrix is unitary on $\mathbb C^\Omega$. Before readout, apply the Walsh basis change to the response register so that
the Walsh basis state indexed by $s$ is carried to the ordinary computational
basis label $s$. This input-independent change can be absorbed into the final
free gate. The output map $\lambda:\Omega\to Y$ is then ordinary
computational-basis measurement followed by classical grouping of outcomes
carrying the same label. Every circuit of
Definition~\ref{def:quantum-query-circuit} is thus a standard quantum
query circuit with the same number of oracle calls and the same output
probabilities.

For the converse, first suppose a standard query algorithm has real amplitudes and real input-independent gates. In the Walsh basis, the response oracle is diagonal. If the response register lies in sector $s$ and the query register addresses $q$, define $\varphi(B_j^q):=(-1)^{s_j}$. This is a phase assignment for $q$, and together with the query label it gives a phase address in the sense of Definition~\ref{def:phase-address}. The remaining basis labels, including the control and work registers, become auxiliary configuration data.

Every real input-independent unitary gate is an orthogonal transformation and
hence a recombination. Before the original computational-basis readout,
apply the inverse Walsh basis change; because it is input independent, it can
be absorbed into the final recombination. The standard algorithm consequently has the form of
Definition~\ref{def:quantum-query-circuit}.

Finally, use the standard realification of complex quantum computation~\cite{BernsteinVazirani1997,Aharonov2003}. Write a complex amplitude assignment as $u+iv$, and write a unitary gate as $U=A+iB$. Replace it by the real transformation
\begin{equation}
    (u,v)
    \longmapsto
    (Au-Bv,\;Bu+Av).
\end{equation}
This transformation preserves
\begin{equation}
    \sum_c\bigl(u(c)^2+v(c)^2\bigr).
\end{equation}
The one-hot response oracle defined at the start of this subsection is a real permutation matrix. Its realification acts as the same oracle on the real and imaginary copies, so no additional oracle call is introduced.

At readout, assign the two real copies of each original configuration the
same output label. Their squared contributions sum to
\begin{equation}
    u(c)^2+v(c)^2
    =
    |u(c)+iv(c)|^2.
\end{equation}
Thus every complex standard quantum query algorithm has an equivalent real
phase-query circuit with the same number of oracle calls and the same output
distribution.

Thus each model simulates the other without changing the number of oracle calls or the output distribution. This proves Theorem~\ref{thm:standard-quantum-equivalence}.
\end{proof}

\subsection{Adversary lower bound}
\label{app:adversary-lower-bound}

The following is the Boolean-output general adversary lower bound of H{\o}yer, Lee, and {\v S}palek~\cite[Theorem~2]{HoyerLeeSpalek2007}, written for the phase-query model of Section~\ref{sec:quantum-query-circuits}. We include the direct proof to verify the normalization for the separation filters $\Delta_q$.

\begin{theorem}
\label{thm:adversary-query-lower-bound}
For every Boolean function $f:D\to\{0,1\}$ and every
$0\leq\epsilon<1/2$,
\begin{equation}
    Q_{\mathcal Q}^{\epsilon}(f)
    \geq
    \frac{1-2\sqrt{\epsilon(1-\epsilon)}}{2}
    \operatorname{Adv}^{\pm}_{\mathcal Q}(f).
\end{equation}
\end{theorem}

\begin{proof}
If $f$ is constant, then $\operatorname{Adv}^{\pm}_{\mathcal Q}(f)=0$ and
the claim is immediate. If $\mathcal Q$ does not separate $f$, then
Proposition~\ref{prop:quantum-separation-cover} rules out any finite-depth
bounded-error circuit for $\epsilon<1/2$, while the adversary value is
$+\infty$. Assume henceforth that $f$ is nonconstant and that $\mathcal Q$
separates $f$.

Let $\Gamma$ be an optimal real symmetric adversary matrix and put
\begin{equation}
    L:=\|\Gamma\|=\operatorname{Adv}^{\pm}_{\mathcal Q}(f).
\end{equation}
With $D_b=f^{-1}(b)$, write
\begin{equation}
    \Gamma=
    \begin{pmatrix}
        0&G\\
        G^{\mathsf T}&0
    \end{pmatrix}.
\end{equation}
Choose unit left and right singular vectors $u_0,u_1$ of $G$ for the singular
value $L$, with $Gu_1=Lu_0$ and $G^{\mathsf T}u_0=Lu_1$, and set
$v=(u_0,u_1)/\sqrt2$. Then $\Gamma v=Lv$, $\|v\|=1$, and
\begin{equation}
    \sum_{x\in D_0}v(x)^2
    =
    \sum_{x\in D_1}v(x)^2
    =\frac12.
\end{equation}

Fix a depth-$T$ circuit which computes $f$ with error at most $\epsilon$. We use the standard adversary progress function~\cite{BarnumSaksSzegedy2003,HoyerLeeSpalek2007}. Using the correlation kernels $C_t$, define
\begin{equation}
    W_t
    :=
    \sum_{x,y\in D}
        \Gamma(x,y)v(x)v(y)C_t(x,y).
\end{equation}
The initial state is independent of the input, so $C_0(x,y)=1$ and hence
$W_0=v^{\mathsf T}\Gamma v=L$. Input-independent recombinations preserve all
correlations, hence preserve $W_t$.

We first bound the change caused by one phase query. For a configuration
$c\in\Omega$, define $h_c\in\mathbb R^D$ by
\begin{equation}
    h_c(x):=v(x)a_x^t(c),
\end{equation}
and let $S_c$ be the diagonal sign matrix with
\begin{equation}
    S_c(x,x):=\varphi_c(B_{q_c}(x)).
\end{equation}
The contribution of the phase query to the change in progress is
\begin{equation}
    \sum_{c\in\Omega}
        h_c^{\mathsf T}
        (\Gamma-S_c\Gamma S_c)
        h_c.
\end{equation}
The sign $S_c(x,x)$ depends only on the response of query $q_c$. Therefore
entries on pairs not separated by $q_c$ cancel, and
\begin{equation}
    \Gamma-S_c\Gamma S_c
    =
    \Delta_{q_c}\Gamma
    -S_c(\Delta_{q_c}\Gamma)S_c.
\end{equation}
Since $S_c$ is orthogonal and
$\|\Delta_{q_c}\Gamma\|\leq1$, the operator norm of the last difference is
at most $2$. Consequently
\begin{equation}
\begin{aligned}
    |W_t-W_{t+1}|
    &\leq
    2\sum_{c\in\Omega}\|h_c\|^2\\
    &=
    2\sum_{x\in D}v(x)^2
        \sum_{c\in\Omega}a_x^t(c)^2
    =2.
\end{aligned}
\end{equation}
Here the subsequent recombination does not affect the progress quantity.

It remains to bound the final progress. Let $P_0,P_1$ be the coordinate
projections in $\mathbb R^\Omega$ onto configurations labeled by outputs $0$
and $1$. Define
\begin{equation}
    \eta_b
    :=
    \sum_{x\in D_b}v(x)e_x\otimes a_x^T
    \qquad (b\in\{0,1\}).
\end{equation}
Then $\|\eta_0\|^2=\|\eta_1\|^2=1/2$. Put
\begin{equation}
    \delta_0:=2\|(I\otimes P_1)\eta_0\|^2,
    \qquad
    \delta_1:=2\|(I\otimes P_0)\eta_1\|^2.
\end{equation}
Because every input is computed with error at most $\epsilon$,
$0\leq\delta_0,\delta_1\leq\epsilon$.

The block form of $\Gamma$ gives
\begin{equation}
    W_T
    =
    2\langle\eta_0,(G\otimes I)\eta_1\rangle.
\end{equation}
Splitting both vectors into their two output sectors and using $\|G\|=L$
yields
\begin{equation}
    |W_T|
    \leq
    L\left(
        \sqrt{(1-\delta_0)\delta_1}
        +
        \sqrt{\delta_0(1-\delta_1)}
    \right).
\end{equation}
Write $\delta_b=\sin^2\theta_b$ with
$0\leq\theta_b\leq\arcsin\sqrt\epsilon<\pi/4$. The expression in
parentheses is $\sin(\theta_0+\theta_1)$ and is thus at most
\begin{equation}
    2\sqrt{\epsilon(1-\epsilon)}.
\end{equation}
Hence
\begin{equation}
    |W_T|
    \leq
    2\sqrt{\epsilon(1-\epsilon)}\,L.
\end{equation}

Telescoping the per-query changes now gives
\begin{equation}
\begin{aligned}
    2T
    &\geq |W_0-W_T|\\
    &\geq
    \bigl(1-2\sqrt{\epsilon(1-\epsilon)}\bigr)L.
\end{aligned}
\end{equation}
Taking the minimum over all depth-$T$ circuits that $\epsilon$-compute $f$ gives the stated H{\o}yer--Lee--{\v S}palek lower bound in the present phase-query model.
\end{proof}

\end{document}